\documentclass[lettersize,journal]{IEEEtran}

\usepackage{array}
\usepackage{stfloats}
\usepackage{verbatim}
\usepackage{cite}
\usepackage{amsmath,amssymb,amsfonts}
\usepackage{amssymb}
\usepackage{graphicx}
\usepackage{algorithm}
\usepackage{algpseudocode}
\usepackage{textcomp}
\usepackage{paralist}
\usepackage{eurosym}
\usepackage{flushend}
\usepackage{url}
\usepackage[utf8]{inputenc}
\usepackage{soul}
\usepackage{caption}
\usepackage[list=true]{subcaption}
\usepackage{booktabs}
\usepackage{tabu}
\usepackage{wrapfig}
\usepackage{arydshln}
\usepackage[normalem]{ulem}
\usepackage{multirow}
\usepackage{hyperref}
\usepackage{siunitx}
\usepackage{diagbox}
\usepackage{tikz}
\usetikzlibrary{arrows.meta,positioning,calc,fit,backgrounds,shapes.geometric,shadows}
\usepackage[table]{xcolor}
\usepackage{colortbl}
\usepackage[capitalize]{cleveref}
\usepackage[T1]{fontenc}
\usepackage{nicefrac}
\usepackage{microtype}
\usepackage{adjustbox}
\usepackage{tabularx}
\usepackage{fontawesome5}

\newcommand{\rotHeader}[1]{\rotatebox[origin=lb]{90}{\strut #1}}

\newtheorem{proposition}{Proposition}

\newtheorem{proof}{Proof}

\newcolumntype{P}[1]{>{\raggedright\arraybackslash}p{#1}}

\begin{document}

\title{Permutation-Invariant Time--Frequency Segmentation and Parameter Estimation of Multiple Overlapping Interference Signals}

\def\authorrefmark#1{\ensuremath{^{#1}}}
\author{\IEEEauthorblockN{
Lucas Heublein\authorrefmark{1},
Christian Wielenberg\authorrefmark{1,}\authorrefmark{2},
Christopher Mutschler\authorrefmark{1,}\authorrefmark{2}},
Felix Ott\authorrefmark{1}\\
\IEEEauthorblockA{\authorrefmark{1}Fraunhofer Institute for Integrated Circuits IIS, 90411 Nürnberg, Germany}\\
\IEEEauthorblockA{\authorrefmark{2}Machine Learning and Positioning Systems Lab, University of Technology Nürnberg (UTN), 90461 Nürnberg}\\
{\tt\small\{lucas.heublein, christian.wielenberg, christopher.mutschler, \scalebox{0.8}{\faEnvelope[regular]}\,felix.ott\}@iis.fraunhofer.de}
\thanks{Manuscript received October 02, 2026; revised XX XX, 2026.}
}

\markboth{UNDER REVIEW, October~2026}%
{Heublein \MakeLowercase{\textit{et al.}}: Segmentation and Parameter Estimation of Multiple Interference Signals}

\IEEEpubid{0000--0000/00\$00.00~\copyright~2021 IEEE}

\maketitle

\begin{abstract}
Radio-frequency interference threatens satellite navigation, wireless communications, radar, and industrial sensing. Reliable operation therefore requires spectrum situational awareness to detect interference and determine the number, time--frequency occupancy, and characteristics of active sources. In realistic environments, multiple sources may overlap, while multipath propagation and unequal received powers can merge their signatures or cause strong emissions to mask weaker ones. Source-specific segmentation and characterization are therefore particularly challenging. In this paper, we propose PI-SCAN (\underline{P}ermutation-\underline{I}nvariant \underline{S}egmentation, \underline{C}haracterization, and \underline{A}ssignment \underline{N}etwork), a joint set-based learning framework for source-specific segmentation and characterization of overlapping interference signals. A logarithmic power spectrogram derived from complex IQ samples is mapped to six unordered prediction slots. For each slot, the model estimates objectness, interference type, center frequency, bandwidth, received signal strength, binary offset carrier (BOC) modulation, and a source-specific segmentation mask. Permutation-invariant Hungarian assignment aligns the predicted slots with the unordered ground-truth sources, while a weighted multi-task objective jointly optimizes the individual prediction tasks. The framework is evaluated using a ray-traced industrial environment containing mixtures of one to six simultaneously active interferers drawn from 94 waveform configurations across six interference classes. The full seven-task configuration achieves a detection F1-score of $74.8\%$, type and BOC accuracies of $97.3\%$ and $98.4\%$, and mean absolute errors of $1.29\,\mathrm{MHz}$, $2.54\,\mathrm{MHz}$, and $3.19\,\mathrm{dB}$ for center frequency, bandwidth, and signal strength. Using only objectness, type, and mask prediction increases the F1-score to $91.1\%$, highlighting the trade-off between detection and comprehensive characterization.
\end{abstract}


\section{Introduction}
\label{label_introduction}

Radio-frequency (RF) interference affects applications in which weak signals must be recovered in increasingly congested electromagnetic environments. GNSS receivers encounter intentional jammers and unintentional emissions, wireless systems accommodate several coexisting technologies in shared bands, radar systems face compound jamming, and radio telescopes observe diverse terrestrial and satellite interference \cite{borio_impact_detection,fonseca_rat_detection,sarkar_radyololet,akeret_rfi_unet}. In realistic scenarios, multiple independent sources may be active simultaneously and overlap in time and frequency, allowing strong emissions to obscure weaker sources and complicating their individual analysis \cite{liu_composite_jamming,heublein_vidal_jammertest}.

Reliable spectrum situational awareness therefore requires more than determining whether interference is present. A monitoring system should estimate the number of active sources, delineate the time--frequency region occupied by each source, classify its interference type, and characterize relevant physical properties such as center frequency, bandwidth, and received signal strength. Maintaining the association between each segmented source and its predicted characteristics is particularly important when several emissions overlap. In addition to supporting interference assessment and mitigation, these source-resolved outputs can be combined with measurements from antenna arrays or distributed monitoring stations for direction-of-arrival estimation and emitter localization \cite{borio_spectral_spatial,wielenberg_sicdf}. Real-world multi-emitter measurements further demonstrate the importance of jointly analyzing waveform type, spectral occupancy, received power, and receiver impact under realistic propagation conditions \cite{heublein_vidal_jammertest,heublein_benschuh}.

Early GNSS monitoring methods detected interference from receiver observables such as automatic gain control, \(C/N_0\), and correlator statistics, or characterized individual emissions through spectral analysis \cite{balaei_cw_interference,axell_jamming_detection,borio_impact_detection,mitch_jammer_characteristics}. Recent learning-based approaches classify time--frequency representations using support-vector machines, convolutional networks, or multi-stage architectures \cite{morales_jammer_classification,elango_disruptive_gnss,mehr_dual_stage}. Beyond GNSS, semantic-segmentation networks identify interference-contaminated time--frequency bins, whereas object detectors localize signals using rectangular bounding boxes and assign technology or modulation labels \cite{akeret_rfi_unet,fonseca_rat_detection,wood_signal_identification,boegner_wideband_detection}. Recent multilabel, transformer-based set-prediction, and YOLO-based approaches additionally address compound or simultaneously active signals \cite{guo_multilabel_jamming,vagollari_transformer_rf,qu_hifi_yolo}. Nevertheless, these methods generally provide recording-level labels, a shared foreground mask, or rectangular signal extents and do not jointly produce source-specific masks together with comprehensive physical characteristics. Consequently, the joint segmentation and characterization of an unordered set of overlapping interference sources remains insufficiently addressed.

\IEEEpubidadjcol

\textbf{Contributions.} We summarize our contributions as follows. (i) First, we develop a simulated dataset of realistic multi-source interference scenarios using the Sionna ray-tracing framework~\cite{hoydis_cammerer}. (ii) Second, we propose a unified set-based framework that jointly detects, segments, and characterizes up to six simultaneously active and overlapping interference sources from their spectrogram representation. (iii) Third, we introduce permutation-invariant source assignment to preserve the association between each source-specific mask and its predicted type, center frequency, signal strength, bandwidth, and BOC property. (iv) Fourth, we investigate MobileNetV2--U-Net~\cite{sandler_mobilenetv2,ronneberger_unet} and attention-based~\cite{vaswani_transformer} architectures under a common prediction interface and systematically evaluate alternative prediction objectives and loss functions. (v) Finally, we demonstrate the method's ability to resolve overlapping sources and quantify the trade-off between interference detection performance and comprehensive multi-parameter characterization.

\section{Related Work}
\label{label_related_work}

GNSS interference monitoring generally comprises detection, classification, and characterization. Traditional detection methods operate either in the receiver domain, where changes in automatic gain control (AGC), carrier-to-noise-density ratio \(C/N_0\), and correlator statistics indicate anomalous received power or tracking degradation, or in the signal domain, where power-spectral-density and time--frequency analyses expose structured emissions \cite{balaei_cw_interference,axell_jamming_detection,borio_impact_detection}. Complementary characterization studies analyze properties such as center frequency, occupied bandwidth, received power, temporal modulation, and spatial behavior \cite{mitch_jammer_characteristics,borio_spectral_spatial}. More recent approaches formulate interference monitoring as a learning problem: support-vector machines and convolutional neural networks classify spectrograms into predefined jammer types, feature-based methods categorize events recorded by distributed monitoring networks, and transfer-learning or dual-stage architectures combine interference detection with the estimation of its type or power \cite{morales_jammer_classification,diez_rfi_classification,elango_disruptive_gnss,mehr_dual_stage}. Initial work on composite interference further considers the simultaneous presence of multiple jammer types, but primarily performs recording-level detection and classification rather than source-resolved segmentation and characterization \cite{liu_composite_jamming}. Related work extends these directions through simulated data for dynamic interference characterization and direction finding \cite{wielenberg_sicdf}, real-world multi-source measurements and characterization \cite{heublein_vidal_jammertest}, receiver-aware interference-impact analysis using the spectral separation coefficient \cite{heublein_benschuh}, and hardware-aware model optimization for embedded inference \cite{heublein_deutel}. Nevertheless, jointly localizing, separating, and characterizing an unknown set of overlapping interference sources within a single spectrogram remains insufficiently addressed, motivating the set-based approach.

Recent machine-learning approaches increasingly treat time--frequency representations as structured images instead of assigning a single label to an entire recording. Pixel-wise semantic segmentation has been applied to RF-interference flagging using U-Net architectures, enabling individual time--frequency bins to be distinguished from uncontaminated regions \cite{akeret_rfi_unet}. Object-detection approaches subsequently introduced instance-level localization: Fonseca et al.~\cite{fonseca_rat_detection} used spectrogram bounding boxes to identify radio-access technologies and derive center frequency, bandwidth, and temporal parameters, while Wood et al.~\cite{wood_signal_identification} evaluated RetinaNet and YOLOv5 for detecting signals with overlapping frequency ranges. Large-scale studies have since compared semantic-segmentation and object-detection networks on complex-valued spectrograms containing numerous signal classes \cite{boegner_wideband_detection}. Related radar-spectrum methods combine YOLO-based localization with wavelet-domain processing to improve low-power detection and parameter estimation \cite{sarkar_radyololet}, whereas recent analyses demonstrate that STFT resolution, windowing, and overlap substantially influence the performance of YOLO- and DETR-based detectors \cite{kang_optimal_preprocessing}. Image segmentation has also been used to isolate compound radar-jamming signatures before attention-based classification, although the resulting segmentation is not source-specific \cite{li_compound_segmentation}. For our experimental comparison, we consider three recent approaches representing complementary formulations: MLCJC-CNN~\cite{guo_multilabel_jamming} performs recording-level multilabel recognition of the constituents of compound radar jamming but does not localize or characterize individual sources; the transformer of Vagollari et al.~\cite{vagollari_transformer_rf} directly processes frequency-domain IQ samples and generates a set of localized technology-class predictions, but does not produce dense time--frequency masks or the considered GNSS interference parameters; and HIFI-YOLO~\cite{qu_hifi_yolo} jointly localizes and recognizes overlapping modulation signals using spectrogram bounding boxes, which provide only rectangular signal extents and modulation labels. In contrast, our method represents every interference source as a permutation-invariant set element and jointly predicts its dense source-specific mask, objectness, type, center frequency, signal strength, bandwidth, and BOC property. It therefore supports instance-level segmentation and physical characterization of up to six overlapping GNSS interference sources, including sources that occupy the same time--frequency bins or exhibit non-rectangular spectral structures.
\section{Methodology of PI-SCAN}
\label{label_method}

In Sec.~\ref{label_method_problem}, we formulate the joint segmentation and characterization problem, while Sec.~\ref{label_method_overview} summarizes the complete processing pipeline. Sec.~\ref{label_method_preprocessing} and~\ref{label_method_architectures} describe the input preparation and evaluated network architectures. Sec.~\ref{label_method_model_tasks} introduces the set-based prediction tasks and their loss functions, followed by the permutation-invariant assignment in Sec.~\ref{label_method_permutation} and the training objective in Sec.~\ref{label_method_training_objective}.

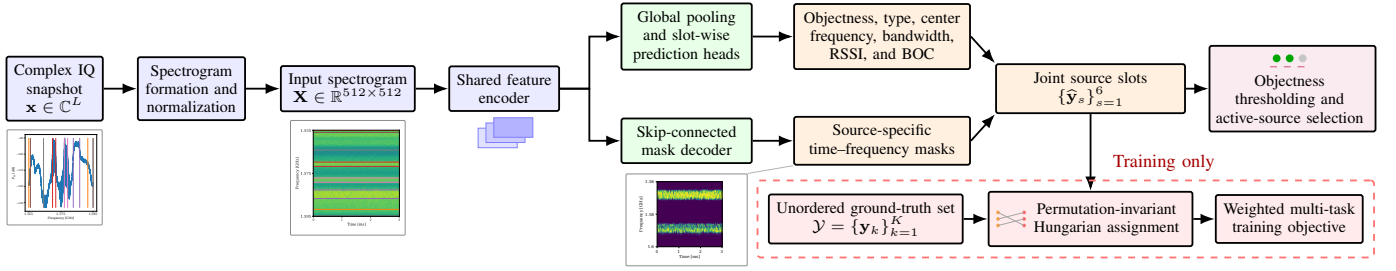
\begin{figure*}[!t]
    \centering
    \resizebox{\textwidth}{!}{%
    \begin{tikzpicture}[
        block/.style={
            draw,
            thick,
            rounded corners=1pt,
            align=center,
            minimum height=6mm,
            minimum width=10mm,
            fill=blue!8,
            font=\scriptsize
        },
        branch/.style={
            draw,
            thick,
            rounded corners=1pt,
            align=center,
            minimum height=6mm,
            minimum width=20mm,
            fill=green!9,
            font=\scriptsize
        },
        output/.style={
            draw,
            thick,
            rounded corners=1pt,
            align=center,
            minimum height=6mm,
            minimum width=27mm,
            fill=orange!12,
            font=\scriptsize
        },
        training/.style={
            draw,
            thick,
            rounded corners=1pt,
            align=center,
            minimum height=6mm,
            minimum width=10mm,
            fill=red!8,
            font=\scriptsize
        },
        inference/.style={
            draw,
            thick,
            rounded corners=1pt,
            align=center,
            minimum height=6mm,
            minimum width=10mm,
            fill=purple!8,
            font=\scriptsize
        },
        thumbnail/.style={
            draw=black!45,
            thin,
            rounded corners=0.5pt,
            fill=white,
            inner sep=0.3mm
        },
        arrow/.style={
            -{Latex[length=2.2mm]},
            thick
        }
    ]

    \node[block] (iq) at (0,0)
        {Complex IQ\\snapshot\\
         $\mathbf{x}\in\mathbb{C}^{L}$};

    \node[block] (preprocess) at (2.1,0)
        {Spectrogram\\formation
         and\\normalization};

    \node[block] (spectrogram) at (4.5,0)
        {Input spectrogram\\
         $\mathbf{X}\in\mathbb{R}^{512\times512}$};

    \node[block] (encoder) at (6.9,0)
        {Shared feature\\
         encoder};

    \begin{scope}[
        shift={([yshift=-1.5mm]encoder.south)},
        x=1mm,
        y=1mm
    ]
        \draw[fill=blue!12, draw=blue!50]
            (-4,-4) rectangle (2,-1);
        \draw[fill=blue!20, draw=blue!50]
            (-2.8,-3.2) rectangle (3.2,-0.2);
        \draw[fill=blue!30, draw=blue!50]
            (-1.6,-2.4) rectangle (4.4,0.6);
    \end{scope}

    \node[branch] (heads) at (9.7,0.8)
        {Global pooling\\and
         slot-wise\\prediction heads};

    \node[branch] (decoder) at (9.7,-0.8)
        {Skip-connected\\
         mask decoder};

    \node[output] (attributes) at (12.7,0.8)
        {Objectness, type, center\\frequency,
         bandwidth,\\RSSI, and BOC};

    \node[output] (masks) at (12.7,-0.8)
        {Source-specific\\
         time--frequency masks};

    \node[output, minimum width=29mm] (slots) at (15.9,0)
        {Joint source slots\\
         $\{\widehat{\mathbf{y}}_s\}_{s=1}^{6}$};

    \node[
        inference,
        inner xsep=1mm,
        inner ysep=0.5mm
    ] (inference) at (19.0,0) {
        \begin{tabular}{@{}c@{}}
            \tikz[baseline=-0.5ex]{
                \fill[green!65!black] (0,0) circle (0.6mm);
                \fill[green!65!black] (2mm,0) circle (0.6mm);
                \fill[gray!45] (4mm,0) circle (0.6mm);
                \draw[purple!70,dashed]
                    (-1mm,-1.3mm) -- (5mm,-1.3mm);
            }\\[0.9mm]
            Objectness\\thresholding and\\active-source selection
        \end{tabular}
    };

    \node[training] (groundtruth) at (12.45,-2.0)
        {Unordered ground-truth set\\
         $\mathcal{Y}=\{\mathbf{y}_k\}_{k=1}^{K}$};

    \node[training] (matching) at (15.9,-2.0) {
        \begin{tikzpicture}[
            x=1.0mm,
            y=1.0mm,
            baseline=(current bounding box.center)
        ]
            \foreach \y in {-1,0,1}{
                \fill[orange!75] (0,\y) circle (0.28);
                \fill[red!60]    (4,\y) circle (0.28);
            }
            \draw[gray!70] (0,-1) -- (4,0);
            \draw[gray!70] (0,0)  -- (4,1);
            \draw[gray!70] (0,1)  -- (4,-1);
            \node[
                anchor=west,
                align=center,
                font=\scriptsize,
                inner sep=0pt,
                outer sep=0pt
            ] at (5.5,0)
            {Permutation-invariant\\Hungarian assignment};
        \end{tikzpicture}
    };

    \node[training] (loss) at (18.95,-2.0)
        {Weighted multi-task\\
         training objective};

    \node[thumbnail, below=1.2mm of iq] (iqfigure) {
        \includegraphics[
            width=14mm
        ]{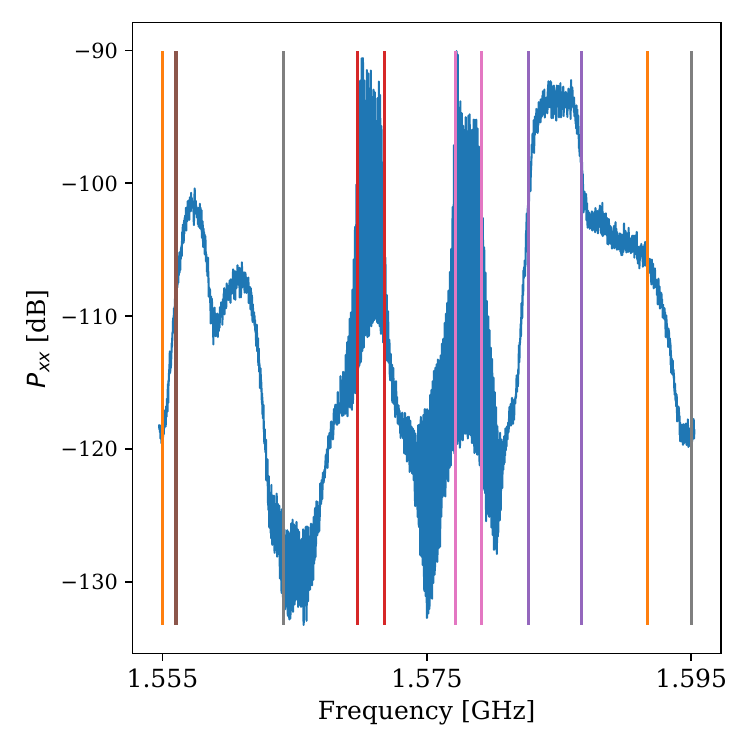}
    };
    
    \node[thumbnail, below=1.2mm of spectrogram] (specfigure) {
        \includegraphics[
            width=17mm
        ]{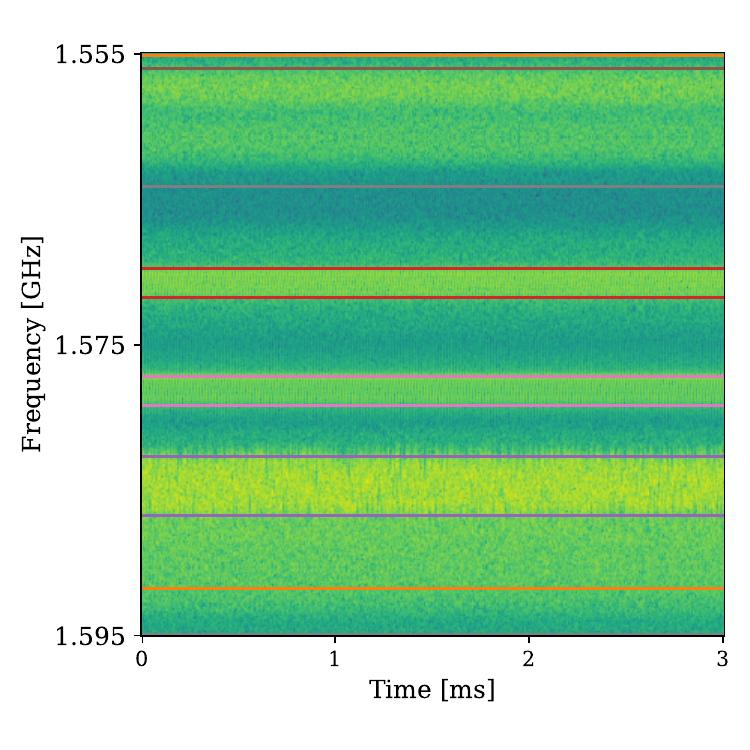}
    };
    
    \node[thumbnail, below=1.2mm of masks, xshift=-3cm] (maskfigure) {
        \includegraphics[
            width=18mm,
            keepaspectratio
        ]{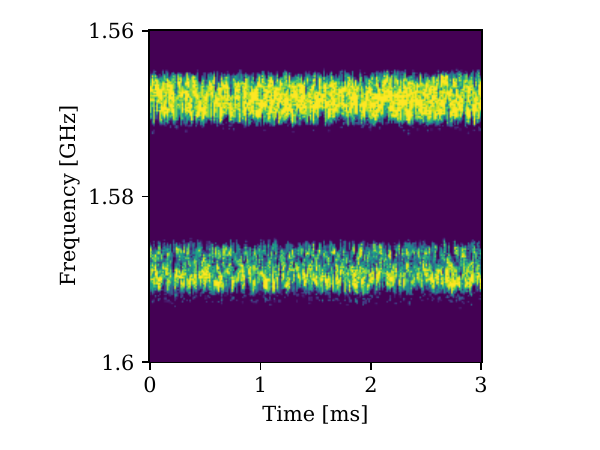}
    };

    \draw[arrow] (iq) -- (preprocess);
    \draw[arrow] (preprocess) -- (spectrogram);
    \draw[arrow] (spectrogram) -- (encoder);

    \draw[arrow] (encoder.east) -- ++(0.45,0) |- (heads.west);
    \draw[arrow] (encoder.east) -- ++(0.45,0) |- (decoder.west);

    \draw[arrow] (heads) -- (attributes);
    \draw[arrow] (decoder) -- (masks);

    \draw[arrow] (attributes.east) -- (slots.north west);
    \draw[arrow] (masks.east) -- (slots.south west);

    \draw[arrow] (slots.east) -- (inference.west);

    \draw[arrow] (slots.south) -- (matching.north);
    \draw[arrow] (groundtruth.east) -- (matching.west);
    \draw[arrow] (matching) -- (loss);

    \node[
        draw=red!55,
        dashed,
        line width=0.8pt,
        rounded corners=1.5pt,
        fit=(groundtruth)(matching)(loss),
        inner sep=1.6mm,
        label={[font=\footnotesize,text=red!65!black]above right:Training only}
    ] {};

    \draw[
        gray!60,
        thin
    ] (maskfigure.north east) -- (masks.south west);

    \end{tikzpicture}%
    }
    \caption{\textbf{Overview of PI-SCAN -- \underline{p}ermutation-\underline{i}nvariant \underline{s}egmentation, \underline{c}haracterization, and \underline{a}ssignment \underline{n}etwork.} A normalized time--frequency representation is
    processed by a shared encoder and parallel characterization and
    segmentation branches. The resulting outputs are combined into six
    source-specific prediction slots. During training, these slots are assigned
    to the unordered ground-truth sources using Hungarian matching before the
    weighted multi-task objective is evaluated. During inference, active
    sources are selected according to their predicted objectness scores.}
    \label{fig:method_overview}
\end{figure*}

\subsection{Problem Formulation and Input Representation}
\label{label_method_problem}

Let
\begin{equation}
    \mathbf{x}
    =
    \left[x_0,\ldots,x_{L-1}\right]^{\mathsf{T}}
    \in \mathbb{C}^{L}
\end{equation}
denote a received complex baseband snapshot containing a superposition of multiple interference signals and the corresponding channel effects. Each snapshot comprises $L$ IQ samples recorded at a sampling rate of $f_{\mathrm{s}}$. The IQ sequence is transformed into a single-channel, logarithmic power spectrogram
\begin{equation}
    \mathbf{X}\in\mathbb{R}^{N_f\times N_t},
    \qquad N_f=N_t=512,
\end{equation}
where the rows and columns represent frequency bins and time frames, respectively. The construction and normalization of this representation are described in Section~\ref{label_method_preprocessing}. The spectrogram contains the superimposed time--frequency signatures of all simultaneously active interference sources and constitutes the input to the proposed model. For each snapshot, the ground truth is represented as an unordered set $\mathcal{Y} = \left\{ \mathbf{y}_k \right\}_{k=1}^{K}$, with $1\leq K\leq 6$,
where $K$ denotes the number of active interference sources. Each source annotation is defined as
\begin{equation}
\label{eq:source_annotation}
    \mathbf{y}_k
    =
    \left(
        c_k,
        f_k,
        b_k,
        q_k,
        r_k,
        \mathbf{M}_k
    \right).
\end{equation}
Here, $c_k\in\{1,\ldots,C\}$ denotes the interference type, with $C=6$ classes corresponding to noise, chirp, frequency-hopping, modulated, multitone, and pulsed interference. Furthermore, $f_k$ and $b_k$ denote the center-frequency offset and occupied bandwidth in MHz, respectively, $q_k\in\{0,1\}$ indicates whether the waveform employs binary offset carrier (BOC) modulation~\cite{betz_boc}, and $r_k$ denotes its received signal strength indicator (RSSI) in dBm. The binary mask
\begin{equation}
    \mathbf{M}_k
    \in
    \{0,1\}^{N_f\times N_t}
\end{equation}
specifies the time--frequency support of the $k$-th interference source. In contrast to conventional semantic segmentation, the source masks are not required to be mutually exclusive. Consequently, two sources $i\neq j$ may occupy the same time--frequency bins, such that
\begin{equation}
    \left\|
        \mathbf{M}_i\odot\mathbf{M}_j
    \right\|_0 > 0,
\end{equation}
where $\odot$ denotes element-wise multiplication. This multi-label mask representation is essential for preserving individual source annotations in regions containing overlapping interference.

Given only the composite spectrogram $\mathbf{X}$, the objective is to jointly estimate the number of active sources and their individual characterization parameters and segmentation masks. Following the direct set-prediction paradigm~\cite{carion_detr}, the model produces a fixed set of $S=6$ candidate predictions, $\widehat{\mathcal{Y}} = \left\{ \widehat{\mathbf{y}}_s \right\}_{s=1}^{S}$, where each output slot is represented by
\begin{equation}
\label{eq:slot_prediction}
    \widehat{\mathbf{y}}_s
    =
    \left(
        \widehat{o}_s,
        \widehat{\mathbf{p}}_s,
        \widehat{f}_s,
        \widehat{b}_s,
        \widehat{q}_s,
        \widehat{r}_s,
        \widehat{\mathbf{M}}_s
    \right).
\end{equation}
The objectness score $\widehat{o}_s$ indicates whether slot $s$ represents an active source, while $\widehat{\mathbf{p}}_s\in[0,1]^C$ contains the interference-type probabilities. The remaining quantities estimate the center frequency, bandwidth, BOC indicator, RSSI, and source-specific time--frequency mask. Since neither the ground-truth sources nor the prediction slots possess an intrinsic ordering, the learning problem is permutation invariant. A one-to-one assignment between active sources and prediction slots is therefore determined during training, as described in Section~\ref{label_method_permutation}.

\subsection{Method Overview}
\label{label_method_overview}

Figure~\ref{fig:method_overview} summarizes the proposed end-to-end
methodology. The received IQ snapshot is first transformed into a normalized
time--frequency representation, as detailed in
Section~\ref{label_method_preprocessing}. A shared feature encoder subsequently
provides representations to two parallel branches: global prediction heads
estimate the presence, type, center frequency, bandwidth, RSSI, and BOC
property of each interference source, while a skip-connected decoder estimates
the corresponding source-specific segmentation masks. These outputs are
combined into a fixed set of six joint source predictions, as described in
Section~\ref{label_method_model_tasks}. During training, the unordered
predictions and ground-truth sources are associated through the
permutation-invariant assignment introduced in
Section~\ref{label_method_permutation}, after which the complete model is
optimized using the multi-task objective presented in
Section~\ref{label_method_training_objective}. At inference time, the
objectness scores determine which prediction slots correspond to active
interference sources.

\subsection{Pre-Processing}
\label{label_method_preprocessing}

Each complex IQ snapshot is transformed into a two-dimensional power spectral density (PSD) representation using the short-time Fourier transform (STFT)~\cite{allen_rabiner_stft}. A Blackman window~\cite{harris_windows} of $N_{\mathrm{FFT}}=512$ samples and an overlap of $256$ samples are employed, corresponding to a hop size of $256$ samples and an overlap ratio of $50\%$. No detrending is applied, and the complete two-sided spectrum is retained. At the sampling rate of $40.5\,\mathrm{MHz}$, this configuration provides a frequency-bin spacing of approximately $79.1\,\mathrm{kHz}$, a window duration of $12.64\,\mu\mathrm{s}$, and a temporal hop of $6.32\,\mu\mathrm{s}$. For an input length of $131{,}328$ IQ samples, the resulting spectrogram contains $512$ frequency bins and $512$ time frames.

The frequency axis is shifted such that zero frequency is located at the center of the representation. The PSD values are subsequently converted to the logarithmic domain according to
\begin{equation}
    \mathbf{S}_{\mathrm{dB}}
    =
    10\log_{10}\!\left(\mathbf{S}_{\mathrm{PSD}}\right).
\end{equation}
To reduce variations in the absolute scale between different snapshots, each spectrogram is standardized independently:
\begin{equation}
    \label{eq:spectrogram_normalization}
    \widetilde{\mathbf{X}}
    =
    \frac{
        \mathbf{S}_{\mathrm{dB}}-\mu_{\mathbf{S}}
    }{
        \sigma_{\mathbf{S}}
    },
\end{equation}
where $\mu_{\mathbf{S}}$ and $\sigma_{\mathbf{S}}$ denote the mean and standard deviation of the respective log-PSD spectrogram. The resulting normalized representation $\widetilde{\mathbf{X}}\in\mathbb{R}^{512\times512}$ is provided to the network as a single input channel. For model configurations requiring multiple input channels, the same normalized spectrogram is replicated across the corresponding channels.

\subsection{Network Architectures}
\label{label_method_architectures}

We evaluate two network architectures that share the same set-based output representation and task-specific prediction interface. In both cases, the normalized spectrogram is processed by a shared encoder, followed by parallel characterization and segmentation branches. The characterization branch predicts objectness, interference type, center frequency, signal strength, bandwidth and its uncertainty, and the BOC property for each of the $S=6$ source slots. In parallel, the segmentation branch estimates a source-specific time--frequency mask for every slot. Keeping the output representation and training formulation fixed enables a controlled comparison of the underlying feature-extraction and refinement mechanisms.

\paragraph{MobileNetV2--U-Net}
The first architecture combines a MobileNetV2 encoder~\cite{sandler_mobilenetv2}, initialized using weights pretrained on ImageNet~\cite{deng_imagenet}, with a lightweight U-Net-style decoder~\cite{ronneberger_unet}. The input convolution is adapted to the number of spectrogram channels. Five encoder stages produce feature maps with $16$, $24$, $32$, $96$, and $1{,}280$ channels at downsampling factors of $2$, $4$, $8$, $16$, and $32$, respectively. The bottleneck is globally averaged and processed using dropout with $p=0.2$, followed by separate linear heads for objectness, interference type, center frequency, signal strength, bandwidth, and BOC detection. The bandwidth mean is constrained using a softplus activation, while its log-variance is restricted to $[-6,6]$. The segmentation decoder uses four encoder skip connections and channel widths of $128$, $64$, $32$, and $16$. Each stage bilinearly upsamples the decoder features, concatenates them with the corresponding encoder features, and applies pointwise and depthwise convolutions with batch normalization and ReLU6 activation. A final upsampling and projection produce one full-resolution mask-logit map per source slot.

\paragraph{Attention-Based Variant}
The second architecture retains the same prediction heads, slot representation, and mask resolution but introduces attention-based feature refinement~\cite{woo_cbam,oktay_attention_unet,vaswani_transformer}. This mechanism emphasizes interference-relevant features and suppresses less informative background regions. Because the prediction interface and training procedure remain unchanged, the influence of attention can be evaluated independently of the remaining methodological components.

\subsection{Joint Set-Based Segmentation and Characterization}
\label{label_method_model_tasks}

Let
\begin{equation}
    \mathcal{A}
    =
    \left\{
        (s,k)
    \right\}
\end{equation}
denote the one-to-one assignment between prediction slot $s$ and ground-truth source $k$. The computation of this assignment is described in Section~\ref{label_method_permutation}. All characterization and segmentation losses are evaluated only for matched pairs $(s,k)\in\mathcal{A}$, while the objectness loss is evaluated for all prediction slots. Table~\ref{tab:method_loss_search_space} summarizes the evaluated loss functions. The binary cross-entropy between a probability $p$ and a target $y\in\{0,1\}$ is defined as
\begin{equation}
\label{eq:bce_definition}
    \ell_{\mathrm{BCE}}(p,y)
    =
    -y\log(p)
    -(1-y)\log(1-p).
\end{equation}

\paragraph{Objectness}
The objectness task determines whether a prediction slot represents an active interference source. The binary target for slot $s$ is defined as
\begin{equation}
    o_s =
    \begin{cases}
        1, & \text{if } (s,k)\in\mathcal{A}
        \text{ for some } k,\\
        0, & \text{otherwise}.
    \end{cases}
\end{equation}
Let $p_s^{\mathrm{O}}=\sigma(\widehat{o}_s)$ denote the predicted objectness probability and
\begin{equation}
    p_{t,s}
    =
    o_s p_s^{\mathrm{O}}
    +(1-o_s)(1-p_s^{\mathrm{O}}).
\end{equation}
A focal-loss variant based on~\cite{lin_focal} is employed to reduce the contribution of easily classified slots:
\begin{equation}
\label{eq:objectness_focal}
    \ell_{\mathrm{Focal}}(p_s^{\mathrm{O}},o_s)
    =
    -\alpha
    \left(1-p_{t,s}\right)^{\gamma}
    \log(p_{t,s}),
\end{equation}
with $\alpha=0.25$ and $\gamma=2$. To account for the different roles of active and inactive slots, their losses are averaged separately:
\begin{equation}
\label{eq:objectness_loss}
    \mathcal{L}_{\mathrm{O}}
    =
    w_{+}
    \underset{s:o_s=1}{\operatorname{mean}}
    \ell_{\mathrm{Focal}}
    +
    w_{-}
    \underset{s:o_s=0}{\operatorname{mean}}
    \ell_{\mathrm{Focal}},
\end{equation}
where $w_{+}=1$ and $w_{-}=1.5$; an empty term is omitted.

\paragraph{Type Classification}
For each matched source, the model predicts a probability distribution
$\widehat{\mathbf{p}}_s$ over the $C=6$ interference types. The corresponding multiclass cross-entropy loss is
\begin{equation}
\label{eq:type_loss}
    \mathcal{L}_{\mathrm{T}}
    =
    -\frac{1}{K}
    \sum_{(s,k)\in\mathcal{A}}
    \log\left(
        \widehat{p}_{s,c_k}
    \right),
\end{equation}
where $c_k$ denotes the ground-truth type of source $k$.

\paragraph{Center Frequency}
The center-frequency offset is estimated as a continuous value in MHz. Its prediction error is
\begin{equation}
    e^{\mathrm{CF}}_{s,k}
    =
    \widehat{f}_s-f_k.
\end{equation}
To reduce the influence of large regression errors, the Huber
function~\cite{huber_robust}
\begin{equation}
    h_{\delta}(e)
    =
    \begin{cases}
        \frac{1}{2}e^2,
        & |e|\leq\delta,\\[1mm]
        \delta\left(|e|-\frac{1}{2}\delta\right),
        & |e|>\delta
    \end{cases}
\end{equation}
is used with $\delta=5$. The center-frequency loss is therefore
\begin{equation}
\label{eq:center_frequency_loss}
    \mathcal{L}_{\mathrm{CF}}
    =
    \frac{1}{K}
    \sum_{(s,k)\in\mathcal{A}}
    h_{\delta}\!\left(
        e^{\mathrm{CF}}_{s,k}
    \right).
\end{equation}

\paragraph{Signal Strength}
The received signal strength is likewise formulated as a regression task. With
\begin{equation}
    e^{\mathrm{SS}}_{s,k}
    =
    \widehat{r}_s-r_k,
\end{equation}
the signal-strength loss is
\begin{equation}
\label{eq:signal_strength_loss}
    \mathcal{L}_{\mathrm{SS}}
    =
    \frac{1}{K}
    \sum_{(s,k)\in\mathcal{A}}
    h_{\delta}\!\left(
        e^{\mathrm{SS}}_{s,k}
    \right),
    \qquad \delta=5.
\end{equation}

\paragraph{Bandwidth}
For the bandwidth task, each slot predicts a positive mean
$\widehat{b}_s$ and a corresponding log-variance
$\widehat{v}_s\in[-6,6]$. The latter permits the model to represent
input-dependent uncertainty. The bandwidth loss is given by the
heteroscedastic Gaussian negative log-likelihood~\cite{kendall_gal_uncertainty}
\begin{equation}
\label{eq:bandwidth_loss}
    \mathcal{L}_{\mathrm{BW}}
    =
    \frac{1}{K}
    \sum_{(s,k)\in\mathcal{A}}
    \left[
        \frac{1}{2}
        \exp\!\left(-\widehat{v}_s\right)
        \left(\widehat{b}_s-b_k\right)^2
        +
        \frac{1}{2}\widehat{v}_s
    \right].
\end{equation}
The bandwidth mean is constrained to be positive using a softplus activation.

\paragraph{BOC}
The presence of binary offset carrier modulation is treated as a binary classification task. With
$p_s^{\mathrm{BOC}}=\sigma(\widehat{q}_s)$, the BOC loss is
\begin{equation}
\label{eq:boc_loss}
    \mathcal{L}_{\mathrm{BOC}}
    =
    \frac{1}{K}
    \sum_{(s,k)\in\mathcal{A}}
    \ell_{\mathrm{BCE}}
    \left(
        p_s^{\mathrm{BOC}},q_k
    \right).
\end{equation}

\paragraph{Mask}
For each matched source, the segmentation branch produces a mask probability
$\widehat{M}_{s,u,v}$ for every time--frequency bin $(u,v)$. The mask loss is the pixel-wise binary cross-entropy
\begin{equation}
\label{eq:mask_loss}
    \mathcal{L}_{\mathrm{M}}
    =
    \frac{1}{K N_f N_t}
    \sum_{(s,k)\in\mathcal{A}}
    \sum_{u=1}^{N_f}
    \sum_{v=1}^{N_t}
    \ell_{\mathrm{BCE}}
    \left(
        \widehat{M}_{s,u,v},
        M_{k,u,v}
    \right).
\end{equation}
Because each source is assigned an independent mask, multiple masks may be active at the same time--frequency position, allowing the model to preserve overlapping interference components.

\subsection{Permutation-Invariant Source Assignment}
\label{label_method_permutation}

\begin{table*}[!t]
\setlength{\tabcolsep}{2.5pt}
\centering
\footnotesize
\caption{Loss-function search space for the individual prediction
components. The reference losses define the baseline configuration in
Table~\ref{tab:loss_functions}. Alternative losses are first evaluated
individually, while selected cross-component combinations are reported in
Table~\ref{tab:objective_loss_combinations}.}
\label{tab:method_loss_search_space}
\begin{tabularx}{\textwidth}{
    P{0.12\textwidth}
    P{0.20\textwidth}
    P{0.26\textwidth}
    P{0.39\textwidth}
}
\toprule
\textbf{Component} &
\textbf{Reference Loss} &
\textbf{Evaluated Alternatives} &
\textbf{Loss Parameters} \\
\midrule
Objectness &
Focal~\cite{lin_focal} &
Varifocal~\cite{zhang_varifocalnet}; quality focal loss (QFL)~\cite{li_generalized_focal} &
Focal exponent $\gamma=2$; scaling factor $\alpha=0.25$;
positive/negative weights $w_{+}=1$ and $w_{-}=1.5$ \\
Type classification &
Cross-entropy (CE) &
Focal cross-entropy~\cite{lin_focal}; label-smoothed cross-entropy~\cite{szegedy_label_smoothing} &
Focal exponent $\gamma=2$; label-smoothing parameter
$\varepsilon_{\mathrm{LS}}$ \\
Center frequency &
Huber~\cite{huber_robust} &
Smooth L1~\cite{girshick_fast_rcnn}; Log-Cosh; Gaussian NLL; quantile~\cite{koenker_quantile} &
Huber threshold $\delta=5$; Smooth-L1 parameter $\beta=0.1$;
unit variance for Gaussian NLL; quantile $\tau=0.5$ \\
Signal strength &
Huber~\cite{huber_robust} &
Smooth L1; Log-Cosh; Gaussian NLL; quantile &
Huber threshold $\delta=5$; Smooth-L1 parameter $\beta=0.1$;
unit variance for Gaussian NLL; quantile $\tau=0.5$ \\
Bandwidth &
Heteroscedastic Gaussian NLL~\cite{kendall_gal_uncertainty} &
Gaussian NLL; Laplace NLL; Cauchy NLL &
Predicted log-variance or log-scale constrained to $[-6,6]$;
bandwidth mean constrained by a softplus activation \\
BOC &
Binary cross-entropy (BCE) &
Weighted BCE; focal BCE~\cite{lin_focal} &
Positive weight $w_{\mathrm{BOC}}=5$ for weighted BCE;
focal exponent $\gamma=2$ \\
Mask &
Binary cross-entropy (BCE) &
Dice~\cite{milletari_vnet}; intersection over union (IoU)~\cite{rahman_iou}; focal mask~\cite{lin_focal};
compound BCE--Dice; Tversky~\cite{salehi_tversky} &
Focal exponent $\gamma=2$; BCE--Dice weights $0.5/0.5$;
Tversky parameters $\alpha_{\mathrm{Tv}}=\beta_{\mathrm{Tv}}=0.5$ \\
\bottomrule
\end{tabularx}
\end{table*}

The interference sources in a snapshot form an unordered set; therefore, the index of a ground-truth source does not correspond to a predefined prediction slot. Following direct set-prediction approaches based on bipartite
matching~\cite{carion_detr}, a pairwise cost matrix
$\mathbf{C}\in\mathbb{R}^{S\times K}$ is constructed between the
$S=6$ prediction slots and the $K$ ground-truth sources. For prediction slot $s$ and source $k$, the assignment cost is defined as
\begin{align}
\label{eq:hungarian_cost}
    C_{s,k}
    ={}&
    c_{\mathrm{O}}\,
    \ell_{\mathrm{BCE}}
    \left(p_s^{\mathrm{O}},1\right)
    -
    c_{\mathrm{T}}\log
    \left(\widehat{p}_{s,c_k}\right)
    \nonumber\\
    &+
    \frac{c_{\mathrm{CF}}}{10}
    \left|\widehat{f}_s-f_k\right|
    +
    \frac{c_{\mathrm{BW}}}{10}
    \left|\widehat{b}_s-b_k\right|
    +
    \frac{c_{\mathrm{SS}}}{10}
    \left|\widehat{r}_s-r_k\right|
    \nonumber\\
    &+
    c_{\mathrm{MB}}\,
    C_{\mathrm{BCE}}^{\mathrm{M}}(s,k)
    +
    c_{\mathrm{MD}}\,
    C_{\mathrm{Dice}}^{\mathrm{M}}(s,k),
\end{align}
where the applied coefficients are $c_{\mathrm{O}}=0.25$, $c_{\mathrm{T}}=0.5$, $c_{\mathrm{CF}}=1$, $c_{\mathrm{BW}}=1$, $c_{\mathrm{SS}}=0.5$, $c_{\mathrm{MB}}=0.5$, and $c_{\mathrm{MD}}=1$.
The mask terms measure the pixel-wise binary cross-entropy and Dice dissimilarity~\cite{milletari_vnet} between the predicted and ground-truth masks. For computational efficiency, the masks are resized to at most $64\times64$ bins when computing the assignment cost. The BOC prediction is optimized after assignment but does not contribute to the matching cost.

The optimal assignment is obtained using the Hungarian algorithm~\cite{kuhn_hungarian}:
\begin{equation}
\label{eq:hungarian_assignment}
    \pi^{\star}
    =
    \arg\min_{\pi\in\Pi_{K,S}}
    \sum_{k=1}^{K}
    C_{\pi(k),k},
\end{equation}
where $\Pi_{K,S}$ denotes the set of injective mappings from the $K$ ground-truth sources to the $S$ prediction slots. The assigned slots receive positive objectness targets and are used for the characterization and segmentation losses, whereas all remaining slots are treated as inactive.

During the first $10\%$ of the training epochs, the mask costs are excluded by setting $c_{\mathrm{MB}}=c_{\mathrm{MD}}=0$. This warm-up prevents initially uninformative mask predictions from dominating the assignment. The mask terms are subsequently activated so that spatial similarity contributes directly to the source association. Likewise, the matching coefficients of disabled regression or segmentation tasks are set to zero. The discrete assignment is treated as fixed during backpropagation, and gradients are computed only through the losses of the resulting matched prediction--target pairs.

\subsection{Multi-Task Training Objective}
\label{label_method_training_objective}

Following the permutation-invariant assignment, the model jointly
optimizes the seven task-specific objectives in a multi-task learning
formulation~\cite{caruana_multitask} (refer to Section~\ref{label_method_model_tasks}).
The seven task losses are collected in
\begin{equation}
    \boldsymbol{\mathcal{L}}
    =
    \left[
        \mathcal{L}_{\mathrm{O}},
        \mathcal{L}_{\mathrm{T}},
        \mathcal{L}_{\mathrm{CF}},
        \mathcal{L}_{\mathrm{SS}},
        \mathcal{L}_{\mathrm{BW}},
        \mathcal{L}_{\mathrm{BOC}},
        \mathcal{L}_{\mathrm{M}}
    \right].
\end{equation}
The total loss is given by
\begin{align}
\label{eq:multitask_objective}
    \mathcal{L}_{\mathrm{total}}
    ={}&
    \lambda_{\mathrm{O}}\mathcal{L}_{\mathrm{O}}
    +
    \lambda_{\mathrm{T}}\mathcal{L}_{\mathrm{T}}
    +
    \lambda_{\mathrm{CF}}\mathcal{L}_{\mathrm{CF}}
    +
    \lambda_{\mathrm{SS}}\mathcal{L}_{\mathrm{SS}}
    \nonumber\\
    &+
    \lambda_{\mathrm{BW}}\mathcal{L}_{\mathrm{BW}}
    +
    \lambda_{\mathrm{BOC}}\mathcal{L}_{\mathrm{BOC}}
    +
    \lambda_{\mathrm{M}}\mathcal{L}_{\mathrm{M}},
\end{align}
where the coefficients $\lambda_{\mathrm{O}}$,
$\lambda_{\mathrm{T}}$, $\lambda_{\mathrm{CF}}$,
$\lambda_{\mathrm{SS}}$, $\lambda_{\mathrm{BW}}$,
$\lambda_{\mathrm{BOC}}$, and $\lambda_{\mathrm{M}}$ control the
contribution of the respective tasks. A task is included by assigning it a
positive weight and excluded by setting its weight to zero. This formulation
allows the same model architecture to be trained using different combinations
of detection, classification, characterization, and segmentation objectives. For each enabled prediction component, we evaluate the task-specific
loss functions summarized in Table~\ref{tab:method_loss_search_space}.
Starting from the reference configuration, we first replace one loss
function at a time and subsequently evaluate selected cross-component
combinations. A loss-function alternative changes only the objective associated with the
respective component; the network architecture, data split, training budget,
and remaining loss functions are kept fixed. In addition to this
one-factor-at-a-time analysis, selected loss functions are combined with
different subsets of the seven prediction components. An omitted component is
implemented by setting its task weight to zero. For the complete joint
configuration, the default loss-weight vector is
\begin{equation}
    \boldsymbol{\lambda}
    =
    \left(
        0.5,\,
        1.0,\,
        0.1,\,
        0.1,\,
        0.1,\,
        1.0,\,
        1.0
    \right),
    \label{eq:default_loss_weights}
\end{equation}
following the task order
$\left(
\mathrm{O},
\mathrm{T},
\mathrm{CF},
\mathrm{SS},
\mathrm{BW},
\mathrm{BOC},
\mathrm{M}
\right)$.
Further details on spectral separability are provided in App.~\ref{label_appendix_separability}.
\section{Experimental Setup}
\label{label_experiments}

\paragraph{Simulation Environment} The dataset is generated using NVIDIA Sionna~\cite{hoydis_cammerer}, an open-source, GPU-accelerated framework for wireless physical-layer and link-level simulations, together with its geometry- and material-aware ray-tracing engine, Sionna RT. Our ablation study employs a 3D environment characterized by dense multipath propagation. The scene consists of a large rectangular enclosure with predominantly metallic walls, ceiling, and floor, as well as interior shelving structures that act as additional reflectors and scatterers. Sionna RT computes the physically grounded propagation paths between the interference sources and the receiver, including the resulting channel and multipath effects. The maximum reflection depth is configurable, allowing the multipath richness to be systematically controlled. The resulting complex baseband signals are exported as raw IQ samples, thereby preserving the amplitude and phase information required for interference segmentation and characterization.

\paragraph{Transmitter \& Receiver Setup} Signal sources are placed at a fixed location within the scene and transmit at a center frequency of $1.57542\,\text{GHz}$, corresponding to the GPS-L1 band center frequency. The transmitted waveform is determined by a configurable source-signal parameter and can be selected from multiple signal types. The receiver is modeled as a single antenna. The interference waveforms are originally digitized at a sampling rate of $125\,\text{MHz}$ and resampled to the simulation sampling rate of $f_{\mathrm{s}}=40.5\,\text{MHz}$. Each simulated snapshot comprises $L=131{,}328$ complex IQ samples, corresponding to an observation duration of approximately $3.24\,\text{ms}$. The spectrogram is computed using segments of $512$ samples with an overlap of $256$ samples, resulting in a hop size of $256$ samples and an overlap ratio of $50\%$. For the complex IQ signals, this configuration yields a time--frequency representation with $512$ frequency bins and $512$ time frames, a frequency-bin spacing of approximately $79.1\,\text{kHz}$, and a temporal hop of approximately $6.32\,\mu\text{s}$.

\begin{figure}[!t]
    \centering
	\begin{minipage}[t]{0.325\linewidth}
        \centering
    	\includegraphics[trim=10 10 10 10, clip, width=1.0\linewidth]{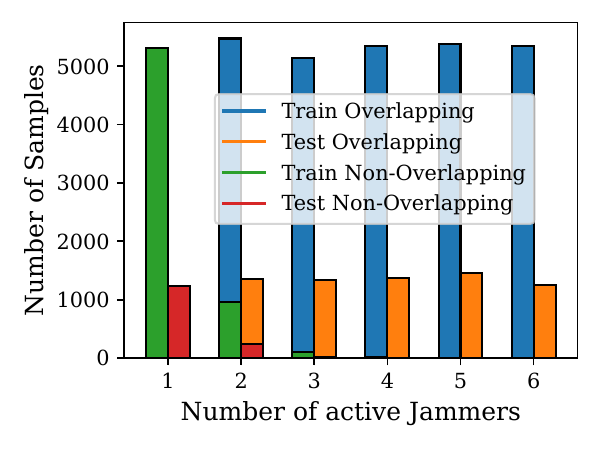}
        \subcaption{Bandwidth samples by number of active jammers.}
        \label{label_figure_statistics1_1}
    \end{minipage}
    \hfill
	\begin{minipage}[t]{0.325\linewidth}
        \centering
    	\includegraphics[trim=10 10 10 10, clip, width=1.0\linewidth]{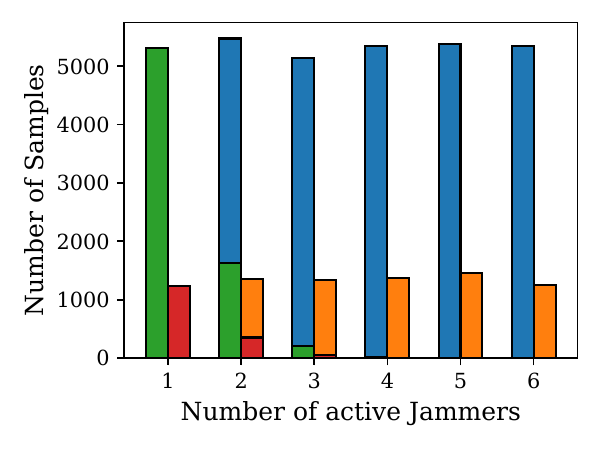}
        \subcaption{Segmentation samples by number of active jammers.}
        \label{label_figure_statistics1_2}
    \end{minipage}
    \hfill
	\begin{minipage}[t]{0.325\linewidth}
        \centering
    	\includegraphics[trim=10 10 10 10, clip, width=1.0\linewidth]{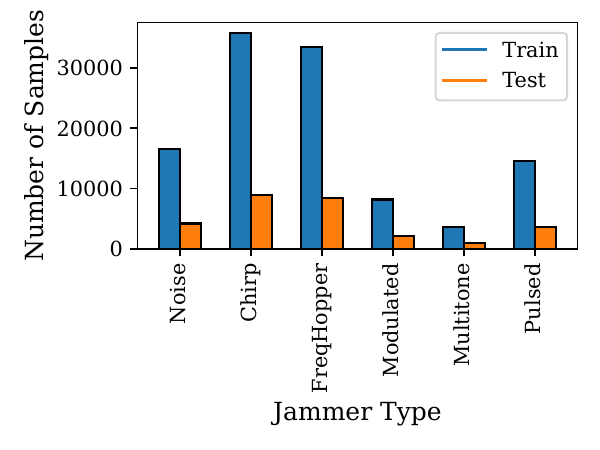}
        \subcaption{Distribution across jammer types.}
        \label{label_figure_statistics1_3}
    \end{minipage}
    \caption{Composition of the training and test sets for overlapping multi-interference characterization. Panels (a) and (b) show the number of overlapping and non-overlapping samples for one to six active jammers in the bandwidth-characterization and segmentation tasks, respectively, while panel (c) shows the distribution across jammer types.}
    \label{label_figure_statistics1}
\end{figure}

\begin{figure}[!t]
    \centering
	\begin{minipage}[t]{0.49\linewidth}
        \centering
    	\includegraphics[trim=10 10 10 10, clip, width=1.0\linewidth]{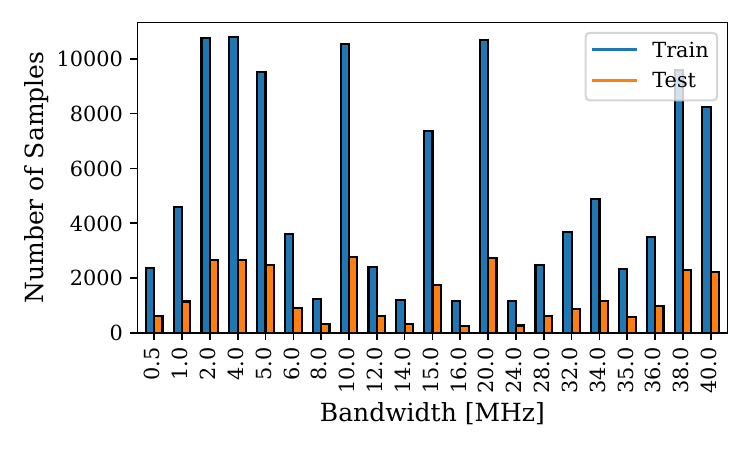}
        \subcaption{Distribution across jammer bandwidths.}
        \label{label_figure_statistics2_1}
    \end{minipage}
    \hfill
	\begin{minipage}[t]{0.49\linewidth}
        \centering
    	\includegraphics[trim=10 10 10 10, clip, width=1.0\linewidth]{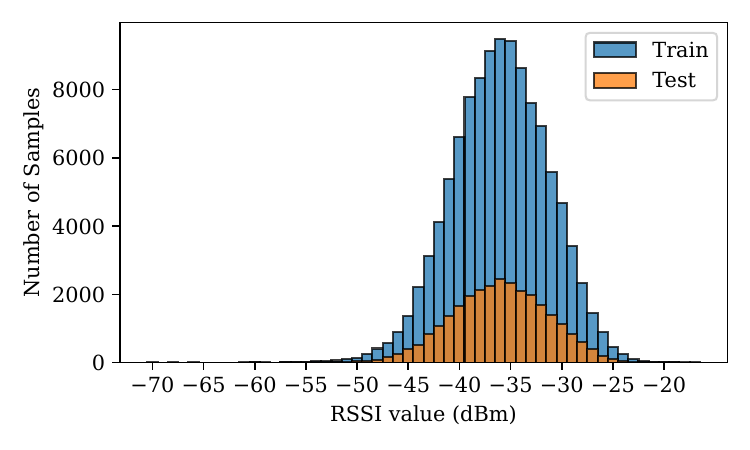}
        \subcaption{Distribution across RSSI levels.}
        \label{label_figure_statistics2_2}
    \end{minipage}
    \caption{Distribution of the jammer parameters in the training and test sets: (a) bandwidth and (b) received signal strength indicator (RSSI).}
    \label{label_figure_statistics2}
\end{figure}

\paragraph{Interference Parameters} Each multi-interference snapshot combines between one and six simultaneously active jammers selected from a catalogue of $94$ waveform configurations. The catalogue comprises six interference classes: \textit{chirp}, \textit{frequency hopping}, \textit{modulated}, \textit{multitone}, \textit{noise}, and \textit{pulsed}. It includes linear and parabolic chirps with different sweep rates, frequency-hopping and BOC-based variants, BPSK- and BOC-modulated signals, multitone emissions, band-limited noise, and pulsed waveforms. Their occupied bandwidths range from $0.5$ to $40\,\text{MHz}$. For every jammer $j$, the transmit power is independently sampled as $P_{\mathrm{TX},j}\sim\mathcal{U}(10,20)\,\text{dBm}$. Furthermore, each waveform is assigned a frequency shift
\begin{equation}
\begin{split}
    \Delta f_j\sim\mathcal{U} \big(-&\min\big(10, (40.5-\text{BW}\_\text{signal}) \cdot 0.9\big)\,\text{MHz}, \\
    &\min\big(10, (40.5-\text{BW}\_\text{signal}) \cdot 0.9\big)\,\text{MHz}\big).
\end{split}
\end{equation}
The shift is accepted only if at least $90\%$ of the jammer's occupied bandwidth remains within the simulated observation band, thereby limiting spectral truncation to at most $10\%$. The shifted jammer signals are subsequently propagated through the Sionna RT channel and superimposed at the receiver to generate the overlapping interference mixture.

\paragraph{Dataset Statistics} Figures~\ref{label_figure_statistics1} and~\ref{label_figure_statistics2} summarize the composition of the generated dataset. The training set contains approximately four times as many samples as the test set, while both splits exhibit similar distributions. Overlapping samples are approximately balanced across scenarios containing one to six simultaneously active jammers, whereas non-overlapping samples are concentrated primarily in scenarios with one or two active jammers. The jammer-type distribution is imbalanced, with \textit{chirp} and \textit{frequency-hopping} signals occurring most frequently, followed by \textit{noise} and \textit{pulsed} interference, while \textit{modulated} and \textit{multitone} signals are less common. The dataset covers 21 discrete bandwidth configurations ranging from $0.5$ to $40$\,MHz, with varying numbers of samples per bandwidth. RSSI values extend approximately from $-70$ to $-20$\,dBm and are concentrated between $-45$ and $-25$\,dBm, with the highest density around $-36$\,dBm. Overall, the similar training and test distributions provide consistent coverage of mixture complexity, jammer types, bandwidths, and received signal strengths across both splits.

\begin{figure*}[!t]
    \centering
	\begin{minipage}[t]{0.138\linewidth}
        \centering
    	\includegraphics[trim=10 22 10 22, clip, width=1.0\linewidth]{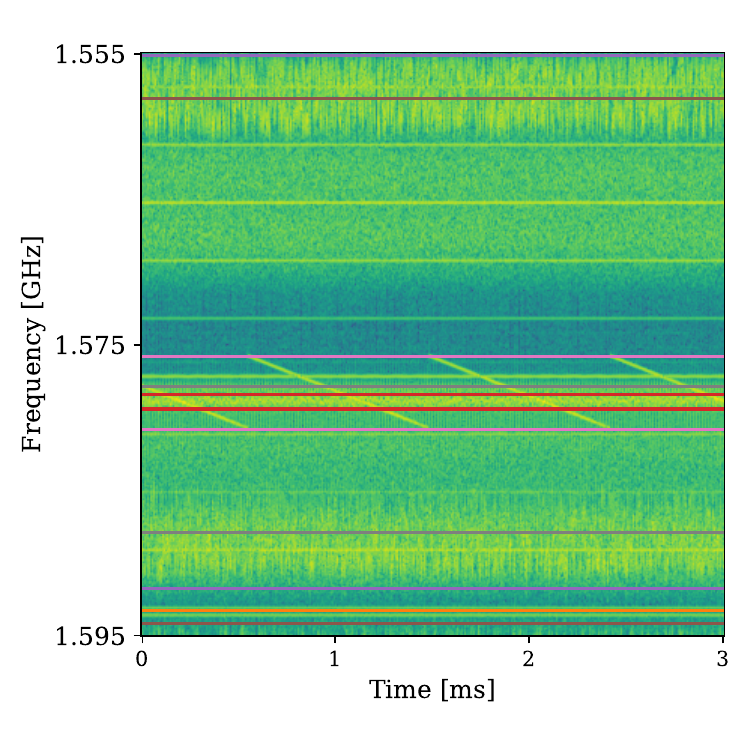}
        \vspace{-0.5cm}
        \subcaption{Spectrogram.}
        \label{label_figure_example_sample1_1}
    \end{minipage}
    \hfill
	\begin{minipage}[t]{0.135\linewidth}
        \centering
    	\includegraphics[trim=10 10 10 10, clip, width=1.0\linewidth]{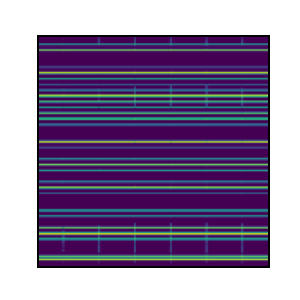}
        \vspace{-0.5cm}
        \subcaption{Source 1.}
        \label{label_figure_example_sample1_2}
    \end{minipage}
    \hfill
	\begin{minipage}[t]{0.135\linewidth}
        \centering
    	\includegraphics[trim=10 10 10 10, clip, width=1.0\linewidth]{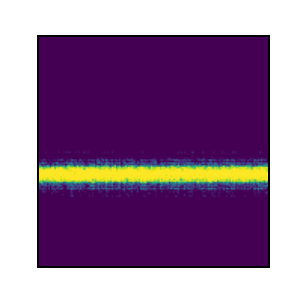}
        \vspace{-0.5cm}
        \subcaption{Source 2.}
        \label{label_figure_example_sample1_3}
    \end{minipage}
    \hfill
	\begin{minipage}[t]{0.135\linewidth}
        \centering
    	\includegraphics[trim=10 10 10 10, clip, width=1.0\linewidth]{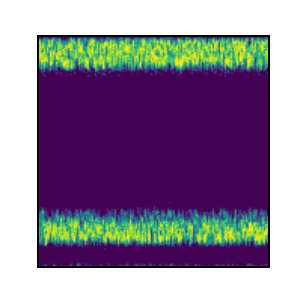}
        \vspace{-0.5cm}
        \subcaption{Source 3.}
        \label{label_figure_example_sample1_4}
    \end{minipage}
    \hfill
	\begin{minipage}[t]{0.135\linewidth}
        \centering
    	\includegraphics[trim=10 10 10 10, clip, width=1.0\linewidth]{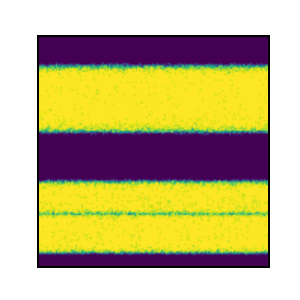}
        \vspace{-0.5cm}
        \subcaption{Source 4.}
        \label{label_figure_example_sample1_5}
    \end{minipage}
    \hfill
	\begin{minipage}[t]{0.135\linewidth}
        \centering
    	\includegraphics[trim=10 10 10 10, clip, width=1.0\linewidth]{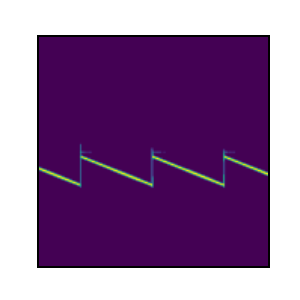}
        \vspace{-0.5cm}
        \subcaption{Source 5.}
        \label{label_figure_example_sample1_6}
    \end{minipage}
    \hfill
	\begin{minipage}[t]{0.135\linewidth}
        \centering
    	\includegraphics[trim=10 10 10 10, clip, width=1.0\linewidth]{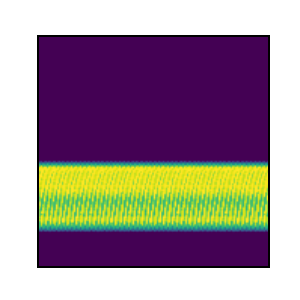}
        \vspace{-0.5cm}
        \subcaption{Source 6.}
        \label{label_figure_example_sample1_7}
    \end{minipage}
    \vspace{-0.15cm}
    \caption{First exemplary snapshot containing six overlapping interference signals: (a) composite spectrogram and (b)--(g) source-specific ground-truth segmentation masks.}
    \label{label_figure_example_sample1}
    \vspace{-0.45cm}
\end{figure*}

\begin{figure*}[!t]
    \centering
	\begin{minipage}[t]{0.138\linewidth}
        \centering
    	\includegraphics[trim=10 22 10 22, clip, width=1.0\linewidth]{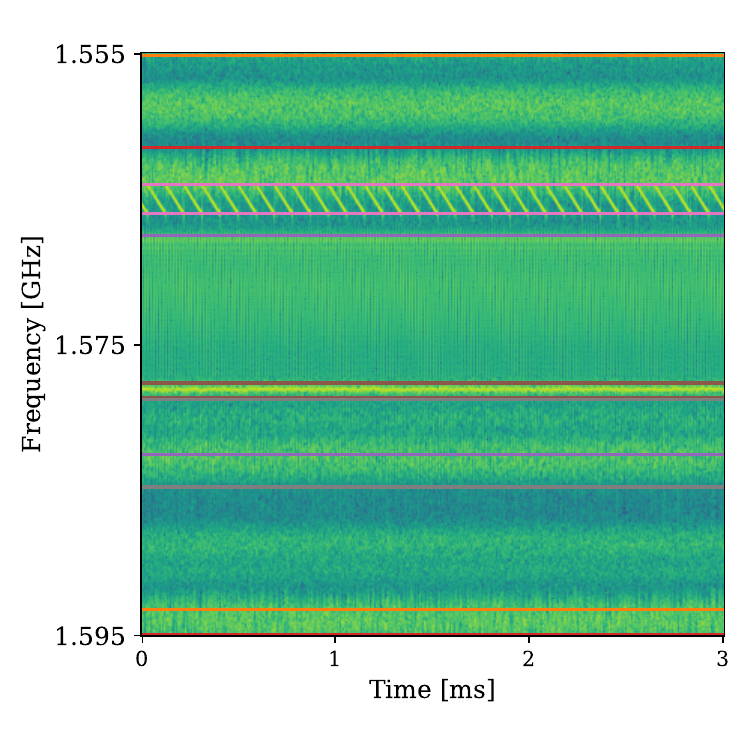}
        \vspace{-0.5cm}
        \subcaption{Spectrogram.}
        \label{label_figure_example_sample2_1}
    \end{minipage}
    \hfill
	\begin{minipage}[t]{0.135\linewidth}
        \centering
    	\includegraphics[trim=10 10 10 10, clip, width=1.0\linewidth]{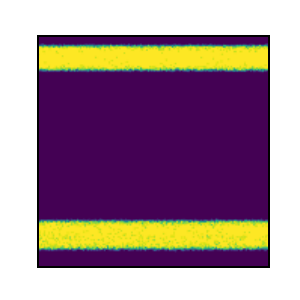}
        \vspace{-0.5cm}
        \subcaption{Source 1.}
        \label{label_figure_example_sample2_2}
    \end{minipage}
    \hfill
	\begin{minipage}[t]{0.135\linewidth}
        \centering
    	\includegraphics[trim=10 10 10 10, clip, width=1.0\linewidth]{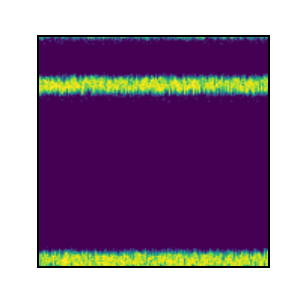}
        \vspace{-0.5cm}
        \subcaption{Source 2.}
        \label{label_figure_example_sample2_3}
    \end{minipage}
    \hfill
	\begin{minipage}[t]{0.135\linewidth}
        \centering
    	\includegraphics[trim=10 10 10 10, clip, width=1.0\linewidth]{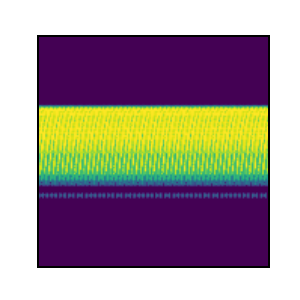}
        \vspace{-0.5cm}
        \subcaption{Source 3.}
        \label{label_figure_example_sample2_4}
    \end{minipage}
    \hfill
	\begin{minipage}[t]{0.135\linewidth}
        \centering
    	\includegraphics[trim=10 10 10 10, clip, width=1.0\linewidth]{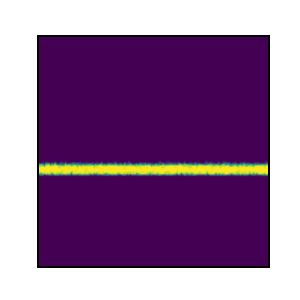}
        \vspace{-0.5cm}
        \subcaption{Source 4.}
        \label{label_figure_example_sample2_5}
    \end{minipage}
    \hfill
	\begin{minipage}[t]{0.135\linewidth}
        \centering
    	\includegraphics[trim=10 10 10 10, clip, width=1.0\linewidth]{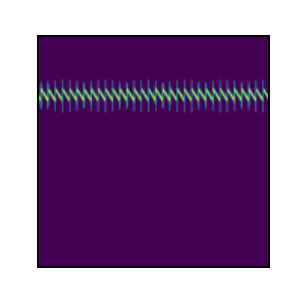}
        \vspace{-0.5cm}
        \subcaption{Source 5.}
        \label{label_figure_example_sample2_6}
    \end{minipage}
    \hfill
	\begin{minipage}[t]{0.135\linewidth}
        \centering
    	\includegraphics[trim=10 10 10 10, clip, width=1.0\linewidth]{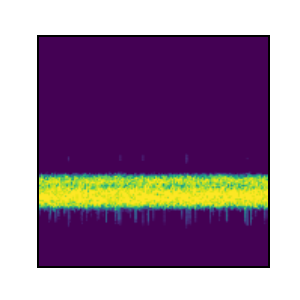}
        \vspace{-0.5cm}
        \subcaption{Source 6.}
        \label{label_figure_example_sample2_7}
    \end{minipage}
    \vspace{-0.15cm}
    \caption{Second exemplary snapshot containing six overlapping interference signals: (a) composite spectrogram and (b)--(g) source-specific ground-truth segmentation masks.}
    \label{label_figure_example_sample2}
    \vspace{-0.45cm}
\end{figure*}

\begin{figure*}[!t]
    \centering
	\begin{minipage}[t]{0.194\linewidth}
        \centering
    	\includegraphics[trim=10 22 10 22, clip, width=1.0\linewidth]{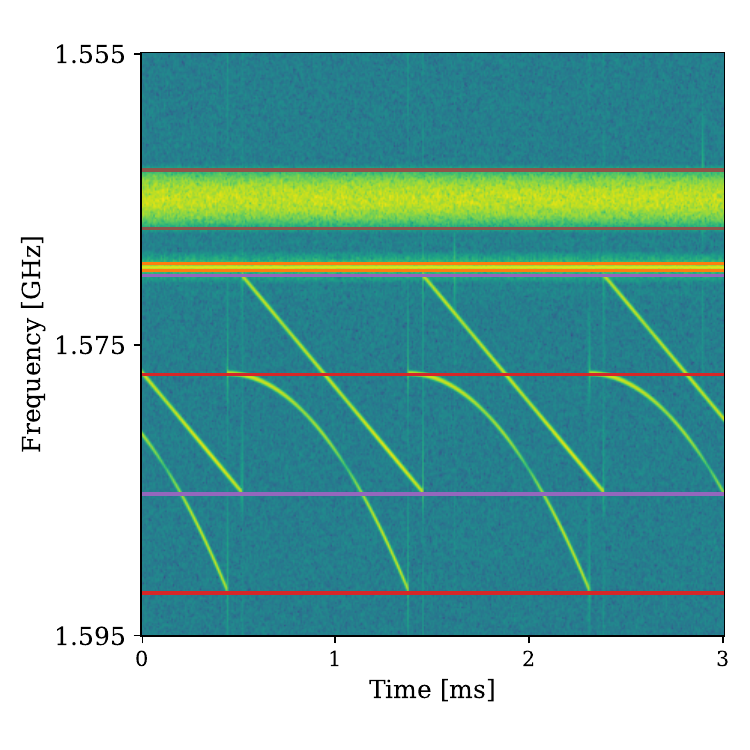}
        \vspace{-0.5cm}
        \subcaption{Composite spectrogram.}
        \label{label_figure_example_sample3_1}
    \end{minipage}
    \hfill
	\begin{minipage}[t]{0.19\linewidth}
        \centering
    	\includegraphics[trim=10 10 10 10, clip, width=1.0\linewidth]{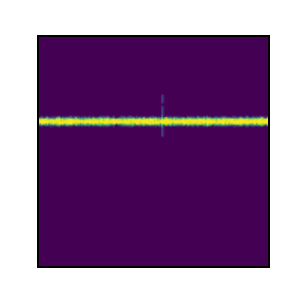}
        \vspace{-0.5cm}
        \subcaption{Source 1.}
        \label{label_figure_example_sample3_2}
    \end{minipage}
    \hfill
	\begin{minipage}[t]{0.19\linewidth}
        \centering
    	\includegraphics[trim=10 10 10 10, clip, width=1.0\linewidth]{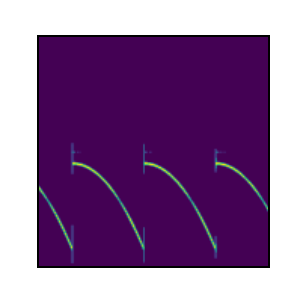}
        \vspace{-0.5cm}
        \subcaption{Source 2.}
        \label{label_figure_example_sample3_3}
    \end{minipage}
    \hfill
	\begin{minipage}[t]{0.19\linewidth}
        \centering
    	\includegraphics[trim=10 10 10 10, clip, width=1.0\linewidth]{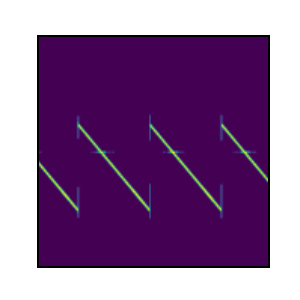}
        \vspace{-0.5cm}
        \subcaption{Source 3.}
        \label{label_figure_example_sample3_4}
    \end{minipage}
    \hfill
	\begin{minipage}[t]{0.19\linewidth}
        \centering
    	\includegraphics[trim=10 10 10 10, clip, width=1.0\linewidth]{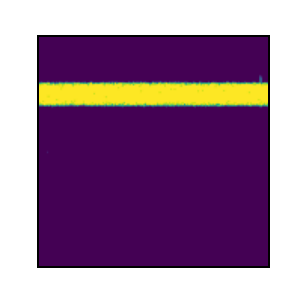}
        \vspace{-0.5cm}
        \subcaption{Source 4.}
        \label{label_figure_example_sample3_5}
    \end{minipage}
    \vspace{-0.15cm}
    \caption{Exemplary snapshot containing four overlapping interference signals: (a) composite spectrogram and (b)--(e) source-specific ground-truth segmentation masks.}
    \label{label_figure_example_sample3}
    \vspace{-0.2cm}
\end{figure*}

\paragraph{Exemplary Samples} Figures~\ref{label_figure_example_sample1}--\ref{label_figure_example_sample3} present representative spectrogram snapshots and their corresponding source-specific segmentation masks. The first two examples each contain six simultaneously active interference signals, whereas the third contains four. The mixtures comprise diverse time--frequency structures, including frequency hoppers, multitone emissions, broadband noise-like signals, and linear or nonlinear frequency sweeps. In the composite spectrograms, these signatures partially obscure one another, and strong emissions can mask weaker components, making direct separation and characterization ambiguous. The corresponding masks isolate the time--frequency support of each interference source while preserving regions shared by multiple signals. These paired mixture--mask samples provide source-specific supervision for the proposed joint segmentation and characterization approach, enabling individual interference signals to be separated before or jointly with estimating their types and signal parameters. Figure~\ref{label_figure_spectrogram_welch_examples} compares four exemplary multi-interference spectrograms with their corresponding Welch power spectral density (PSD) estimates. The examples contain different combinations of stationary narrowband components, broadband emissions, and frequency-varying signals with varying degrees of spectral overlap. While the spectrograms preserve their temporal evolution, Welch's method averages the signal energy over the observation interval and represents occupied frequency ranges as peaks or elevated spectral regions. The agreement between these regions and the colored ground-truth frequency boundaries confirms that the generated signals occupy their intended bandwidths. However, contributions from spectrally overlapping sources merge in the Welch estimates and cannot be assigned unambiguously to individual interferers. This loss of temporal and source-specific information motivates the proposed segmentation-based approach for separating and characterizing overlapping interference signals.

\begin{figure}[!t]
    \centering
	\begin{minipage}[t]{0.25\linewidth}
        \centering
    	\includegraphics[trim=10 22 10 22, clip, width=1.0\linewidth]{images/example_samples/Spec_2.pdf}
    \end{minipage}
    \hfill
	\begin{minipage}[t]{0.231\linewidth}
        \centering
    	\includegraphics[trim=10 10 10 10, clip, width=1.0\linewidth]{images/example_Welch_plots/Welch_2.pdf}
    \end{minipage}
	\begin{minipage}[t]{0.250\linewidth}
        \centering
    	\includegraphics[trim=10 22 10 22, clip, width=1.0\linewidth]{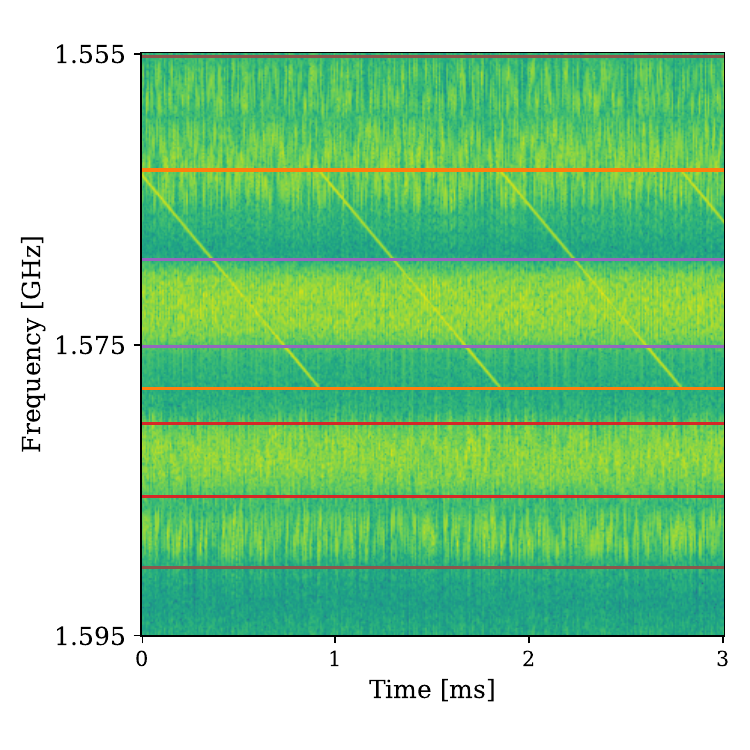}
    \end{minipage}
    \hfill
	\begin{minipage}[t]{0.231\linewidth}
        \centering
    	\includegraphics[trim=10 10 10 10, clip, width=1.0\linewidth]{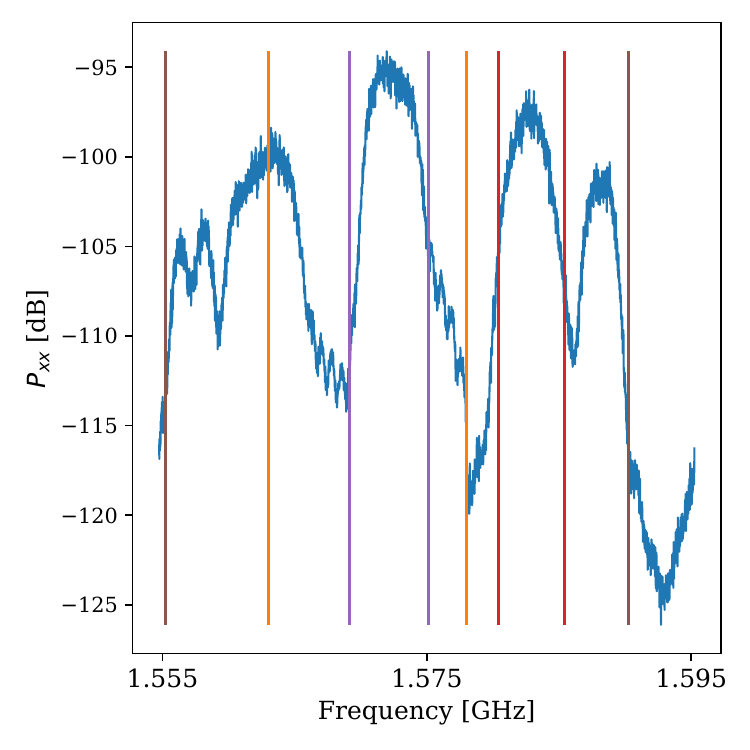}
    \end{minipage}
	\begin{minipage}[t]{0.25\linewidth}
        \centering
    	\includegraphics[trim=10 22 10 22, clip, width=1.0\linewidth]{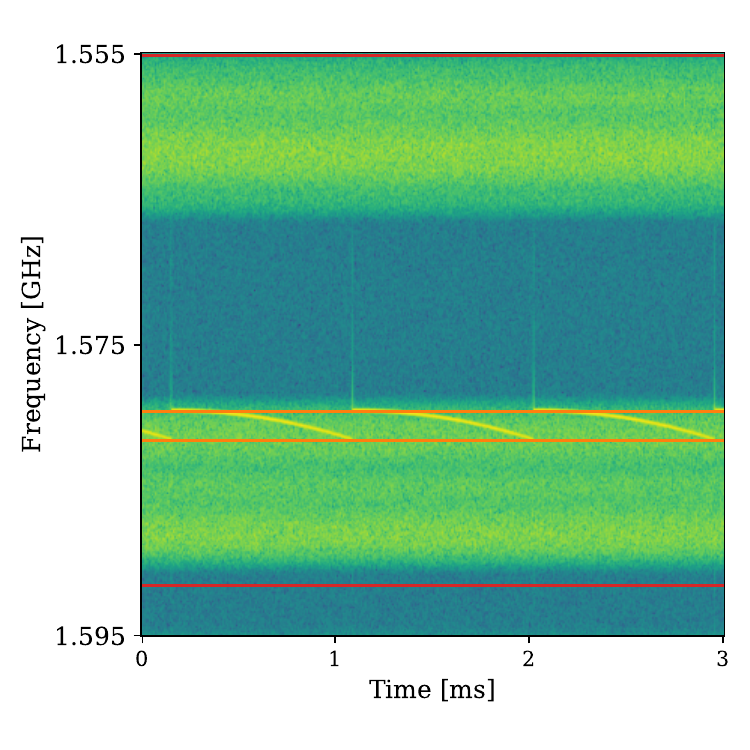}
    \end{minipage}
    \hfill
	\begin{minipage}[t]{0.231\linewidth}
        \centering
    	\includegraphics[trim=10 10 10 10, clip, width=1.0\linewidth]{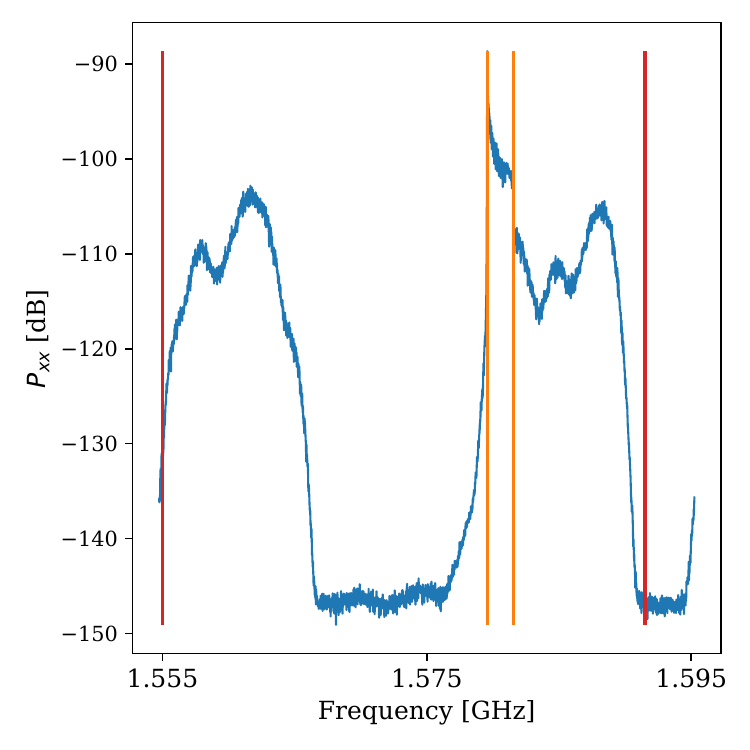}
    \end{minipage}
	\begin{minipage}[t]{0.25\linewidth}
        \centering
    	\includegraphics[trim=10 22 10 22, clip, width=1.0\linewidth]{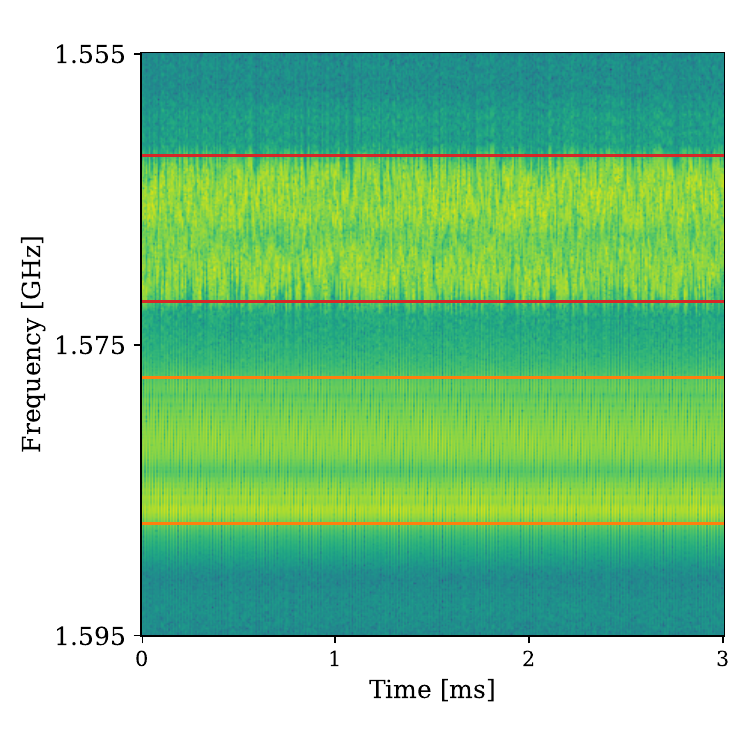}
    \end{minipage}
    \hfill
	\begin{minipage}[t]{0.231\linewidth}
        \centering
    	\includegraphics[trim=10 10 10 10, clip, width=1.0\linewidth]{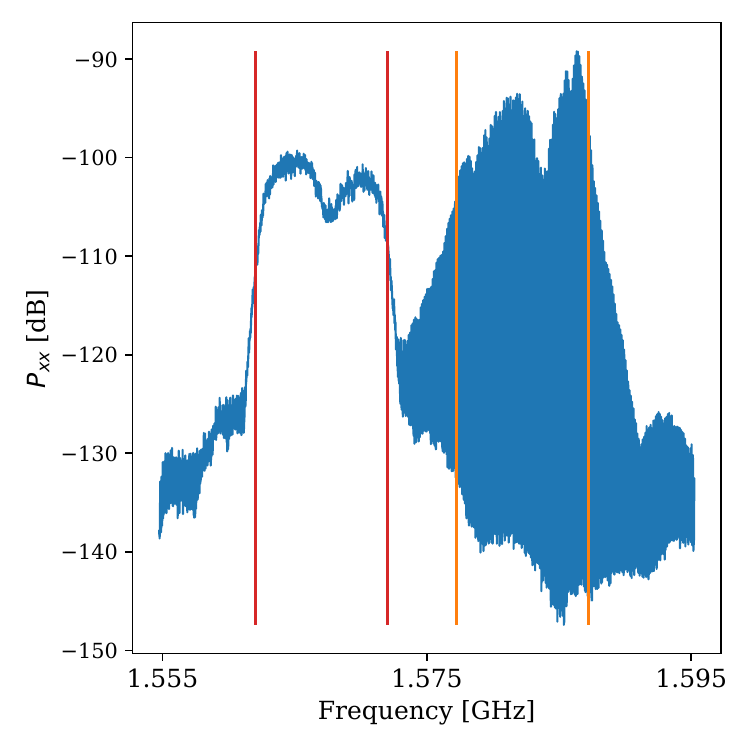}
    \end{minipage}
    \caption{Representative multi-interference snapshots (left) and their corresponding Welch power spectral density estimates (right). Each spectrogram shows a different signal mixture, while the colored lines indicate the ground-truth lower and upper frequency bounds of the individual interference signals.}
    \label{label_figure_spectrogram_welch_examples}
\end{figure}

\section{Evaluation}
\label{label_evaluation}

\subsection{Hardware and Training Setup}
\label{label_evaluation_setup}

All experiments are conducted on compute nodes equipped with NVIDIA Tesla V100-SXM2 GPUs with $32\,\mathrm{GB}$ of VRAM, Intel Xeon CPUs, and $192\,\mathrm{GB}$ of system memory. Each model is trained on a single GPU using a batch size of $32$, with the training samples shuffled at the beginning of each epoch. We train all models for $100$ epochs using the Adam optimizer with an initial learning rate of $10^{-2}$ and its default exponential-decay and moment coefficients. The learning rate is reduced by a factor of $10$ at epochs $60$ and $80$, resulting in learning rates of $10^{-3}$ and $10^{-4}$, respectively. For reproducibility, the Python, NumPy, and PyTorch random-number generators are initialized with a fixed seed for each run, while deterministic cuDNN execution is enabled. Detailed definitions of the evaluation metrics are provided in Appendix~\ref{label_appendix_evaluation_metrics}. Preliminary experiments showed that the MobileNetV2--U-Net variant substantially outperformed the attention-based variant proposed in Section~\ref{label_method_architectures}. Accordingly, all reported results were obtained using the variant without attention.

\subsection{Loss Elements}
\label{label_evaluation_loss_elements}

\begin{table}[t]
\centering
\footnotesize
\setlength{\tabcolsep}{2pt}
\caption{Ablation results for the prediction tasks and mask.
Blank cells indicate omitted components. \textbf{Bold} denotes best results.}
\label{tab:task_ablation}
\begin{tabular}{@{}*{7}{c}@{\hspace{1em}}*{9}{r}@{}}
\toprule
\multicolumn{7}{c}{\textbf{Enabled Components}} &
\multicolumn{9}{c}{\textbf{Evaluation Metrics}} \\
\cmidrule(r){1-7}\cmidrule(l){8-16}
\rotHeader{Objectness} &
\rotHeader{Type Classification} &
\rotHeader{Center Frequency} &
\rotHeader{Signal Strength} &
\rotHeader{Bandwidth} &
\rotHeader{BOC} &
\rotHeader{Mask} &
\rotHeader{Precision} &
\rotHeader{Recall} &
\rotHeader{F1-Score [\%]} &
\rotHeader{Count Accuracy [\%]} &
\rotHeader{Type Accuracy [\%]} &
\rotHeader{BOC Accuracy [\%]} &
\rotHeader{Frequency Error [MHz]} &
\rotHeader{Bandwidth Error [MHz]} &
\rotHeader{RSSI Error [dB]} \\
\midrule
x & x & x & x & x & x & x & 0.70 & 0.71 & 70.6 & 47.2 & 91.6 & 96.5 & 1.33 & 2.79  & 3.34  \\
x & x &   & x & x & x & x & 0.66 & 0.69 & 67.3 & 40.0 & 88.9 & 95.9 & -- & 2.81  & 3.33  \\
x & x & x &   & x & x & x & 0.80 & 0.80 & 79.9 & 59.9 & \textbf{93.7} & \textbf{97.3} & 1.13 & 1.85  & -- \\
x & x & x & x &   & x & x & 0.72 & 0.74 & 73.2 & 27.5 & 88.9 & 89.4 & 1.29 & -- & 3.30  \\
x & x & x & x & x &   & x & 0.71 & 0.72 & 71.3 & 46.6 & 92.0 & -- & 1.33 & 2.68  & 3.32  \\
x & x & x & x & x & x &   & 0.70 & 0.71 & 70.4 & 48.4 & 91.2 & 96.6 & 1.36 & 2.73  & 3.31  \\
x & x &   &   &   &   & x & 0.93 & \textbf{0.90} & \textbf{91.6} & \textbf{64.8} & 97.3 & -- & -- & -- & -- \\
x & x & x &   & x &   &   & 0.80 & 0.80 & 80.1 & 60.1 & 93.5 & -- & \textbf{1.10} & \textbf{1.84}  & -- \\
x & x & x &   & x & x &   & 0.80 & 0.80 & 80.1 & 60.2 & \textbf{93.7} & 97.2 & 1.12 & 1.86  & -- \\
x & x &   & x &   &   & x & 0.74 & 0.77 & 75.3 & 28.8 & 87.9 & -- & -- & -- & \textbf{2.71}  \\
x & x &   &   &   &   &   & \textbf{0.93} & \textbf{0.90} & 91.1 & 63.8 & 97.0 & -- & -- & --  & --  \\
\bottomrule
\end{tabular}
\end{table}

We first analyze the contribution of the seven components of the multi-task
objective in Eq.~\ref{eq:multitask_objective}: objectness
$\mathcal{L}_{\mathrm{O}}$ (Eq.~\ref{eq:objectness_loss}), type classification
$\mathcal{L}_{\mathrm{T}}$ (Eq.~\ref{eq:type_loss}), center frequency
$\mathcal{L}_{\mathrm{CF}}$ (Eq.~\ref{eq:center_frequency_loss}), signal strength
$\mathcal{L}_{\mathrm{SS}}$ (Eq.~\ref{eq:signal_strength_loss}), bandwidth
$\mathcal{L}_{\mathrm{BW}}$ (Eq.~\ref{eq:bandwidth_loss}), BOC detection
$\mathcal{L}_{\mathrm{BOC}}$ (Eq.~\ref{eq:boc_loss}), and mask prediction
$\mathcal{L}_{\mathrm{M}}$ (Eq.~\ref{eq:mask_loss}). Components are omitted by
setting their corresponding loss weights to zero. As shown in Table~\ref{tab:task_ablation}, the complete objective achieves an
F1-score of \(70.6\%\), a count accuracy of \(47.2\%\), and type and BOC
accuracies of \(91.6\%\) and \(96.5\%\). The center-frequency, bandwidth, and
signal-strength errors are \(1.33\,\mathrm{MHz}\), \(2.79\,\mathrm{MHz}\),
and \(3.34\,\mathrm{dB}\), respectively. Removing center-frequency estimation
degrades most metrics, indicating that it provides useful supervision for the
shared representation. Omitting bandwidth estimation increases the F1-score
slightly to \(73.2\%\), but substantially reduces count, type, and BOC
accuracy, showing that bandwidth information supports source separation.
Removing either BOC or mask prediction has only a minor effect on the
remaining metrics. The largest improvement results from omitting signal-strength estimation. In
this case, the F1-score increases to \(79.9\%\), count accuracy to \(59.9\%\),
and the center-frequency and bandwidth errors decrease to
\(1.13\,\mathrm{MHz}\) and \(1.85\,\mathrm{MHz}\). This behavior suggests
negative transfer from signal-strength regression to the other objectives.

\begin{figure}[!t]
    \centering
	\begin{minipage}[t]{0.49\linewidth}
        \centering
    	\includegraphics[trim=10 10 10 10, clip, width=1.0\linewidth]{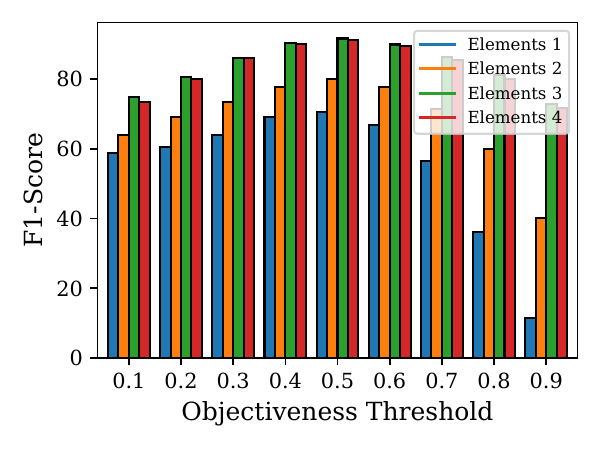}
        \subcaption{F1-Score [\%].}
        \label{label_figure_objectiveness1}
    \end{minipage}
    \hfill
	\begin{minipage}[t]{0.49\linewidth}
        \centering
    	\includegraphics[trim=10 10 10 10, clip, width=1.0\linewidth]{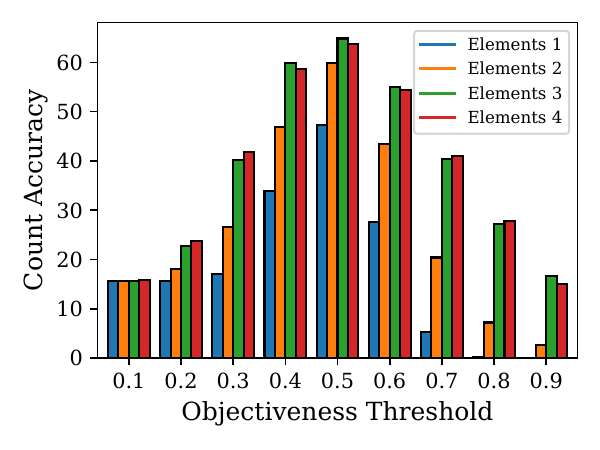}
        \subcaption{Count accuracy [\%].}
        \label{label_figure_objectiveness2}
    \end{minipage}
    \caption{Influence of the objectness threshold on source-detection and count-estimation performance for four component combinations from Table~\ref{tab:task_ablation}: the complete objective (row~1), the objective without signal-strength estimation (row~3), objectness, type classification, and mask prediction (row~7), and objectness with type classification (row~11).}
    \label{label_figure_objectiveness}
\end{figure}

The best two configurations use objectness, type classification, and mask
prediction, or objectness and type classification alone. They achieve
F1-scores of \(91.6\%\) and \(91.1\%\), count accuracies of \(64.8\%\) and
\(63.8\%\), and type accuracies of \(97.3\%\) and \(97.0\%\), respectively.
Mask supervision therefore provides a small additional improvement. Overall, the results reveal a trade-off between comprehensive characterization
and source-detection performance.

Figure~\ref{label_figure_objectiveness} evaluates four representative
configurations: the complete objective (Elements~1), the objective without
signal-strength estimation (Elements~2), objectness, type, and mask prediction
(Elements~3), and objectness with type classification (Elements~4). All
configurations attain their highest F1-score and count accuracy at an
objectness threshold of \(0.5\). Lower thresholds activate too many prediction
slots and therefore reduce count accuracy, whereas higher thresholds reject
true sources and reduce recall. Elements~3 and Elements~4 outperform the characterization-oriented
configurations across almost the entire threshold range and behave similarly
around the optimum, with a small advantage from mask supervision. Elements~2
also consistently outperforms the complete Elements~1 configuration,
confirming that the adverse effect of signal-strength estimation is not
specific to the selected threshold.

\subsection{Loss Functions}
\label{label_evaluation_loss_functions}

\begin{table*}[t]
\centering
\scriptsize
\setlength{\tabcolsep}{3.7pt}
\caption{Evaluation of alternative loss functions. The first row is the
reference configuration; each subsequent row changes one loss function.
Precision/recall are scores, whereas F1/accuracy values are
reported in percent. \textbf{Bold} denotes best results.}
\label{tab:loss_functions}
\vspace{-0.1cm}
\begin{tabular}{@{}*{7}{l}@{\hspace{0.8em}}*{9}{r}@{}}
\toprule
\multicolumn{7}{c}{\textbf{Loss Function by Component}} &
\multicolumn{9}{c}{\textbf{Evaluation Metrics}} \\
\cmidrule(r){1-7}\cmidrule(l){8-16}
\rotHeader{Objectness} &
\rotHeader{Type Classif.} &
\rotHeader{\shortstack[c]{Center\\Frequency}} &
\rotHeader{\shortstack[c]{Signal\\Strength}} &
\rotHeader{Bandwidth} &
\rotHeader{BOC} &
\rotHeader{Mask} &
\rotHeader{Precision} &
\rotHeader{Recall} &
\rotHeader{F1-Score [\%]} &
\rotHeader{\shortstack[c]{Count\\Accuracy [\%]}} &
\rotHeader{\shortstack[c]{Type\\Accuracy [\%]}} &
\rotHeader{\shortstack[c]{BOC\\Accuracy [\%]}} &
\rotHeader{\shortstack[c]{Frequency\\Error [MHz]}} &
\rotHeader{\shortstack[c]{Bandwidth\\Error [MHz]}} &
\rotHeader{\shortstack[c]{RSSI\\Error [dB]}} \\
\midrule
Focal & CE & Huber & Huber & Heteroscedastic & BCE & BCE
& 0.70 & 0.71 & 70.8 & 46.8 & 91.7 & 96.8 & 1.36 & 2.69 & 3.35 \\
\addlinespace[2pt]
Varifocal & CE & Huber & Huber & Heteroscedastic & BCE & BCE
& 0.47 & \textbf{0.79} & 58.8 & 15.7 & 92.5 & \textbf{97.0} & 1.35 & 2.50 & \textbf{3.23} \\
QFL & CE & Huber & Huber & Heteroscedastic & BCE & BCE
& 0.76 & 0.22 & 33.6 & 15.4 & 94.7 & 93.2 & 1.46 & 4.10 & 3.34 \\
\addlinespace[2pt]
Focal & Focal CE & Huber & Huber & Heteroscedastic & BCE & BCE
& 0.70 & 0.70 & 69.9 & 49.4 & 90.3 & 96.5 & 1.34 & 2.75 & 3.30 \\
Focal & LabelSmoothing & Huber & Huber & Heteroscedastic & BCE & BCE
& 0.70 & 0.71 & 70.8 & 47.0 & 91.8 & 96.6 & 1.36 & 2.71 & 3.38 \\
\addlinespace[2pt]
Focal & CE & L1 & Huber & Heteroscedastic & BCE & BCE
& 0.68 & 0.70 & 68.8 & 47.5 & 91.3 & 96.5 & 1.48 & 2.86 & 3.32 \\
Focal & CE & Log-Cosh & Huber & Heteroscedastic & BCE & BCE
& 0.69 & 0.71 & 70.1 & 50.3 & 91.7 & 96.5 & 1.41 & 2.80 & 3.35 \\
Focal & CE & Gaussian-NLL & Huber & Heteroscedastic & BCE & BCE
& 0.68 & 0.70 & 69.0 & 46.4 & 91.2 & 96.3 & 1.43 & 2.90 & 3.37 \\
Focal & CE & Quantile & Huber & Heteroscedastic & BCE & BCE
& 0.67 & 0.69 & 67.8 & 47.1 & 90.8 & 96.0 & 1.61 & 2.91 & 3.34 \\
\addlinespace[2pt]
Focal & CE & Huber & L1 & Heteroscedastic & BCE & BCE
& 0.73 & 0.74 & 73.8 & 52.1 & 92.5 & 96.6 & 1.24 & 2.35 & 3.39 \\
Focal & CE & Huber & Log-Cosh & Heteroscedastic & BCE & BCE
& 0.74 & 0.75 & \textbf{74.5} & \textbf{52.7} & 92.4 & 96.7 & 1.27 & 2.20 & 3.40 \\
Focal & CE & Huber & Gaussian-NLL & Heteroscedastic & BCE & BCE
& 0.69 & 0.71 & 70.1 & 46.3 & 91.5 & 96.5 & 1.38 & 2.76 & 3.36 \\
Focal & CE & Huber & Quantile & Heteroscedastic & BCE & BCE
& \textbf{0.76} & 0.76 & 76.1 & 56.0 & \textbf{93.0} & \textbf{97.0} & \textbf{1.14} & \textbf{1.97} & 3.56 \\
\addlinespace[2pt]
Focal & CE & Huber & Huber & Gaussian-NLL & BCE & BCE
& 0.70 & 0.71 & 70.2 & 45.5 & 91.7 & 96.4 & 1.36 & 2.80 & 3.38 \\
Focal & CE & Huber & Huber & Laplace-NLL & BCE & BCE
& 0.72 & 0.73 & 72.5 & 52.2 & 91.6 & 96.4 & 1.33 & 2.51 & 3.41 \\
Focal & CE & Huber & Huber & Cauchy-NLL & BCE & BCE
& 0.71 & 0.72 & 71.5 & 48.0 & 91.4 & 96.1 & 1.39 & 2.40 & 3.38 \\
\addlinespace[2pt]
Focal & CE & Huber & Huber & Heteroscedastic & Weighted-BCE & BCE
& 0.70 & 0.71 & 70.7 & 45.8 & 91.6 & 96.4 & 1.40 & 2.72 & 3.34 \\
Focal & CE & Huber & Huber & Heteroscedastic & Focal-BCE & BCE
& 0.70 & 0.72 & 71.1 & 48.0 & 91.7 & 96.0 & 1.32 & 2.71 & 3.39 \\
\addlinespace[2pt]
Focal & CE & Huber & Huber & Heteroscedastic & BCE & Dice
& 0.68 & 0.70 & 69.0 & 45.4 & 91.3 & 96.1 & 1.36 & 2.88 & 3.35 \\
Focal & CE & Huber & Huber & Heteroscedastic & BCE & IoU
& 0.71 & 0.72 & 71.3 & 49.2 & 91.9 & 96.7 & 1.34 & 2.71 & 3.35 \\
Focal & CE & Huber & Huber & Heteroscedastic & BCE & Focal-Mask
& 0.70 & 0.71 & 70.3 & 45.3 & 91.3 & 96.8 & 1.39 & 2.74 & 3.32 \\
Focal & CE & Huber & Huber & Heteroscedastic & BCE & BCE-Dice
& 0.70 & 0.71 & 70.6 & 46.9 & 91.8 & 96.5 & 1.32 & 2.78 & 3.33 \\
Focal & CE & Huber & Huber & Heteroscedastic & BCE & Tversky
& 0.72 & 0.73 & 72.6 & 49.7 & 91.9 & 96.8 & 1.31 & 2.57 & 3.36 \\
\bottomrule
\end{tabular}%
\end{table*}

We next evaluate the loss-function alternatives introduced in
Table~\ref{tab:method_loss_search_space}. Starting from the reference
configuration in the first row of Table~\ref{tab:loss_functions}, one loss
function is replaced at a time while the remaining objectives, network
architecture, and training setup are kept unchanged. The reference
configuration obtains an F1-score of \(70.8\%\), a count accuracy of
\(46.8\%\), and type and BOC accuracies of \(91.7\%\) and \(96.8\%\),
respectively. For objectness prediction, the reference focal loss provides the most balanced
precision and recall. Varifocal loss increases recall to \(0.79\), but reduces
precision to \(0.47\), resulting in an F1-score of \(58.8\%\). Conversely,
quality focal loss favors precision (\(0.76\)) at the expense of recall
(\(0.22\)), producing the lowest F1-score of \(33.6\%\). Both alternatives
also reduce count accuracy to approximately \(15\%\). For type
classification, focal cross-entropy slightly decreases F1 and type accuracy,
whereas label smoothing performs almost identically to conventional
cross-entropy. Similarly, none of the alternative center-frequency losses
outperforms Huber loss.

The most pronounced improvements are obtained by changing the
signal-strength loss. L1 and Log-Cosh increase the F1-score to \(73.8\%\) and
\(74.5\%\), respectively. Quantile loss provides the best overall result,
reaching a precision and recall of \(0.76\), an F1-score of \(76.1\%\), and a
count accuracy of \(56.0\%\). It also improves type and BOC accuracy to
\(93.0\%\) and \(97.0\%\) and reduces the frequency and bandwidth errors to
\(1.14\,\mathrm{MHz}\) and \(1.97\,\mathrm{MHz}\). However, its RSSI error
increases from \(3.35\,\mathrm{dB}\) to \(3.56\,\mathrm{dB}\). The improved
characterization of the other properties therefore comes at the
expense of signal-strength accuracy. Together with the component ablation in
Table~\ref{tab:task_ablation}, this result confirms that the signal-strength
objective has the strongest influence on multi-task optimization.

For bandwidth estimation, Laplace NLL offers the best balance, increasing the
F1-score to \(72.5\%\) and count accuracy to \(52.2\%\), while reducing the
bandwidth error to \(2.51\,\mathrm{MHz}\). Cauchy NLL obtains an even lower
bandwidth error of \(2.40\,\mathrm{MHz}\), but smaller improvements in the
detection metrics. Changing the BOC loss has only a minor influence:
focal BCE slightly improves the F1-score to \(71.1\%\), but reduces BOC
accuracy to \(96.0\%\).

Among the mask objectives, Tversky loss performs best, achieving an F1-score
of \(72.6\%\), a count accuracy of \(49.7\%\), and type and BOC accuracies of
\(91.9\%\) and \(96.8\%\). IoU loss also provides a small improvement, whereas
Dice, focal-mask, and combined BCE--Dice losses do not outperform the
reference BCE objective. Overall, the results show that the choice of
signal-strength loss has a substantially greater effect than the alternative
classification, bandwidth, BOC, or mask losses. Quantile loss consequently
provides the strongest detection-oriented configuration, while the reference
Huber loss remains preferable when RSSI estimation accuracy is prioritized.

\subsection{Joint Objective and Source Assignment}
\label{label_evaluation_hungarian_objectiveness}

\begin{table*}[t]
\centering
\scriptsize
\caption{Evaluation of selected objective and loss-function combinations. Objectness is active.
Blank cells indicate omitted objectives. Precision and recall are scores;
F1 and accuracy values are reported in percent. \textbf{Bold} denotes best results.}
\label{tab:objective_loss_combinations}
\vspace{-0.1cm}
\begin{tabular}{@{}*{7}{l}@{\hspace{0.8em}}*{9}{r}@{}}
\toprule
\multicolumn{7}{c}{\textbf{Loss Function by Component}} &
\multicolumn{9}{c}{\textbf{Evaluation Metrics}} \\
\cmidrule(r){1-7}\cmidrule(l){8-16}
\rotHeader{Objectness} &
\rotHeader{Type Classif.} &
\rotHeader{\shortstack[c]{Center\\Frequency}} &
\rotHeader{\shortstack[c]{Signal\\Strength}} &
\rotHeader{Bandwidth} &
\rotHeader{BOC} &
\rotHeader{Mask} &
\rotHeader{Precision} &
\rotHeader{Recall} &
\rotHeader{F1-Score} &
\rotHeader{\shortstack[c]{Count\\Accuracy [\%]}} &
\rotHeader{\shortstack[c]{Type\\Accuracy [\%]}} &
\rotHeader{\shortstack[c]{BOC\\Accuracy [\%]}} &
\rotHeader{\shortstack[c]{Frequency\\Error [MHz]}} &
\rotHeader{\shortstack[c]{Bandwidth\\Error [MHz]}} &
\rotHeader{\shortstack[c]{RSSI\\Error [dB]}} \\
\midrule
Focal & CE & Huber & Huber & Heteroscedastic & BCE & BCE
& 0.81 & 0.69 & 74.8 & 50.6 & 97.3 & \textbf{98.4} & 1.29 & 2.54 & \textbf{3.19} \\
Focal & CE & Huber &       & Heteroscedastic & BCE & BCE
& 0.88 & 0.79 & 83.5 & \textbf{61.7} & 97.2 & 98.1 & \textbf{1.12} & \textbf{1.62} & -- \\
Focal & CE &       &       &                  &     & BCE
& \textbf{0.98} & \textbf{0.85} & \textbf{91.1} & 59.9 & 98.7 & -- & -- & -- & -- \\
Focal & CE & Huber &       & Heteroscedastic & BCE &
& 0.88 & 0.77 & 82.0 & 57.4 & 97.3 & \textbf{98.4} & 1.13 & 1.67 & -- \\
\addlinespace[3pt]
Focal & CE & Huber & Quantile & Laplace-NLL & Focal-BCE & Tversky
& 0.72 & 0.61 & 65.9 & 49.6 & 95.6 & 96.3 & 1.39 & 3.47 & 3.39 \\
Focal & CE & Huber &          & Laplace-NLL & Focal-BCE & Tversky
& 0.71 & 0.61 & 65.4 & 48.3 & 95.7 & 96.7 & 1.45 & 3.61 & -- \\
Focal & CE &       &          &             &           & Tversky
& \textbf{0.98} & 0.80 & 88.1 & 50.5 & \textbf{99.1} & -- & -- & -- & -- \\
Focal & CE & Huber &          & Laplace-NLL & Focal-BCE &
& 0.86 & 0.74 & 79.8 & 55.2 & 96.7 & 97.7 & 1.28 & 2.08 & -- \\
\bottomrule
\vspace{-0.4cm}
\end{tabular}%
\end{table*}

\begin{table*}[t]
\centering
\scriptsize
\caption{Evaluation results for the three methods under different objective
configurations and Hungarian assignment settings. Blank loss cells indicate
omitted objectives. Reference methods are Radar-Compound-Jamming (RCJ)~\cite{li_compound_segmentation}, Hierarchical-Feature-Integration (HFI)~\cite{qu_hifi_yolo}, and Transformer-Net (TN)~\cite{vagollari_transformer_rf}. \textbf{Bold} denotes best results.}
\label{tab:objective_hungarian_configurations}
\vspace{-0.1cm}
\begin{tabular}{@{}*{9}{l}@{\hspace{0.8em}}*{9}{r}@{}}
\toprule
\multicolumn{9}{c}{\textbf{Model Configuration}} &
\multicolumn{9}{c}{\textbf{Evaluation Metrics}} \\
\cmidrule(r){1-9}\cmidrule(l){10-18}
\rotHeader{Method} &
\rotHeader{Objectness} &
\rotHeader{Type Classif.} &
\rotHeader{\shortstack[c]{Center\\Frequency}} &
\rotHeader{\shortstack[c]{Signal\\Strength}} &
\rotHeader{Bandwidth} &
\rotHeader{BOC} &
\rotHeader{Mask} &
\rotHeader{Hungarian} &
\rotHeader{Precision} &
\rotHeader{Recall} &
\rotHeader{F1-Score [\%]} &
\rotHeader{\shortstack[c]{Count\\Accuracy [\%]}} &
\rotHeader{\shortstack[c]{Type\\Accuracy [\%]}} &
\rotHeader{\shortstack[c]{BOC\\Accuracy [\%]}} &
\rotHeader{\shortstack[c]{Frequency\\Error [MHz]}} &
\rotHeader{\shortstack[c]{Bandwidth\\Error [MHz]}} &
\rotHeader{\shortstack[c]{RSSI\\Error [dB]}} \\
\midrule
\shortstack[l]{RCJ}
& Focal & CE &       &       &                  &     &     & type
& 0.92 & \textbf{0.91} & 91.5 & 68.7 & 96.5 & -- & -- & -- & -- \\
\shortstack[l]{HFI}
& Focal & CE &       &       &                  &     &     & type
& 0.45 & 0.75 & 56.1 & 16.1 & 75.4 & -- & -- & -- & -- \\
TN
& Focal & CE & Huber & Huber & Heteroscedastic & BCE & BCE & type
& 0.01 & 0.01 & 0.9 & 15.7 & 29.9 & 72.6 & \textbf{4.30} & 11.90 & \textbf{4.16} \\
\addlinespace[3pt]
\shortstack[l]{RCJ}
& Focal & CE &       &       &                  &     &     & type
& 0.92 & 0.89 & 90.9 & 74.8 & 96.6 & -- & -- & -- & -- \\
\shortstack[l]{HFI}
& Focal & CE &       &       &                  &     &     & type
& 0.68 & 0.33 & 44.3 & 30.9 & 70.7 & -- & -- & -- & -- \\
TN
& Focal & CE & Huber & Huber & Heteroscedastic & BCE & BCE & type
& 0.04 & 0.08 & 5.1 & 6.4 & 42.9 & \textbf{100.0} & 5.51 & \textbf{7.08} & 4.18 \\
\addlinespace[3pt]
\shortstack[l]{RCJ}
& Focal & CE &       &       &                  &     &     & type+objectness
& \textbf{0.99} & 0.87 & \textbf{92.2} & 64.5 & \textbf{99.0} & -- & -- & -- & -- \\
\shortstack[l]{HFI}
& Focal & CE &       &       &                  &     &     & type+objectness
& 0.65 & 0.34 & 44.8 & 16.3 & 71.2 & -- & -- & -- & -- \\
TN
& Focal & CE & Huber & Huber & Heteroscedastic & BCE & BCE & type+objectness
& 0.00 & 0.00 & 0.0 & 0.0 & 0.0 & 0.0 & nan & nan & nan \\
\addlinespace[3pt]
\shortstack[l]{RCJ}
& Focal & CE &       &       &                  &     &     & type+objectness
& 0.98 & 0.87 & \textbf{92.2} & \textbf{75.1} & 98.9 & -- & -- & -- & -- \\
\shortstack[l]{HFI}
& Focal & CE &       &       &                  &     &     & type+objectness
& 0.68 & 0.26 & 38.2 & 31.1 & 68.6 & -- & -- & -- & -- \\
TN
& Focal & CE & Huber & Huber & Heteroscedastic & BCE & BCE & type+objectness
& 0.00 & 0.00 & 0.0 & 0.0 & 0.0 & 0.0 & nan & nan & nan \\
\bottomrule
\vspace{-0.35cm}
\end{tabular}%
\end{table*}

For all configurations, prediction slots are assigned to the unordered
ground-truth sources using the Hungarian cost in Eq.~\ref{eq:hungarian_cost}
and assignment in Eq.~\ref{eq:hungarian_assignment}. At inference, slots with
an objectness probability of at least \(0.5\) are considered active, as this
threshold provides the best balance between F1-score and count accuracy in
Figure~\ref{label_figure_objectiveness}. Table~\ref{tab:objective_loss_combinations} compares selected objective and
loss-function combinations under this common procedure. With all objectives enabled, the reference configuration achieves an
F1-score of \(74.8\%\), type and BOC accuracies of \(97.3\%\) and \(98.4\%\),
and frequency, bandwidth, and RSSI errors of \(1.29\,\mathrm{MHz}\),
\(2.54\,\mathrm{MHz}\), and \(3.19\,\mathrm{dB}\). Removing signal-strength
estimation increases the F1-score to \(83.5\%\) and count accuracy from
\(50.6\%\) to \(61.7\%\), while reducing the frequency and bandwidth errors
to \(1.12\,\mathrm{MHz}\) and \(1.62\,\mathrm{MHz}\). This confirms the
negative interaction between signal-strength regression and the remaining
tasks.

The best detection result is obtained using objectness, type classification,
and BCE mask supervision, reaching a precision of \(0.98\), recall of
\(0.85\), F1-score of \(91.1\%\), and type accuracy of \(98.7\%\). Removing
mask supervision from the configuration with objectness, type, center
frequency, bandwidth, and BOC reduces the F1-score from \(83.5\%\) to
\(82.0\%\) and count accuracy from \(61.7\%\) to \(57.4\%\). Mask supervision
therefore supports source detection and assignment beyond its segmentation
objective.

The loss alternatives in the lower half of the table perform favorably when
evaluated individually, but their improvements do not transfer to the joint
objective. Combining quantile, Laplace NLL, focal BCE, and Tversky losses
reduces the F1-score to \(65.9\%\), compared with \(74.8\%\) for the reference
losses. Replacing BCE with Tversky in the reduced segmentation configuration
similarly decreases the F1-score from \(91.1\%\) to \(88.1\%\), despite
increasing type accuracy to \(99.1\%\). Overall, the results demonstrate
interactions between the losses through the shared representation and source
assignment. The objectness--type--mask configuration provides the strongest
detection performance, whereas omitting signal-strength estimation offers
the best compromise between detection and multi-parameter characterization.

\subsection{Comparison to Related Work}
\label{label_evaluation_related_work}

Table~\ref{tab:objective_hungarian_configurations} compares the proposed
approach with adapted implementations of RCJ, HFI, and TN. Since RCJ and HFI
predict only source presence and type, their comparison is restricted to the
corresponding detection and classification metrics, whereas TN is also
evaluated for characterization. RCJ provides the strongest reference performance, achieving F1-scores between
\(90.9\%\) and \(92.2\%\), count accuracies between \(64.5\%\) and \(75.1\%\),
and type accuracies between \(96.5\%\) and \(99.0\%\). Including objectness in
the Hungarian cost increases precision from \(0.92\) to \(0.98\)--\(0.99\)
and improves the F1-score to \(92.2\%\), although recall decreases slightly
to \(0.87\). Objectness-aware assignment therefore produces more selective
source predictions for RCJ. HFI performs substantially worse, with F1-scores between \(38.2\%\) and
\(56.1\%\), while including objectness further reduces recall. TN fails to
learn reliable source predictions: type-based matching yields F1-scores of
only \(0.9\%\) and \(5.1\%\), and adding objectness causes all prediction
slots to be rejected. The corresponding parameter errors are therefore
undefined. These results indicate that objectness-aware matching requires
sufficiently reliable objectness predictions.

The reduced objectness--type--mask configuration of the proposed method
achieves an F1-score of \(91.1\%\) and a type accuracy of \(98.7\%\), making
it competitive with RCJ while additionally producing source-specific masks.
The complete model reaches an F1-score of \(74.8\%\), type and BOC accuracies
of \(97.3\%\) and \(98.4\%\), and frequency, bandwidth, and RSSI errors of
\(1.29\,\mathrm{MHz}\), \(2.54\,\mathrm{MHz}\), and \(3.19\,\mathrm{dB}\).
Unlike RCJ and HFI, the proposed method therefore jointly detects, segments,
and characterizes multiple overlapping sources while retaining competitive
detection performance.

\subsection{Multi-Interference Characterization}
\label{label_evaluation_multi_interferences}

\begin{table}[t]
\centering
\scriptsize
\caption{Performance by number of jammers and overlap status.
Precision and recall are scores; F1 and accuracy values are percentages. \textbf{Bold} denotes best results.}
\label{tab:overlap_jammers}
\begin{tabular}{@{}cc@{\hspace{1em}}rrrrrrr@{}}
\toprule
\multicolumn{2}{c}{\textbf{Scenario}} &
\multicolumn{7}{c}{\textbf{Evaluation Metrics}} \\
\cmidrule(r){1-2}\cmidrule(l){3-9}
\rotHeader{Overlapping} &
\rotHeader{\shortstack[c]{Number of\\Jammers}} &
\rotHeader{Precision} &
\rotHeader{Recall} &
\rotHeader{F1-Score [\%]} &
\rotHeader{\shortstack[c]{Type\\Accuracy [\%]}} &
\rotHeader{\shortstack[c]{BOC\\Accuracy [\%]}} &
\rotHeader{\shortstack[c]{IoU}} &
\rotHeader{\shortstack[c]{DICE}} \\
\midrule
False & 1 & 0.97 & \textbf{1.00} & \textbf{98.6} & \textbf{100.0} & \textbf{100.0} & \textbf{0.92} & \textbf{0.96} \\
False & 2 & 0.99 & 0.98 & 98.2 & \textbf{100.0} & \textbf{100.0} & 0.86 & 0.92 \\
False & 3 & 0.98 & 0.83 & 90.0 & \textbf{100.0} & 97.8 & 0.81 & 0.88 \\
False & 4 & \textbf{1.00} & 0.75 & 85.7 & \textbf{100.0} & \textbf{100.0} & 0.80 & 0.88 \\
\addlinespace[3pt]
True  & 2 & 0.95 & 0.93 & 93.9 & 99.1 & 99.8 & 0.80 & 0.87 \\
True  & 3 & 0.93 & 0.86 & 89.2 & 98.6 & 98.9 & 0.74 & 0.82 \\
True  & 4 & 0.90 & 0.82 & 85.7 & 97.8 & 98.2 & 0.67 & 0.77 \\
True  & 5 & 0.86 & 0.75 & 80.3 & 96.7 & 98.0 & 0.61 & 0.71 \\
True  & 6 & 0.82 & 0.69 & 74.6 & 95.2 & 96.5 & 0.55 & 0.66 \\
\bottomrule
\end{tabular}
\end{table}

Table~\ref{tab:overlap_jammers} analyzes performance with respect to the
number of simultaneously active interference sources and their overlap.
For one or two non-overlapping sources, the model achieves F1-scores of
\(98.6\%\) and \(98.2\%\), respectively, together with almost perfect type
and BOC classification. With three and four sources, the F1-score decreases
to \(90.0\%\) and \(85.7\%\). Precision remains between \(0.98\) and \(1.00\),
whereas recall decreases to \(0.75\), indicating that the model primarily
misses sources rather than generating false detections as the cardinality
increases.

\begin{figure}[!t]
    \centering
	\begin{minipage}[t]{0.24\linewidth}
        \centering
    	\includegraphics[trim=26 12 54 10, clip, width=1.0\linewidth]{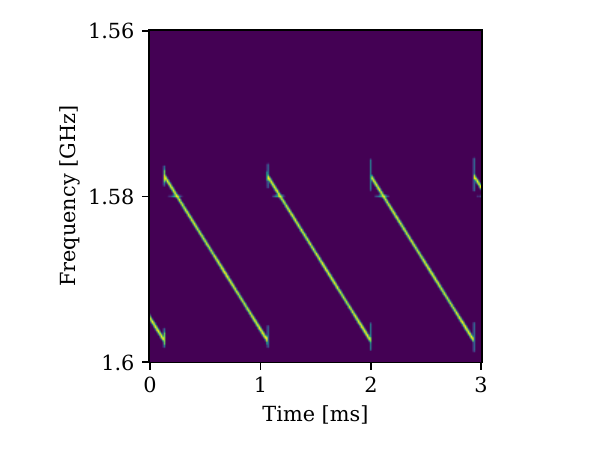}
    \end{minipage}
    \hfill
	\begin{minipage}[t]{0.24\linewidth}
        \centering
    	\includegraphics[trim=26 12 54 10, clip, width=1.0\linewidth]{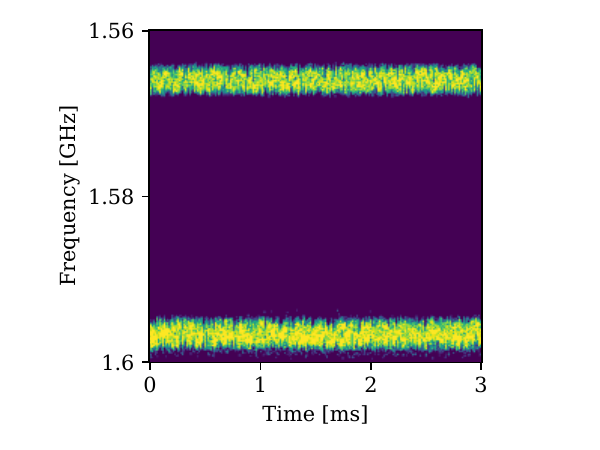}
    \end{minipage}
    \hfill
	\begin{minipage}[t]{0.24\linewidth}
        \centering
    	\includegraphics[trim=26 12 54 10, clip, width=1.0\linewidth]{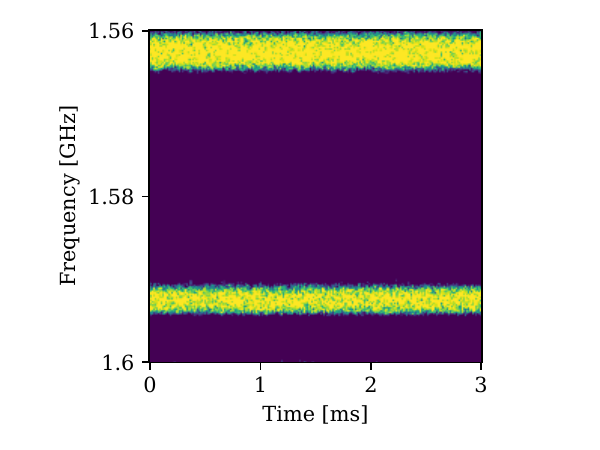}
    \end{minipage}
    \hfill
	\begin{minipage}[t]{0.24\linewidth}
        \centering
    	\includegraphics[trim=26 12 54 10, clip, width=1.0\linewidth]{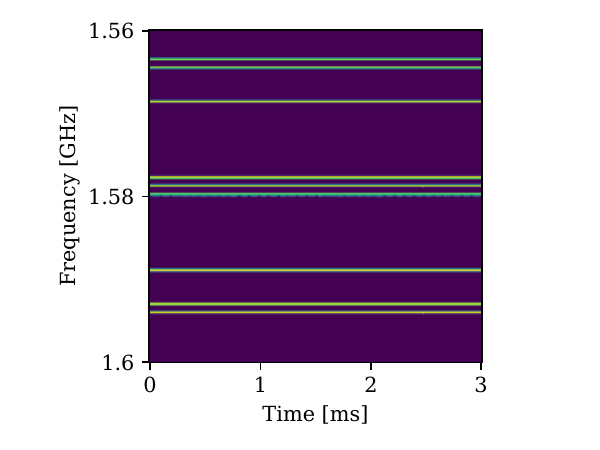}
    \end{minipage}
	\begin{minipage}[t]{0.24\linewidth}
        \centering
    	\includegraphics[trim=26 12 54 10, clip, width=1.0\linewidth]{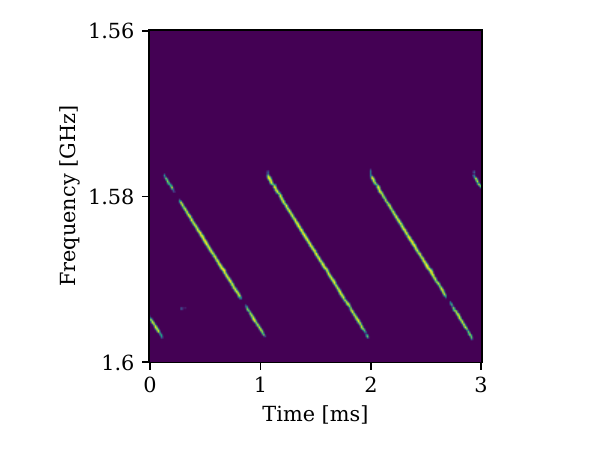}
    \end{minipage}
    \hfill
	\begin{minipage}[t]{0.24\linewidth}
        \centering
    	\includegraphics[trim=26 12 54 10, clip, width=1.0\linewidth]{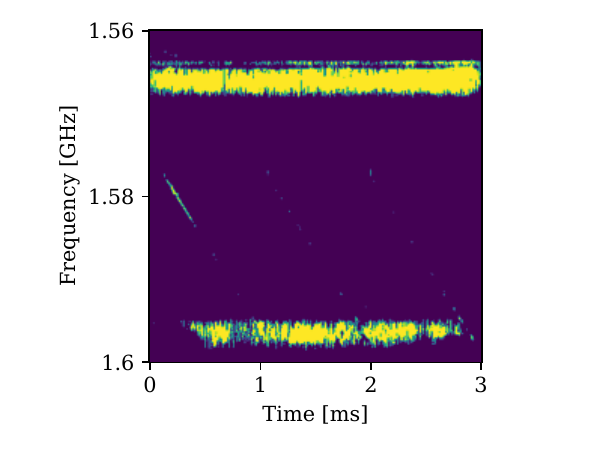}
    \end{minipage}
    \hfill
	\begin{minipage}[t]{0.24\linewidth}
        \centering
    	\includegraphics[trim=26 12 54 10, clip, width=1.0\linewidth]{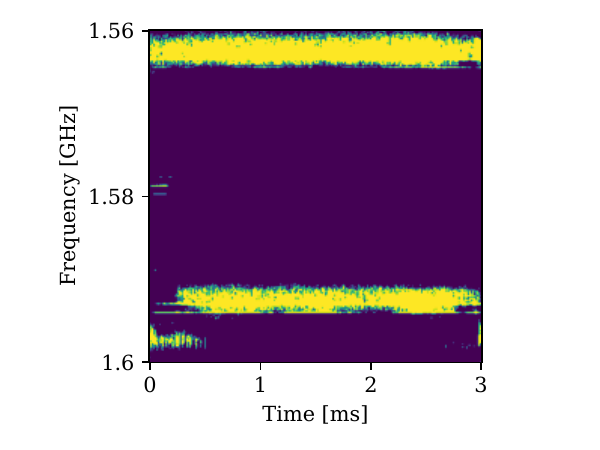}
    \end{minipage}
    \hfill
	\begin{minipage}[t]{0.24\linewidth}
        \centering
    	\includegraphics[trim=26 12 54 10, clip, width=1.0\linewidth]{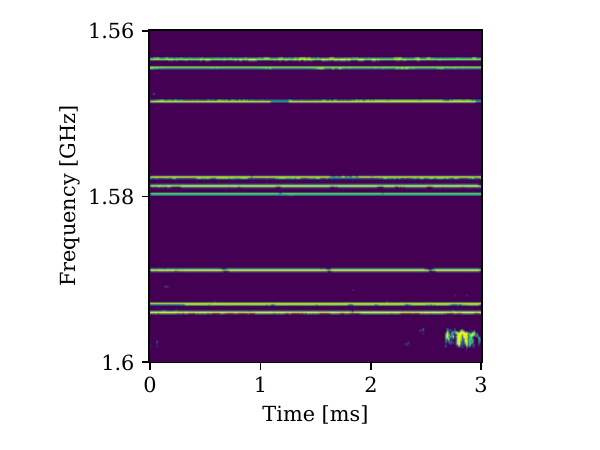}
    \end{minipage}
    \caption{Qualitative source-specific mask predictions for four overlapping interferences; top: ground truth; bottom: prediction.}
    \label{label_mask_example_4}
\end{figure}

\begin{figure}[!t]
    \centering
	\begin{minipage}[t]{0.155\linewidth}
        \centering
    	\includegraphics[trim=26 12 54 10, clip, width=1.0\linewidth]{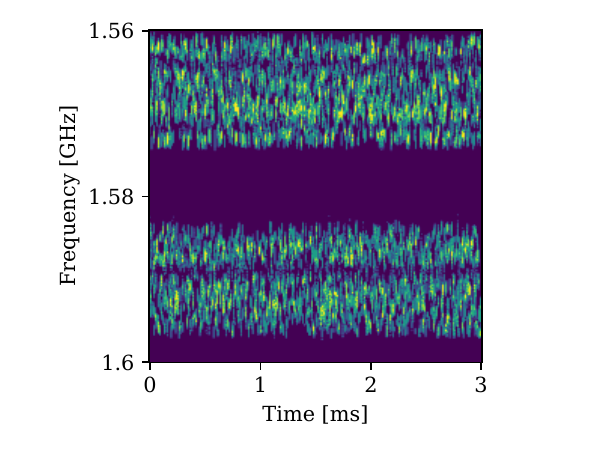}
    \end{minipage}
    \hfill
	\begin{minipage}[t]{0.155\linewidth}
        \centering
    	\includegraphics[trim=26 12 54 10, clip, width=1.0\linewidth]{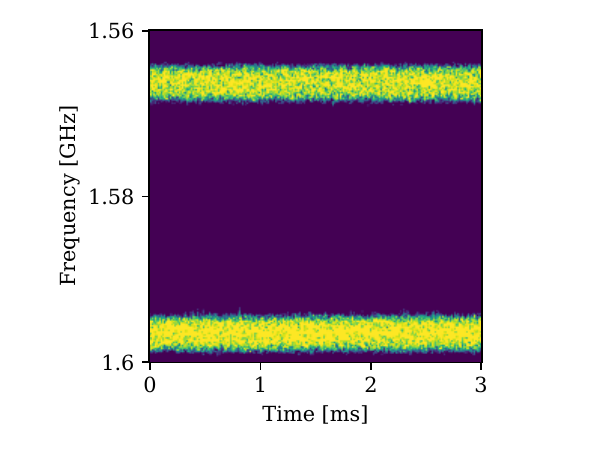}
    \end{minipage}
    \hfill
	\begin{minipage}[t]{0.155\linewidth}
        \centering
    	\includegraphics[trim=26 12 54 10, clip, width=1.0\linewidth]{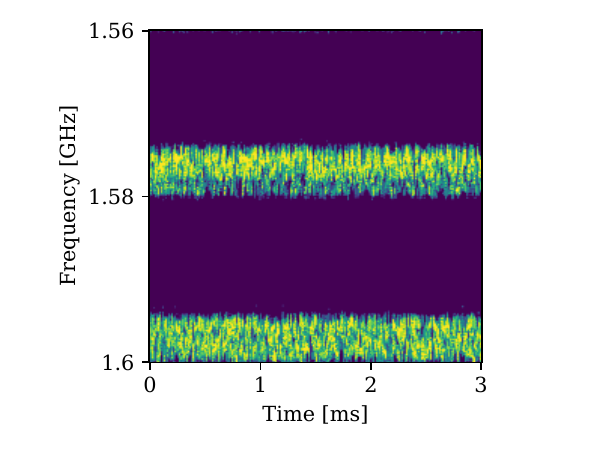}
    \end{minipage}
    \hfill
	\begin{minipage}[t]{0.155\linewidth}
        \centering
    	\includegraphics[trim=26 12 54 10, clip, width=1.0\linewidth]{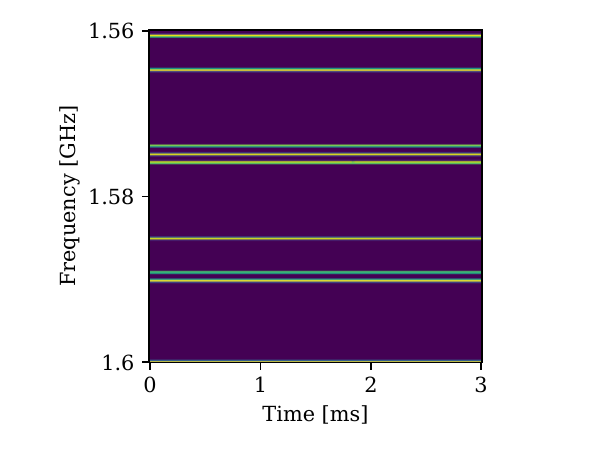}
    \end{minipage}
    \hfill
	\begin{minipage}[t]{0.155\linewidth}
        \centering
    	\includegraphics[trim=26 12 54 10, clip, width=1.0\linewidth]{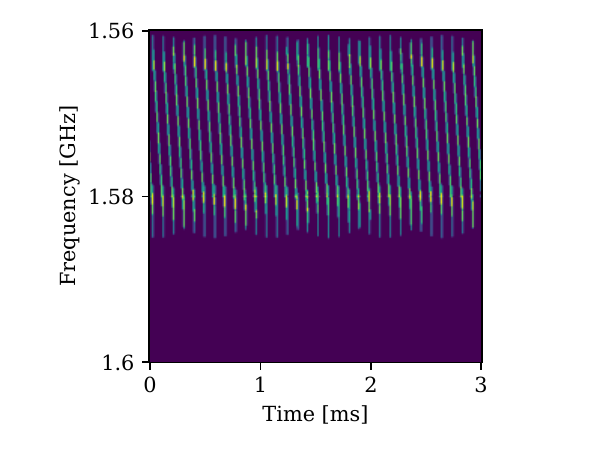}
    \end{minipage}
    \hfill
	\begin{minipage}[t]{0.155\linewidth}
        \centering
    	\includegraphics[trim=26 12 54 10, clip, width=1.0\linewidth]{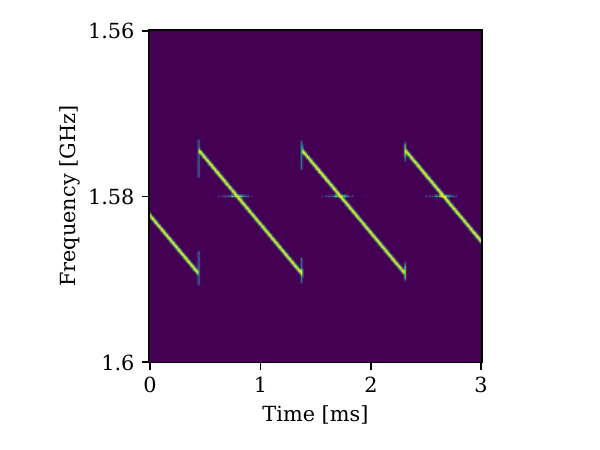}
    \end{minipage}
	\begin{minipage}[t]{0.155\linewidth}
        \centering
    	\includegraphics[trim=26 12 54 10, clip, width=1.0\linewidth]{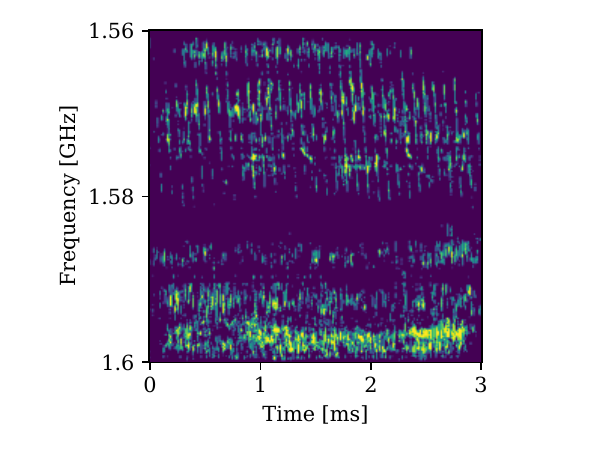}
    \end{minipage}
    \hfill
	\begin{minipage}[t]{0.155\linewidth}
        \centering
    	\includegraphics[trim=26 12 54 10, clip, width=1.0\linewidth]{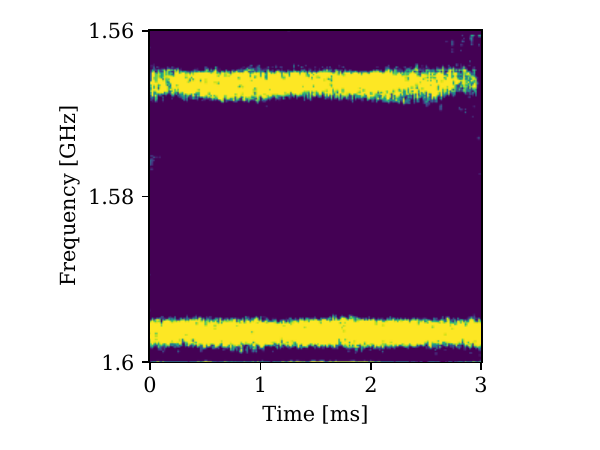}
    \end{minipage}
    \hfill
	\begin{minipage}[t]{0.155\linewidth}
        \centering
    	\includegraphics[trim=26 12 54 10, clip, width=1.0\linewidth]{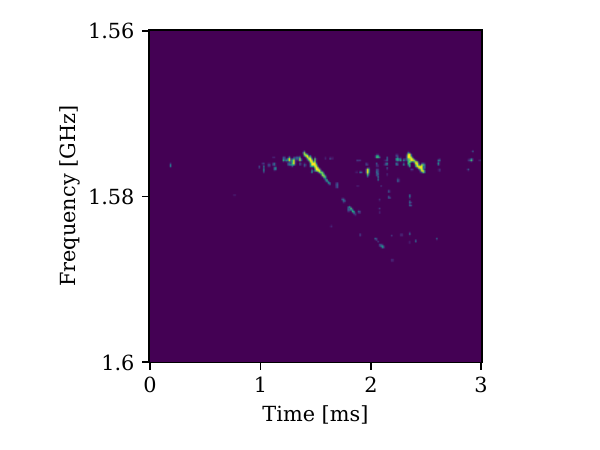}
    \end{minipage}
    \hfill
	\begin{minipage}[t]{0.155\linewidth}
        \centering
    	\includegraphics[trim=26 12 54 10, clip, width=1.0\linewidth]{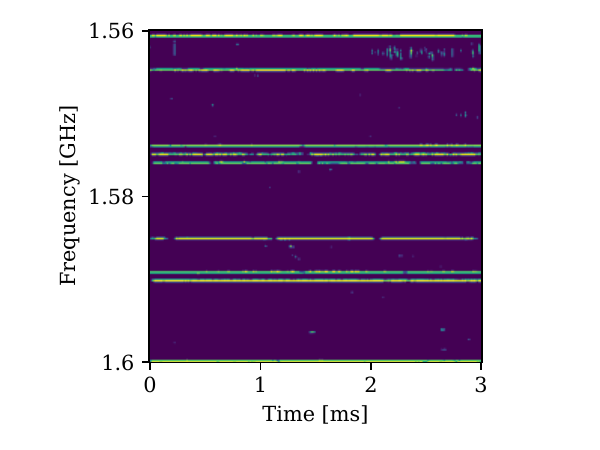}
    \end{minipage}
    \hfill
	\begin{minipage}[t]{0.155\linewidth}
        \centering
    	\includegraphics[trim=26 12 54 10, clip, width=1.0\linewidth]{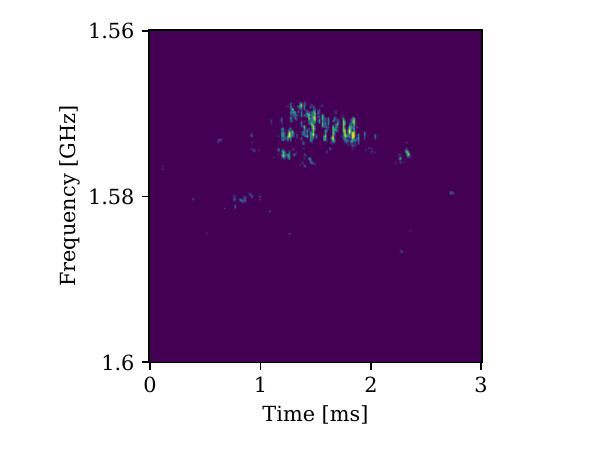}
    \end{minipage}
    \hfill
	\begin{minipage}[t]{0.155\linewidth}
        \centering
    	\includegraphics[trim=26 12 54 10, clip, width=1.0\linewidth]{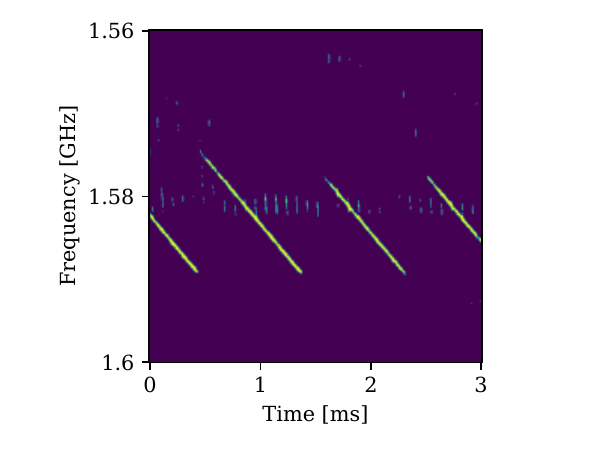}
    \end{minipage}
    \caption{Qualitative source-specific mask predictions for six overlapping interferences; top: ground truth; bottom: prediction.}
    \label{label_mask_example_6}
\end{figure}

For overlapping interference, performance decreases gradually from an
F1-score of \(93.9\%\) for two sources to \(74.6\%\) for six sources.
The corresponding reduction in recall from \(0.93\) to \(0.69\) indicates
that weak or strongly obscured sources become increasingly difficult to
detect. Nevertheless, type and BOC accuracies remain at \(95.2\%\) and
\(96.5\%\), respectively, even when all six prediction slots are occupied.
The principal challenge is therefore detecting every active source, whereas
the properties of successfully assigned sources remain accurately classified.

Figures~\ref{label_mask_example_4} and~\ref{label_mask_example_6} qualitatively
illustrate the source-specific segmentation results. Refer to App.~\ref{label_appendix_mask_predictions}, for additional results. For two and three
overlapping sources, the predicted masks closely reproduce both narrow
time-varying structures and broad stationary frequency bands. With four
sources, the model continues to separate chirped, broadband, and multitone
patterns, including structures occupying common time--frequency regions. For
five and six sources, the principal source signatures remain distinguishable,
although mask fragmentation, boundary deviations, and leakage between nearby
components become more pronounced. These observations are consistent with
the decreasing recall in Table~\ref{tab:overlap_jammers} and demonstrate that
the proposed set-based representation can separate diverse overlapping
interference structures up to the maximum supported source count.

\subsection{Inference Times}
\label{label_evaluation_inference_times}

Figure~\ref{label_figure_inference_times} shows that inference is fast and
largely independent of the selected objectives and loss functions. GPU
inference typically requires approximately \(8\)--\(9\,\mathrm{ms}\) per
sample, compared with approximately \(240\)--\(260\,\mathrm{ms}\) on the CPU.
The small variations between configurations indicate that adding prediction
elements or changing their training losses introduces negligible inference
overhead.

\begin{figure}[!t]
    \centering
	\begin{minipage}[t]{0.49\linewidth}
        \centering
    	\includegraphics[trim=10 10 10 10, clip, width=1.0\linewidth]{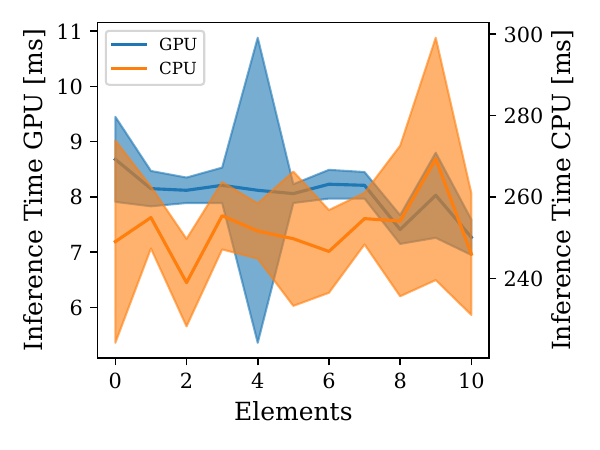}
        \subcaption{Elements.}
        \label{label_figure_inference_times1}
    \end{minipage}
    \hfill
	\begin{minipage}[t]{0.49\linewidth}
        \centering
    	\includegraphics[trim=10 10 10 10, clip, width=1.0\linewidth]{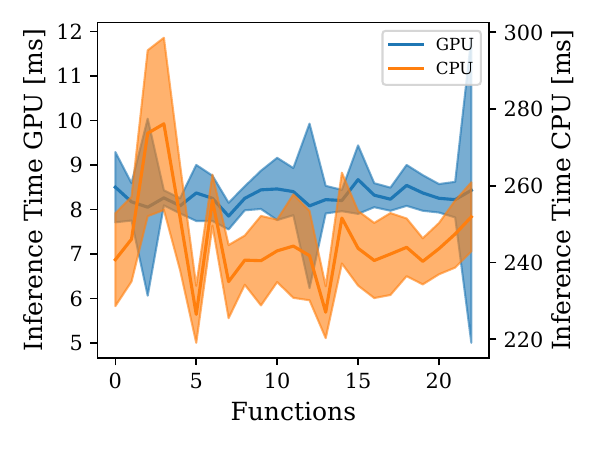}
        \subcaption{Loss functions.}
        \label{label_figure_inference_times2}
    \end{minipage}
    \caption{Inference times averaged over all test samples using an NVIDIA Tesla V100-SXM2 GPU with \(32\,\mathrm{GB}\) of VRAM and an Intel Xeon CPU for the corresponding elements (Table~\ref{tab:task_ablation}) and loss functions (Table~\ref{tab:loss_functions}).}
    \label{label_figure_inference_times}
\end{figure}

\section{Conclusion}
\label{label_conclusion}

This paper presented a unified set-based framework for jointly detecting,
segmenting, and characterizing multiple overlapping interference sources from
time--frequency representations. The framework maps each spectrogram to six
unordered source slots containing objectness, interference type, center
frequency, bandwidth, received signal strength, BOC property, and a
source-specific mask. Permutation-invariant Hungarian assignment preserves
the association between the predicted masks and source characteristics, while
a weighted multi-task objective enables flexible combinations of prediction
tasks. The method was evaluated on a Sionna ray-tracing dataset containing
realistic multipath propagation and mixtures drawn from \(94\) waveform
configurations across six interference classes. For limitations, refer to App.~\ref{label_appendix_limitations}.

The five principal findings are summarized as follows: (i) The seven-task model achieves \(74.8\%\) F1, type and BOC accuracies of \(97.3\%\) and \(98.4\%\), and errors of \(1.29\,\mathrm{MHz}\), \(2.54\,\mathrm{MHz}\), and \(3.19\,\mathrm{dB}\) for frequency, bandwidth, and signal strength. (ii) The objectness--type--mask configuration reaches \(91.1\%\) F1, revealing a trade-off between detection accuracy and comprehensive characterization. (iii) Loss improvements are not additive across tasks, while mask supervision improves both source detection and count estimation. (iv) For two to six overlapping sources, F1 decreases from \(93.9\%\) to \(74.6\%\), while type and BOC accuracies remain above \(95\%\). (v) The method is competitive with the strongest reference approach and requires only \(8\)--\(9\,\mathrm{ms}\) per sample on the GPU.

\section*{Acknowledgments}

This work has been carried out within the SAIP project, funding code 50NP2603A, sponsored by the German Federal Ministry for Economic Affairs and Climate Action (BMWK) and the German Federal Ministry for Transport (BMV), and supported by the German Space Agency at DLR, the Bundesnetzagentur (BNetzA), and the Federal Agency for Cartography and Geodesy (BKG).

\bibliography{TSP2026}
\bibliographystyle{IEEEtran}
\vspace{0.5cm}
\section*{Biography}
\vspace{-0.3cm}

\begin{IEEEbiography}[{\includegraphics[width=1in, height=1.25in, clip, keepaspectratio]{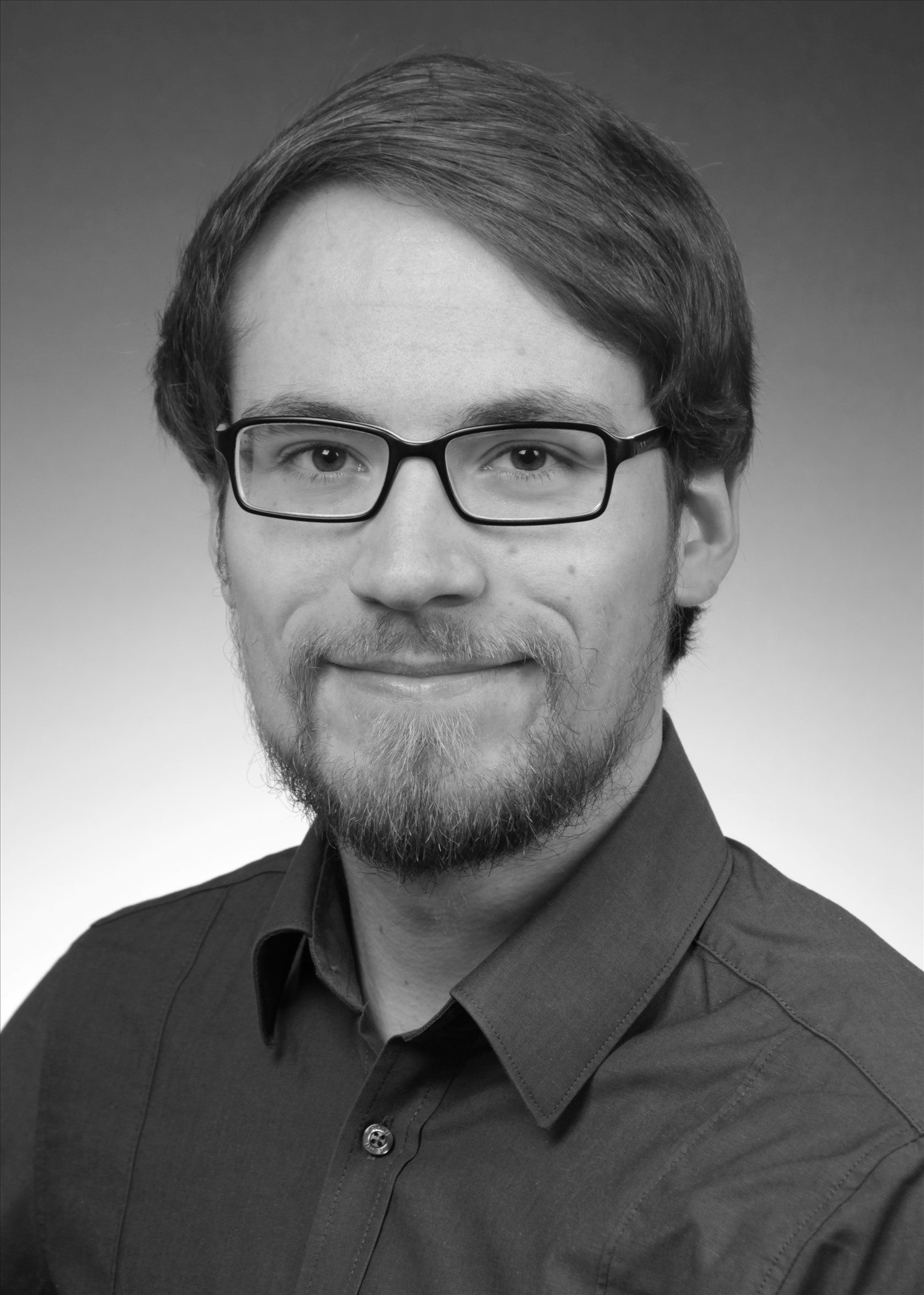}}]{Lucas Heublein} earned his M.Sc.~in Integrated Life Science from Friedrich-Alexander University (FAU) Erlangen-Nürnberg and completed his Master of Science in Computer Science in 2024. He began his association with the Hybrid Positioning \& Information Fusion group at the Fraunhofer Institute for Integrated Circuits IIS in 2020 as a student assistant and currently serves as a research assistant in the Self-Learning Systems group.
\end{IEEEbiography}

\vspace{-0.2cm}

\begin{IEEEbiography}[{\includegraphics[width=1in, height=1.25in, clip, keepaspectratio]{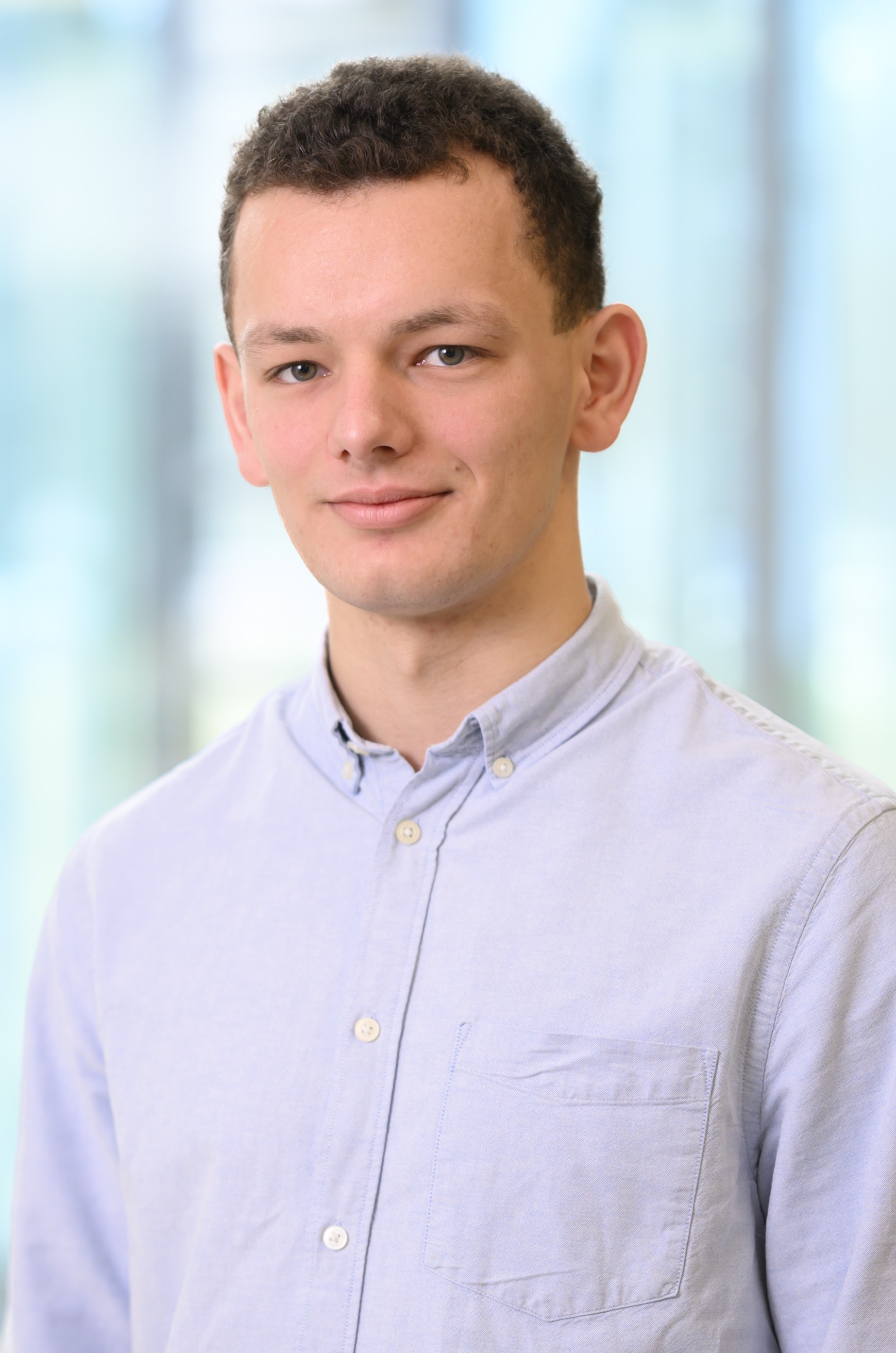}}]{Christian Wielenberg} is a research scientist at Fraunhofer IIS in Nuremberg, where he pursues a Ph.D.~on Structured Signal Separation. He received his Bachelor's degree in Computer Science from Otto-von-Guericke University Magdeburg and his Master's degree in Computer Science from FAU Erlangen-Nürnberg. His Master's thesis on GNSS jammer localization using deep learning was conducted in cooperation with Fraunhofer IIS. His research interests are RF signal processing, GNSS, wireless communication and bayesian inference.
\end{IEEEbiography}

\vspace{-0.2cm}

\begin{IEEEbiography}[{\includegraphics[width=1in, height=1.25in, clip, keepaspectratio]{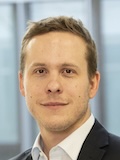}}]{Christopher Mutschler} has been a Full Professor at the University of Technology Nuremberg (UTN) since 2026 and heads the Machine Learning and Positioning Systems research group. At the same time, he is also Division Director for “Positioning and Networks” at Fraunhofer IIS in Nürnberg. His work combines fundamental research in the field of machine learning with the design and deployment of real-world wireless systems. His research focuses on reinforcement learning methods for wireless communications, with particular emphasis on wireless sensing and positioning. He investigates how learning-based approaches can enable adaptive, robust, and efficient wireless systems, especially in complex and dynamic environments.
\end{IEEEbiography}

\vspace{-0.2cm}

\begin{IEEEbiography}[{\includegraphics[width=1in, height=1.25in, clip, keepaspectratio]{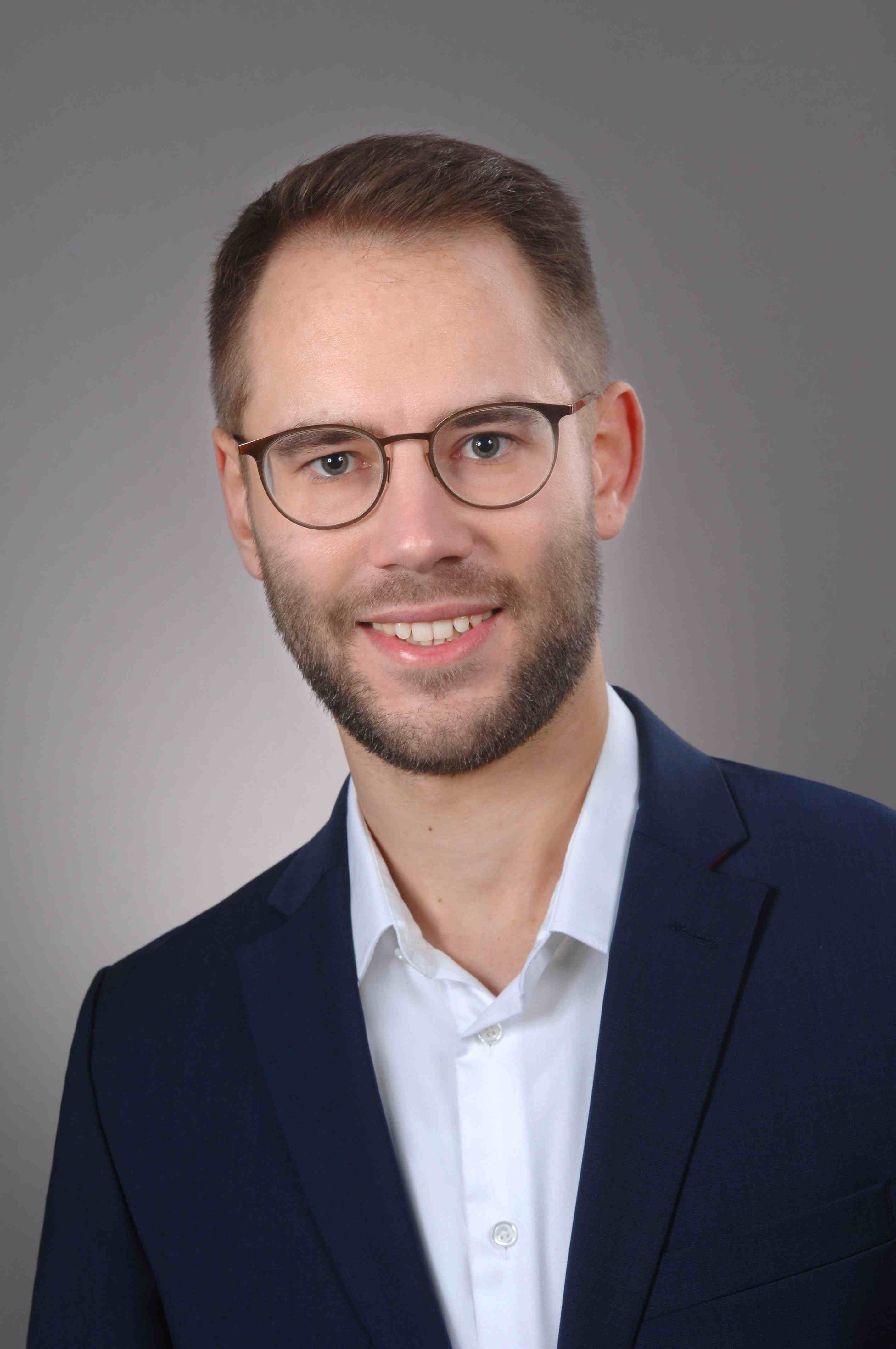}}]{Felix Ott} received his M.Sc.~degree in Computational Engineering from Friedrich-Alexander University (FAU) Erlangen-Nürnberg in 2019. He subsequently joined the Self-Learning Systems group within the Machine Intelligence department at the Fraunhofer Institute for Integrated Circuits IIS. In 2023, he was awarded his Ph.D.~from Ludwig-Maximilians University (LMU) Munich, where he conducted research in the Probabilistic Machine and Deep Learning group. He is currently serving as a postdoctoral researcher and manages the program Signals Intelligence (SIGINT). His research focuses on representation learning, domain adaptation, and few-shot learning for GNSS-based interference monitoring, such as signal classification, characterization, and source localization.
\end{IEEEbiography}

\clearpage
\section{Appendix}
\label{label_app}

\subsection{Spectral Separability and Identifiability}
\label{label_appendix_separability}

Classical time--frequency analysis establishes that the recoverability of
multiple signal components depends on their separation relative to the
frequency concentration of the analysis window. In particular, for
multicomponent amplitude- and frequency-modulated signals, components occupy
disjoint regions of the short-time Fourier transform (STFT) when their
instantaneous frequencies differ by more than twice the effective spectral
half-width of the window \cite{oberlin_fsst}. To provide a representation-level resolution reference, we analyze an
idealized single-frame case with a periodic discrete Blackman window,
exactly DFT-grid-aligned components, and finite discrete spectral supports.
This analysis concerns only the input representation and does not provide a
separation guarantee for the learned model.

Consider one STFT frame of length \(N\), containing \(K\) additive interference
components,
\begin{equation}
    x[n] = \sum_{s=1}^{K} x_s[n],
    \qquad n=0,\ldots,N-1.
\end{equation}
Let \(X_s[k]\) denote the \(N\)-point discrete Fourier transform (DFT) of
\(x_s[n]\), and let
\begin{equation}
    \mathcal{B}_s
    =
    \operatorname{supp}(X_s)
    =
    \{k:X_s[k]\neq 0\}
\end{equation}
denote its discrete spectral support. The circular distance between two
spectral supports is defined as
\begin{equation}
    d_N(\mathcal{B}_i,\mathcal{B}_j)
    =
    \min_{\substack{k\in\mathcal{B}_i\\
                    \ell\in\mathcal{B}_j}}
    \min_{q\in\mathbb{Z}}
    |k-\ell+qN|.
\end{equation}

\begin{proposition}[Idealized support separation for a periodic Blackman DFT]
\label{prop:blackman_separation}
Consider a single \(N\)-sample STFT frame. Assume that every source component
is exactly aligned with the \(N\)-point DFT grid and has an exact finite
discrete support \(\mathcal{B}_s\), such that
\[
    X_s[k]=0
    \qquad
    \text{for all } k\notin\mathcal{B}_s.
\]
Furthermore, assume that the periodic three-term Blackman window
\begin{equation}
    w[n]
    =
    a_0-a_1\cos\left(\frac{2\pi n}{N}\right)
        +a_2\cos\left(\frac{4\pi n}{N}\right),
\end{equation}
with \(a_0=0.42\), \(a_1=0.5\), and \(a_2=0.08\), is used. If
\begin{equation}
    d_N(\mathcal{B}_i,\mathcal{B}_j)\geq 5
    \qquad
    \text{for all } i\neq j,
    \label{eq:blackman_separation}
\end{equation}
then the windowed DFT supports of the interference components are pairwise
disjoint. Consequently, the noiseless power spectrogram of their mixture is
binwise additive and contains no cross-source terms:
\begin{equation}
    \left|
        \sum_{s=1}^{K}V_s[k]
    \right|^2
    =
    \sum_{s=1}^{K}|V_s[k]|^2,
\end{equation}
where \(V_s[k]\) is the windowed DFT of source \(s\).
\end{proposition}

\begin{proof}
Let \(W[k]\) denote the \(N\)-point DFT of the periodic Blackman window.
Using
\begin{equation}
    \cos\left(\frac{2\pi mn}{N}\right)
    =
    \frac{1}{2}
    \left(
        e^{\mathrm{j}2\pi mn/N}
        +
        e^{-\mathrm{j}2\pi mn/N}
    \right),
\end{equation}
the DFT of the window is
\begin{align}
    W[k]
    ={}&
    Na_0\delta_N[k]
    -\frac{Na_1}{2}
        \left(
            \delta_N[k-1]+\delta_N[k+1]
        \right)
    \nonumber\\
    &+
    \frac{Na_2}{2}
        \left(
            \delta_N[k-2]+\delta_N[k+2]
        \right),
\end{align}
where \(\delta_N[\cdot]\) denotes the periodic Kronecker delta. Therefore,
\begin{equation}
    \operatorname{supp}(W)
    \subseteq
    \mathcal{D}
    :=
    \{-2,-1,0,1,2\}.
\end{equation}

Multiplication by the window in the time domain corresponds to circular
convolution in the DFT domain. Hence,
\begin{equation}
    V_s[k]
    =
    \frac{1}{N}(X_s*W)[k],
\end{equation}
and its support satisfies
\begin{equation}
    \operatorname{supp}(V_s)
    \subseteq
    \mathcal{B}_s\oplus\mathcal{D},
\end{equation}
where \(\oplus\) denotes circular Minkowski addition.

Suppose that the windowed supports of two components \(i\) and \(j\) intersect.
Then there exist \(k_i\in\mathcal{B}_i\), \(k_j\in\mathcal{B}_j\), and
\(d_i,d_j\in\mathcal{D}\) such that
\begin{equation}
    k_i+d_i
    \equiv
    k_j+d_j
    \pmod N.
\end{equation}
It follows that
\begin{equation}
    d_N(\mathcal{B}_i,\mathcal{B}_j)
    \leq |d_j-d_i|
    \leq 4,
\end{equation}
which contradicts Eq.~\ref{eq:blackman_separation}. Thus, the windowed supports
are pairwise disjoint.

By linearity of the STFT, the mixture coefficient is
\(V[k]=\sum_s V_s[k]\). Since at most one \(V_s[k]\) is nonzero at each
frequency bin,
\begin{align}
    |V[k]|^2
    &=
    \Big|\sum_s V_s[k]\Big|^2 \\
    &=
    \sum_s |V_s[k]|^2
    +
    \sum_{i\neq j}V_i[k]V_j^{*}[k] \\
    &=
    \sum_s |V_s[k]|^2.
\end{align}
The cross-source terms therefore vanish, completing the proof.
\end{proof}

For a sampling rate \(f_{\mathrm{s}}\), the DFT-bin spacing is
\begin{equation}
    \Delta f_{\mathrm{bin}}
    =
    \frac{f_{\mathrm{s}}}{N}.
\end{equation}
With \(f_{\mathrm{s}}=40.5\,\mathrm{MHz}\) and \(N=512\), this yields
\begin{equation}
    \Delta f_{\mathrm{bin}}
    \approx 79.10\,\mathrm{kHz}.
\end{equation}
Under the exact-support and DFT-grid-alignment assumptions of
Proposition~\ref{prop:blackman_separation}, the sufficient support-separation
condition corresponds to
\begin{equation}
    \Delta f_{\min}
    =
    5\Delta f_{\mathrm{bin}}
    \approx 0.396\,\mathrm{MHz}
\end{equation}
between the nearest occupied DFT bins of two components. This value is not a
general minimum separation between source center frequencies and does not
apply directly to off-grid, broadband, or time-varying signals.

The condition is sufficient but not necessary and applies only to the
idealized single-frame setting stated in
Proposition~\ref{prop:blackman_separation}. In particular, it assumes a
periodic discrete Blackman window, exact DFT-grid alignment, and spectral
coefficients that are identically zero outside the stated discrete supports.
Off-grid tones, broadband signals, rapidly varying chirps, additive noise,
and the continuous sidelobes of practical window responses generally produce
spectral leakage and violate these assumptions. In such cases, separation is
instead characterized relative to the effective bandwidth and concentration
of the analysis window~\cite{harris_windows,oberlin_fsst}. The proposition
therefore serves only as an idealized resolution reference for the employed
time--frequency representation; it is not a performance bound for source
detection, segmentation, characterization, or the learned model itself.

\begin{proposition}[Non-identifiability under unrestricted overlap]
\label{prop:overlap_nonidentifiability}
If two interference components are unrestricted within a common signal space
and overlap completely in time and frequency, their source-specific
decomposition cannot be uniquely determined from a single-channel mixture.
\end{proposition}

\begin{proof}
Let the observed mixture be
\begin{equation}
    x=x_1+x_2.
\end{equation}
For any nonzero signal \(h\) supported in the common time--frequency region,
define
\begin{equation}
    \widetilde{x}_1=x_1+h,
    \qquad
    \widetilde{x}_2=x_2-h.
\end{equation}
Then
\begin{equation}
    \widetilde{x}_1+\widetilde{x}_2
    =
    x_1+x_2
    =
    x.
\end{equation}
Thus, infinitely many source pairs produce exactly the same observation.
Because the spectrogram is a deterministic transformation of \(x\), all such
pairs also produce the same input spectrogram. The decomposition is therefore
not identifiable unless additional structural assumptions or learned source
priors are introduced.
\end{proof}

Proposition~\ref{prop:overlap_nonidentifiability} clarifies that the proposed
network does not circumvent a general information-theoretic impossibility.
Instead, its ability to distinguish overlapping interference sources is enabled
by waveform structure and source-specific patterns learned from the training
distribution.

\subsection{Training and Source-Assignment Algorithm}
\label{label_appendix_training_algorithm}

Algorithm~\ref{alg:joint_training} summarizes the implemented training
procedure. Each complex-valued IQ recording is transformed into a two-sided
power spectrogram using a 512-sample Blackman window and an overlap of
256 samples. After centering the frequency axis, the power is converted to
decibels and standardized independently for each recording. Ground-truth
sources outside the selected observation bandwidth are discarded, while the
remaining annotations are padded to the maximum number of \(S=6\) sources.
Padding entries are ignored during matching and loss computation.

For every mini-batch, the network jointly produces six characterization
vectors and six full-resolution mask-logit maps. A pairwise cost matrix is
constructed between the predicted slots and the valid ground-truth sources,
and the Hungarian algorithm determines the minimum-cost one-to-one assignment.
During the first \(10\%\) of the training epochs, mask costs are excluded from
the assignment to stabilize the prediction heads. Mask supervision remains
active, but its assignment is initially determined from objectness,
classification, and the enabled characterization attributes. After this
warm-up period, mask BCE and Dice costs are added to the matching cost.
Objectness is optimized over all slots, whereas the remaining objectives are
evaluated only for matched source--slot pairs.

\begin{algorithm}[!t]
\caption{Joint set-based segmentation and characterization}
\label{alg:joint_training}
\begin{algorithmic}[1]
\Require Training set \(\mathcal{D}\), epochs \(E\), slots \(S=6\),
         task weights \(\boldsymbol{\lambda}\)
\Require Mask-matching warm-up \(E_{\mathrm{w}}=\lfloor0.1E\rfloor\)
\State Initialize model parameters \(\theta\), optimizer, and scheduler
\For{\(e=1,\ldots,E\)}
    \If{\(e\leq E_{\mathrm{w}}\)}
        \State Set mask-matching weights to zero
    \Else
        \State Enable mask BCE and Dice matching costs
    \EndIf
    \For{each mini-batch
          \(\{(x_b,\mathcal{Y}_b,\mathcal{M}_b)\}_{b=1}^{B}\)}
        \For{\(b=1,\ldots,B\)}
            \State Compute the Blackman-windowed spectrogram
            \[
                X_b \gets
                \operatorname{Standardize}
                \left(
                    10\log_{10}
                    \left|
                        \operatorname{STFT}_{512,256}(x_b)
                    \right|^2
                \right)
            \]
        \EndFor
        \State Predict slot attributes and masks:
        \[
            (\widehat{\mathcal{Y}},
             \widehat{\mathcal{M}})
             \gets F_{\theta}(X)
        \]
        \State Initialize mini-batch loss \(\mathcal{L}\gets0\)
        \For{\(b=1,\ldots,B\)}
            \State Remove padded entries from \(\mathcal{Y}_b\)
            \State Construct pairwise assignment-cost matrix
            \[
                \mathbf{C}_b
                \gets
                \Call{MatchingCost}{
                    \widehat{\mathcal{Y}}_b,
                    \widehat{\mathcal{M}}_b,
                    \mathcal{Y}_b,
                    \mathcal{M}_b}
            \]
            \State Obtain assignment
            \[
                \pi_b\gets\Call{Hungarian}{\mathbf{C}_b}
            \]
            \State Construct objectness targets for all \(S\) slots
            \State Compute objectness loss over all slots
            \State Compute type, frequency, signal-strength,
            \State \,\,\,bandwidth, BOC, and mask losses over
            \State \,\,\,pairs in $\pi_b$
            \State Accumulate the weighted multi-task objective
            \[
                \mathcal{L}\gets\mathcal{L}
                +\sum_{q}
                \lambda_q\mathcal{L}_{q}^{(b)}
            \]%
        \EndFor
        \State \(\mathcal{L}\gets\mathcal{L}/B\)
        \State Update \(\theta\) by backpropagation
    \EndFor
    \State Update the learning-rate scheduler
\EndFor
\State \Return trained parameters \(\theta\)
\end{algorithmic}
\end{algorithm}

\subsection{Evaluation Metrics}
\label{label_appendix_evaluation_metrics}

This section formally defines the metrics used throughout the evaluation.
Let \(N\) denote the number of evaluated spectrograms, \(K_n\) the number of
ground-truth sources in sample \(n\), and \(S=6\) the number of available
prediction slots. The ground-truth sources are denoted by
\(\mathcal{Y}_n=\{\mathbf{y}_{n,k}\}_{k=1}^{K_n}\).

\paragraph{Source Activation}
For slot \(s\), the objectness probability is obtained as
\begin{equation}
    p^{\mathrm{O}}_{n,s}
    =
    \sigma\!\left(\widehat{o}_{n,s}\right).
\end{equation}
A slot is considered active if
\begin{equation}
    p^{\mathrm{O}}_{n,s}\geq\tau_{\mathrm{O}},
    \qquad
    \tau_{\mathrm{O}}=0.5.
\end{equation}
The set of active slots and the corresponding predicted source count are
\begin{equation}
    \widehat{\mathcal{S}}_n
    =
    \left\{
        s\in\{1,\ldots,S\}:
        p^{\mathrm{O}}_{n,s}\geq\tau_{\mathrm{O}}
    \right\},
    \qquad
    \widehat{K}_n
    =
    \left|\widehat{\mathcal{S}}_n\right|.
\end{equation}

\paragraph{Evaluation Assignment}
Because neither the prediction slots nor the ground-truth sources possess a
predefined order, the active predictions are assigned to the ground-truth
sources before the metrics are computed. This evaluation assignment is
distinct from the training assignment in
Eqs.~\ref{eq:hungarian_cost} and~\ref{eq:hungarian_assignment}, since it
is performed only over slots retained by the objectness threshold. For active prediction \(s\) and ground-truth source \(k\), the evaluation
cost is
\begin{align}
    D_{n,s,k}
    ={}&
    20\,
    \mathbb{I}
    \left[
        \widehat{c}_{n,s}\neq c_{n,k}
    \right]
    +
    \eta_{\mathrm{CF}}
    \left|
        \widehat{f}_{n,s}-f_{n,k}
    \right|
    \nonumber\\
    &+
    \eta_{\mathrm{BW}}
    \left|
        \widehat{b}_{n,s}-b_{n,k}
    \right|
    +
    0.5\,\eta_{\mathrm{SS}}
    \left|
        \widehat{r}_{n,s}-r_{n,k}
    \right|,
    \label{eq:evaluation_cost}
\end{align}
where \(\widehat{c}_{n,s}\), \(\widehat{f}_{n,s}\),
\(\widehat{b}_{n,s}\), and \(\widehat{r}_{n,s}\) denote the predicted type,
center frequency, bandwidth, and signal strength. The corresponding
ground-truth quantities are denoted without hats. The indicator
\(\eta_j\in\{0,1\}\) specifies whether objective \(j\) is enabled in the
evaluated configuration. Consequently, disabled regression tasks do not
contribute to the evaluation assignment.

Applying the Hungarian algorithm to \(\mathbf{D}_n\) yields the set of
one-to-one prediction--target assignments
\begin{equation}
    \mathcal{M}_n
    =
    \operatorname{Hungarian}
    \left(\mathbf{D}_n\right),
\end{equation}
with
\begin{equation}
    \left|\mathcal{M}_n\right|
    =
    \min\!\left(\widehat{K}_n,K_n\right).
\end{equation}
Objectness is not included in Eq.~\ref{eq:evaluation_cost}, because it has
already been used to construct the active prediction set. BOC predictions
and segmentation masks do not contribute to the evaluation assignment.

\paragraph{Valid Source Detection}
An assigned pair \((s,k)\in\mathcal{M}_n\) is considered a valid source
detection if its predicted type is correct and all enabled continuous
attributes satisfy their respective tolerances. The corresponding indicator
is
\begin{align}
    z_{n,s,k}
    ={}&
    \mathbb{I}
    \left[
        \widehat{c}_{n,s}=c_{n,k}
    \right]
    \nonumber\\
    &\times
    \mathbb{I}
    \left[
        \eta_{\mathrm{CF}}=0
        \ \lor\
        \left|
            \widehat{f}_{n,s}-f_{n,k}
        \right|
        \leq\epsilon_{\mathrm{CF}}
    \right]
    \nonumber\\
    &\times
    \mathbb{I}
    \left[
        \eta_{\mathrm{BW}}=0
        \ \lor\
        \left|
            \widehat{b}_{n,s}-b_{n,k}
        \right|
        \leq\epsilon_{\mathrm{BW}}
    \right]
    \nonumber\\
    &\times
    \mathbb{I}
    \left[
        \eta_{\mathrm{SS}}=0
        \ \lor\
        \left|
            \widehat{r}_{n,s}-r_{n,k}
        \right|
        \leq\epsilon_{\mathrm{SS}}
    \right],
    \label{eq:evaluation_valid_match}
\end{align}
where
\begin{equation}
    \epsilon_{\mathrm{CF}}=5\,\mathrm{MHz},
    \qquad
    \epsilon_{\mathrm{BW}}=5\,\mathrm{MHz},
    \qquad
    \epsilon_{\mathrm{SS}}=10\,\mathrm{dB}.
\end{equation}
BOC correctness and mask overlap are evaluated separately and are not
required for an assigned pair to constitute a valid detection.

A valid assigned pair contributes one true positive. An assigned pair that
fails at least one condition in
Eq.~\ref{eq:evaluation_valid_match} contributes one false positive and one
false negative, since the prediction does not correctly represent its
assigned ground-truth source. Each unmatched active prediction contributes
one additional false positive, whereas each unmatched ground-truth source
contributes one additional false negative. Aggregated over the complete test
set, the counts are
\begin{align}
    \mathrm{TP}
    &=
    \sum_{n=1}^{N}
    \sum_{(s,k)\in\mathcal{M}_n}
    z_{n,s,k},
    \\
    \mathrm{FP}
    &=
    \sum_{n=1}^{N}
    \left[
        \widehat{K}_n
        -
        \sum_{(s,k)\in\mathcal{M}_n}
        z_{n,s,k}
    \right],
    \\
    \mathrm{FN}
    &=
    \sum_{n=1}^{N}
    \left[
        K_n
        -
        \sum_{(s,k)\in\mathcal{M}_n}
        z_{n,s,k}
    \right].
\end{align}
Thus, if no slot is active, all \(K_n\) ground-truth sources are counted as
false negatives. Predictions exceeding the number of ground-truth sources
are counted as false positives.

\paragraph{Precision, Recall, and F1-Score}
Precision and recall are micro-averaged over all interference sources in the
test set:
\begin{equation}
    \mathrm{Precision}
    =
    \frac{\mathrm{TP}}
         {\mathrm{TP}+\mathrm{FP}},
    \qquad
    \mathrm{Recall}
    =
    \frac{\mathrm{TP}}
         {\mathrm{TP}+\mathrm{FN}}.
\end{equation}
The F1-score is their harmonic mean:
\begin{equation}
    \mathrm{F1}
    =
    2 \cdot
    \frac{
        \mathrm{Precision}\,
        \mathrm{Recall}
    }{
        \mathrm{Precision}
        +
        \mathrm{Recall}
    }.
\end{equation}
Precision and recall are reported as scores in \([0,1]\), whereas the
F1-score is multiplied by \(100\) and reported in percent.

\paragraph{Count Accuracy}
Count accuracy measures the fraction of test samples for which the number of
active predictions exactly equals the number of ground-truth sources:
\begin{equation}
    \mathrm{Acc}_{\mathrm{count}}
    =
    \frac{1}{N}
    \sum_{n=1}^{N}
    \mathbb{I}
    \left[
        \widehat{K}_n=K_n
    \right].
\end{equation}
Count accuracy considers only the number of predicted sources. A sample can
therefore have a correct count even if one or more assigned attributes are
incorrect.

\paragraph{Type and BOC Accuracy}
Let
\begin{equation}
    N_{\mathrm{M}}
    =
    \sum_{n=1}^{N}
    \left|\mathcal{M}_n\right|
\end{equation}
denote the total number of Hungarian-assigned prediction--target pairs. Type
accuracy is defined as
\begin{equation}
    \mathrm{Acc}_{\mathrm{type}}
    =
    \frac{1}{N_{\mathrm{M}}}
    \sum_{n=1}^{N}
    \sum_{(s,k)\in\mathcal{M}_n}
    \mathbb{I}
    \left[
        \widehat{c}_{n,s}=c_{n,k}
    \right].
\end{equation}

The predicted BOC label is obtained by thresholding its probability at
\(0.5\):
\begin{equation}
    \widehat{q}_{n,s}
    =
    \mathbb{I}
    \left[
        \sigma\!\left(\widehat{q}^{\,\mathrm{logit}}_{n,s}\right)
        \geq0.5
    \right].
\end{equation}
BOC accuracy is then
\begin{equation}
    \mathrm{Acc}_{\mathrm{BOC}}
    =
    \frac{1}{N_{\mathrm{M}}}
    \sum_{n=1}^{N}
    \sum_{(s,k)\in\mathcal{M}_n}
    \mathbb{I}
    \left[
        \widehat{q}_{n,s}=q_{n,k}
    \right].
\end{equation}
Both accuracies are conditional on a prediction being included in the
Hungarian assignment. Missed ground-truth sources and unmatched predictions
do not enter their denominators; their effects are captured by recall and
precision, respectively. Consequently, type and BOC accuracies can remain
high even when not all sources are detected and must therefore be interpreted
jointly with the detection metrics.

\paragraph{Continuous-Parameter Errors}
Center-frequency, bandwidth, and signal-strength errors are reported as mean
absolute errors over all assigned prediction--target pairs:
\begin{align}
    E_{\mathrm{CF}}
    &=
    \frac{1}{N_{\mathrm{M}}}
    \sum_{n=1}^{N}
    \sum_{(s,k)\in\mathcal{M}_n}
    \left|
        \widehat{f}_{n,s}-f_{n,k}
    \right|,
    \\
    E_{\mathrm{BW}}
    &=
    \frac{1}{N_{\mathrm{M}}}
    \sum_{n=1}^{N}
    \sum_{(s,k)\in\mathcal{M}_n}
    \left|
        \widehat{b}_{n,s}-b_{n,k}
    \right|,
    \\
    E_{\mathrm{SS}}
    &=
    \frac{1}{N_{\mathrm{M}}}
    \sum_{n=1}^{N}
    \sum_{(s,k)\in\mathcal{M}_n}
    \left|
        \widehat{r}_{n,s}-r_{n,k}
    \right|.
\end{align}
The respective units are \(\mathrm{MHz}\), \(\mathrm{MHz}\), and
\(\mathrm{dB}\). These errors include every assigned pair, irrespective of
whether it satisfies the detection tolerances in
Eq.~\ref{eq:evaluation_valid_match}. Unmatched sources and unmatched
predictions do not enter these conditional regression metrics. Accordingly,
the regression errors characterize the parameter quality of assigned
predictions rather than the end-to-end ability to recover every
ground-truth source.

\paragraph{Segmentation Metrics}
When mask performance is evaluated quantitatively, the predicted mask
probabilities are binarized at \(0.5\). If necessary, the ground-truth mask
is resized to the resolution of the predicted mask using nearest-neighbor
interpolation. For an assigned pair \((s,k)\), the mask intersection over
union is
\begin{equation}
    \operatorname{IoU}_{n,s,k}
    =
    \frac{
        \left|
            \widehat{M}_{n,s}
            \cap
            M_{n,k}
        \right|
    }{
        \left|
            \widehat{M}_{n,s}
            \cup
            M_{n,k}
        \right|
    }.
\end{equation}
If both masks are empty, their IoU is defined as one. The mean matched-mask
IoU is
\begin{equation}
    \overline{\operatorname{IoU}}
    =
    \frac{1}{N_{\mathrm{M}}}
    \sum_{n=1}^{N}
    \sum_{(s,k)\in\mathcal{M}_n}
    \operatorname{IoU}_{n,s,k}.
\end{equation}
In addition, mask accuracy at an IoU threshold of \(0.5\) is defined as
\begin{equation}
    \mathrm{Acc}_{\mathrm{mask}}
    =
    \frac{1}{N_{\mathrm{M}}}
    \sum_{n=1}^{N}
    \sum_{(s,k)\in\mathcal{M}_n}
    \mathbb{I}
    \left[
        \operatorname{IoU}_{n,s,k}\geq0.5
    \right].
\end{equation}
As with the characterization metrics, masks belonging to missed sources do
not enter these conditional averages; their absence is reflected through
the false-negative count and the resulting recall.

\paragraph{Stratification by Source Count and Overlap}
For the results in Table~\ref{tab:overlap_jammers}, the preceding metrics are
recomputed separately for samples containing a fixed number of ground-truth
sources. A sample is considered overlapping if at least two ground-truth
masks share one or more active time--frequency bins:
\begin{equation}
    \exists\,k\neq l:
    \sum_{u=1}^{N_f}
    \sum_{v=1}^{N_t}
    M_{n,k,u,v}M_{n,l,u,v}>0.
\end{equation}
Otherwise, the sample is categorized as non-overlapping. If no prediction--target assignment is available, the conditional type,
BOC, regression, and segmentation metrics are undefined and are reported as
``--''. Missing sources are not assigned artificial parameter errors; they
affect the detection metrics through the false-negative count.

\subsection{Additional Multi-Interference Characterization Results}
\label{label_appendix_mask_predictions}

\begin{figure}[!t]
    \centering
	\begin{minipage}[t]{0.24\linewidth}
        \centering
    	\includegraphics[trim=26 12 54 10, clip, width=1.0\linewidth]{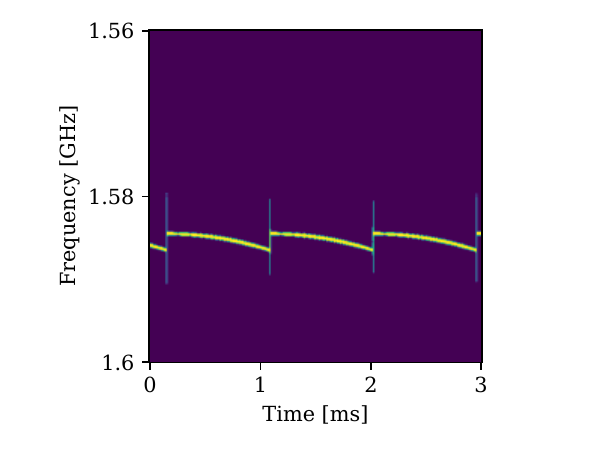}
    \end{minipage}
    \hfill
	\begin{minipage}[t]{0.24\linewidth}
        \centering
    	\includegraphics[trim=26 12 54 10, clip, width=1.0\linewidth]{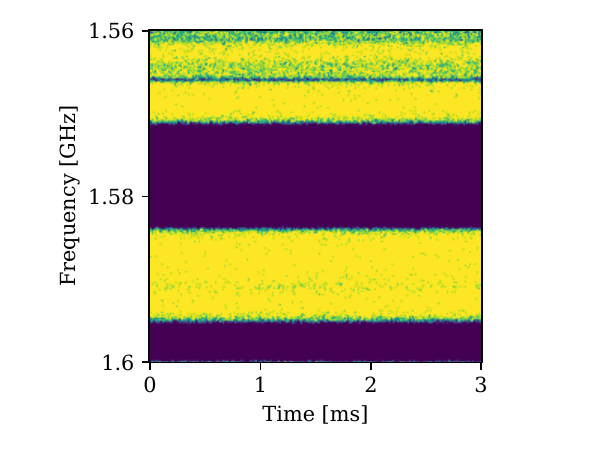}
    \end{minipage}
    \hfill
	\begin{minipage}[t]{0.24\linewidth}
        \centering
    	\includegraphics[trim=26 12 54 10, clip, width=1.0\linewidth]{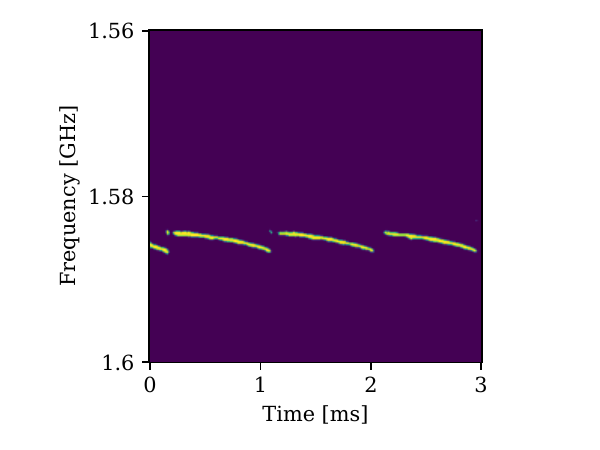}
    \end{minipage}
    \hfill
	\begin{minipage}[t]{0.24\linewidth}
        \centering
    	\includegraphics[trim=26 12 54 10, clip, width=1.0\linewidth]{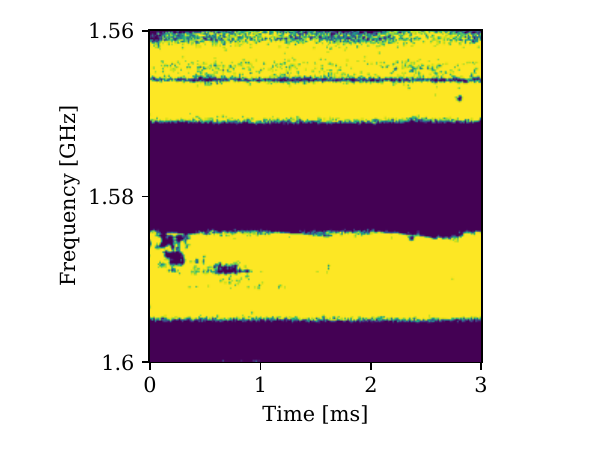}
    \end{minipage}
    \caption{Qualitative source-specific mask predictions for two overlapping interferences; 1--2: ground truth; 3--4: prediction.}
    \label{label_app_mask_example_2}
\end{figure}

\begin{figure}[!t]
    \centering
	\begin{minipage}[t]{0.32\linewidth}
        \centering
    	\includegraphics[trim=26 12 54 10, clip, width=1.0\linewidth]{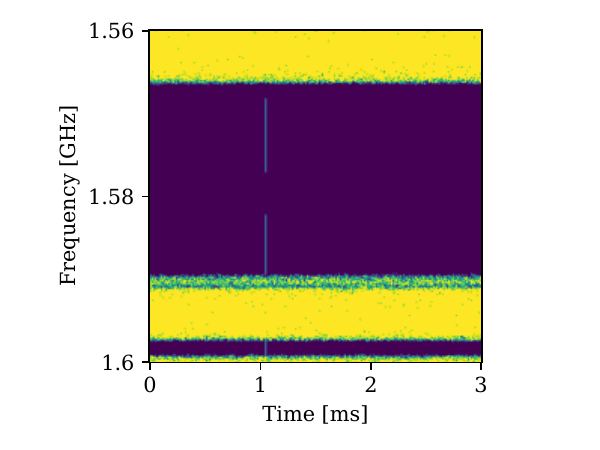}
    \end{minipage}
    \hfill
	\begin{minipage}[t]{0.32\linewidth}
        \centering
    	\includegraphics[trim=26 12 54 10, clip, width=1.0\linewidth]{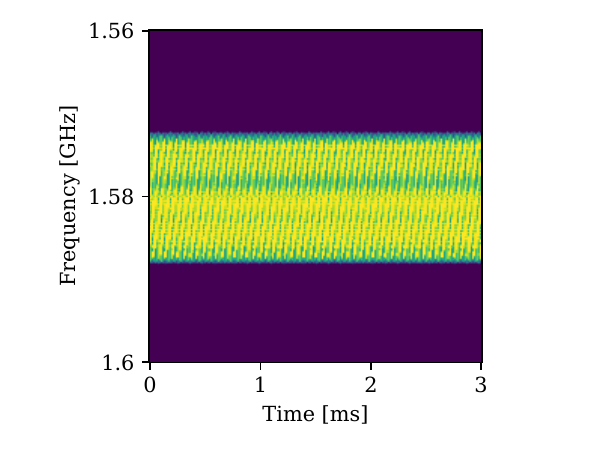}
    \end{minipage}
    \hfill
	\begin{minipage}[t]{0.32\linewidth}
        \centering
    	\includegraphics[trim=26 12 54 10, clip, width=1.0\linewidth]{images/mask_prediction/lab_3_2.pdf}
    \end{minipage}
	\begin{minipage}[t]{0.32\linewidth}
        \centering
    	\includegraphics[trim=26 12 54 10, clip, width=1.0\linewidth]{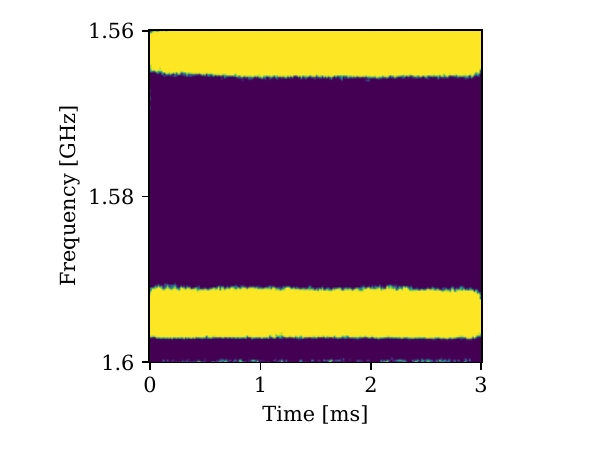}
    \end{minipage}
    \hfill
	\begin{minipage}[t]{0.32\linewidth}
        \centering
    	\includegraphics[trim=26 12 54 10, clip, width=1.0\linewidth]{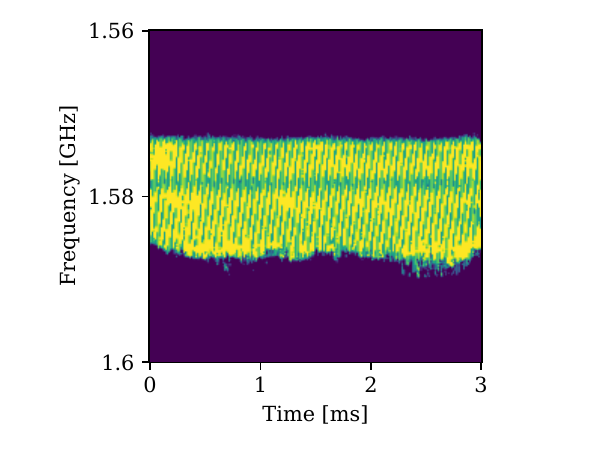}
    \end{minipage}
    \hfill
	\begin{minipage}[t]{0.32\linewidth}
        \centering
    	\includegraphics[trim=26 12 54 10, clip, width=1.0\linewidth]{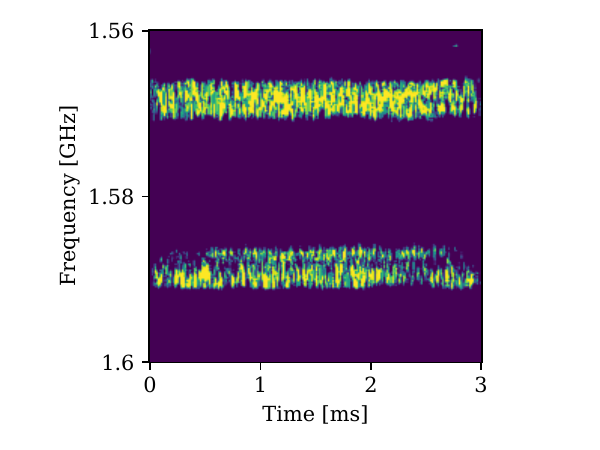}
    \end{minipage}
    \caption{Qualitative source-specific mask predictions for three overlapping interferences; top: ground truth; bottom: prediction.}
    \label{label_app_mask_example_3}
\end{figure}

\begin{figure*}[!t]
    \centering
	\begin{minipage}[t]{0.24\linewidth}
        \centering
    	\includegraphics[trim=26 12 52 10, clip, width=1.0\linewidth]{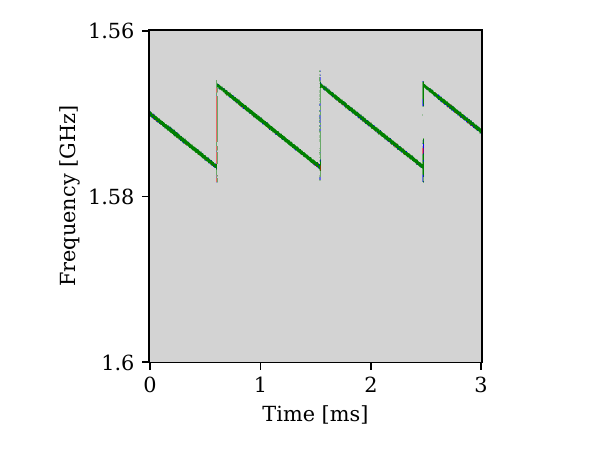}
        \subcaption{Example 1.}
        \label{label_iou_evaluation1}
    \end{minipage}
    \hfill
	\begin{minipage}[t]{0.24\linewidth}
        \centering
    	\includegraphics[trim=26 12 52 10, clip, width=1.0\linewidth]{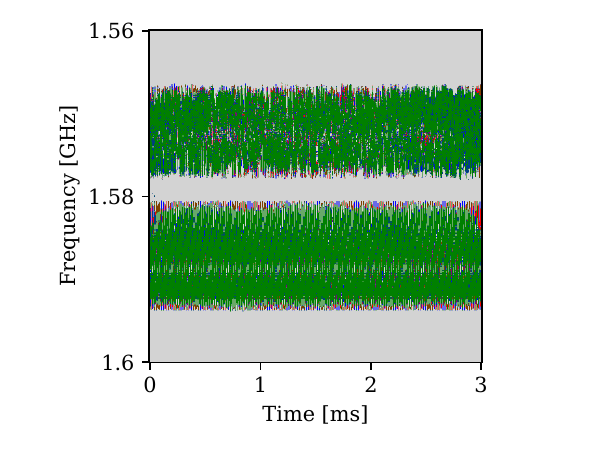}
        \subcaption{Example 2.}
        \label{label_iou_evaluation2}
    \end{minipage}
    \hfill
	\begin{minipage}[t]{0.24\linewidth}
        \centering
    	\includegraphics[trim=26 12 52 10, clip, width=1.0\linewidth]{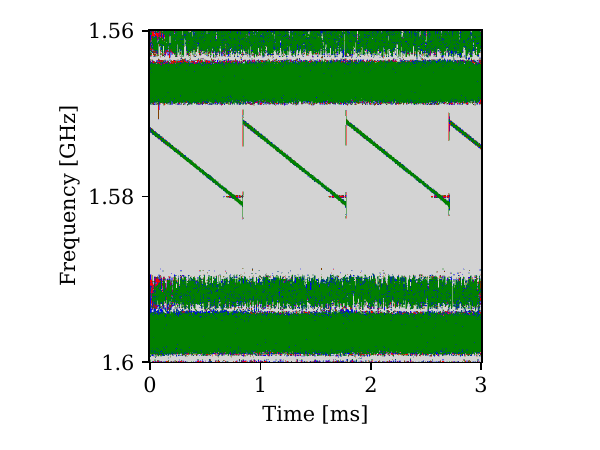}
        \subcaption{Example 3.}
        \label{label_iou_evaluation3}
    \end{minipage}
    \hfill
	\begin{minipage}[t]{0.24\linewidth}
        \centering
    	\includegraphics[trim=26 12 52 10, clip, width=1.0\linewidth]{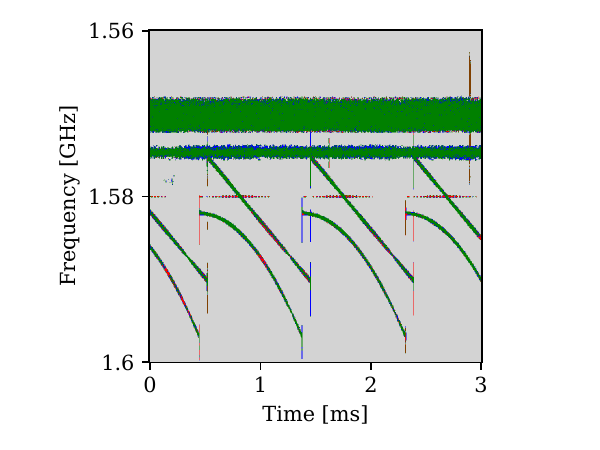}
        \subcaption{Example 4.}
        \label{label_iou_evaluation4}
    \end{minipage}
    \caption{Qualitative comparison of predicted and ground-truth segmentation masks for four exemplary spectrograms of interference sources. Green regions indicate agreement between prediction and ground truth, whereas red and blue regions represent segmentation disagreements. The principal time--frequency structures remain accurately recovered as the source count increases, although increasing spectral density and overlap reduce the IoU from \(0.890\) to \(0.795\).}
    \label{label_iou_evaluation}
\end{figure*}

Figures~\ref{label_app_mask_example_2}, \ref{label_app_mask_example_3}, \ref{label_app_mask_example_4}, \ref{label_app_mask_example_5}, and~\ref{label_app_mask_example_6}
provide additional qualitative results for the source-specific mask predictions.
The ground-truth panels show binary time--frequency masks, whereas the
prediction panels visualize the corresponding probabilities in \([0,1]\).
Since the prediction slots possess no intrinsic order, the predicted masks are
displayed according to their Hungarian assignment to the ground-truth
sources. Consequently, their original slot indices do not necessarily match
the order of the annotations.

Figure~\ref{label_app_mask_example_2} shows two sources with substantially
different spectral structures. The first mask contains a narrow, periodically
interrupted, frequency-varying component, while the second comprises broad
stationary frequency bands. The predictions recover both the narrow trajectory
and the extent of the broadband regions. Minor deviations occur around the
boundaries of the broad mask and at the discontinuities of the narrow
component, where the predictions become slightly fragmented. Nevertheless,
the two sources remain clearly separated despite their simultaneous
occurrence.

\begin{figure}[!t]
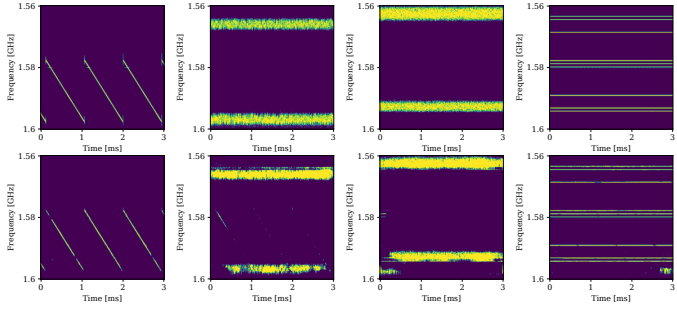

    \centering
	\begin{minipage}[t]{0.24\linewidth}
        \centering
    	\includegraphics[trim=26 12 54 10, clip, width=1.0\linewidth]{images/mask_prediction/lab_4_0.pdf}
    \end{minipage}
    \hfill
	\begin{minipage}[t]{0.24\linewidth}
        \centering
    	\includegraphics[trim=26 12 54 10, clip, width=1.0\linewidth]{images/mask_prediction/lab_4_1.pdf}
    \end{minipage}
    \hfill
	\begin{minipage}[t]{0.24\linewidth}
        \centering
    	\includegraphics[trim=26 12 54 10, clip, width=1.0\linewidth]{images/mask_prediction/lab_4_2.pdf}
    \end{minipage}
    \hfill
	\begin{minipage}[t]{0.24\linewidth}
        \centering
    	\includegraphics[trim=26 12 54 10, clip, width=1.0\linewidth]{images/mask_prediction/lab_4_3.pdf}
    \end{minipage}
	\begin{minipage}[t]{0.24\linewidth}
        \centering
    	\includegraphics[trim=26 12 54 10, clip, width=1.0\linewidth]{images/mask_prediction/pred_4_0.pdf}
    \end{minipage}
    \hfill
	\begin{minipage}[t]{0.24\linewidth}
        \centering
    	\includegraphics[trim=26 12 54 10, clip, width=1.0\linewidth]{images/mask_prediction/pred_4_3.pdf}
    \end{minipage}
    \hfill
	\begin{minipage}[t]{0.24\linewidth}
        \centering
    	\includegraphics[trim=26 12 54 10, clip, width=1.0\linewidth]{images/mask_prediction/pred_4_2.pdf}
    \end{minipage}
    \hfill
	\begin{minipage}[t]{0.24\linewidth}
        \centering
    	\includegraphics[trim=26 12 54 10, clip, width=1.0\linewidth]{images/mask_prediction/pred_4_1.pdf}
    \end{minipage}
    \caption{Qualitative source-specific mask predictions for four overlapping interferences; top: ground truth; bottom: prediction.}
    \label{label_app_mask_example_4}
\end{figure}

\begin{figure}[!t]
    \centering
	\begin{minipage}[t]{0.19\linewidth}
        \centering
    	\includegraphics[trim=26 12 54 10, clip, width=1.0\linewidth]{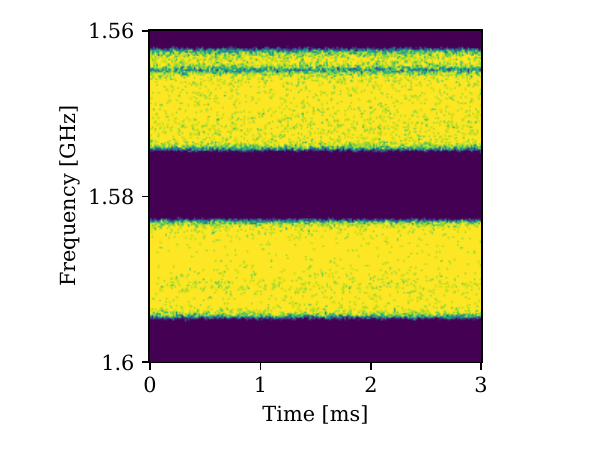}
    \end{minipage}
    \hfill
	\begin{minipage}[t]{0.19\linewidth}
        \centering
    	\includegraphics[trim=26 12 54 10, clip, width=1.0\linewidth]{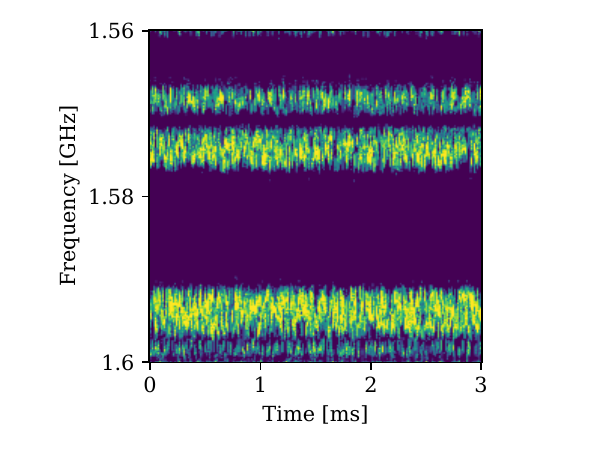}
    \end{minipage}
    \hfill
	\begin{minipage}[t]{0.19\linewidth}
        \centering
    	\includegraphics[trim=26 12 54 10, clip, width=1.0\linewidth]{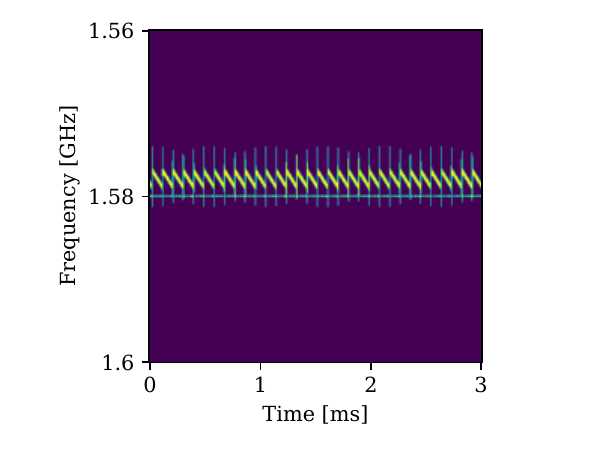}
    \end{minipage}
    \hfill
	\begin{minipage}[t]{0.19\linewidth}
        \centering
    	\includegraphics[trim=26 12 54 10, clip, width=1.0\linewidth]{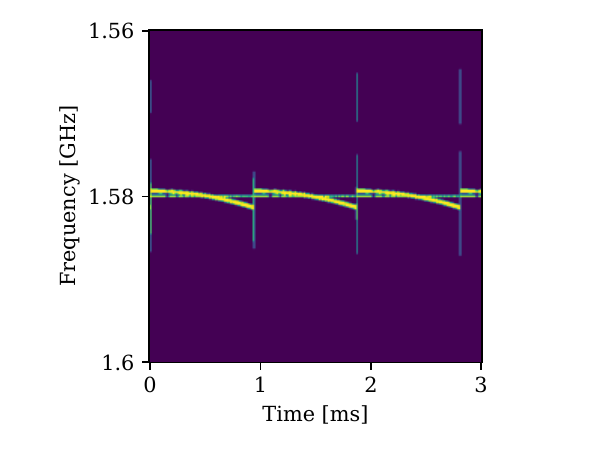}
    \end{minipage}
    \hfill
	\begin{minipage}[t]{0.19\linewidth}
        \centering
    	\includegraphics[trim=26 12 54 10, clip, width=1.0\linewidth]{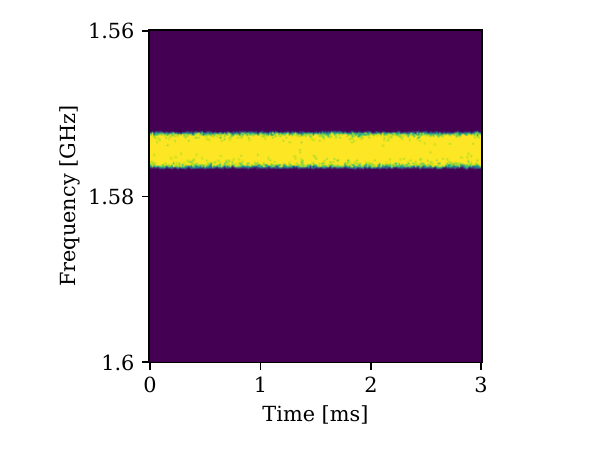}
    \end{minipage}
	\begin{minipage}[t]{0.19\linewidth}
        \centering
    	\includegraphics[trim=26 12 54 10, clip, width=1.0\linewidth]{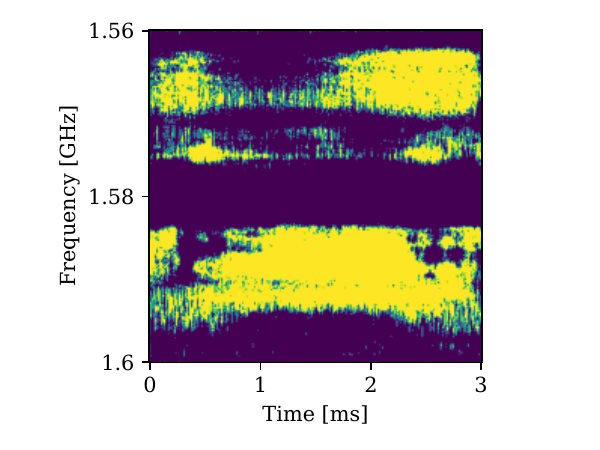}
    \end{minipage}
    \hfill
	\begin{minipage}[t]{0.19\linewidth}
        \centering
    	\includegraphics[trim=26 12 54 10, clip, width=1.0\linewidth]{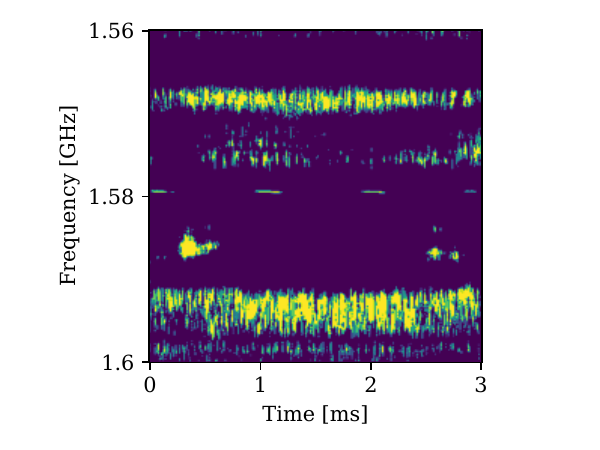}
    \end{minipage}
    \hfill
	\begin{minipage}[t]{0.19\linewidth}
        \centering
    	\includegraphics[trim=26 12 54 10, clip, width=1.0\linewidth]{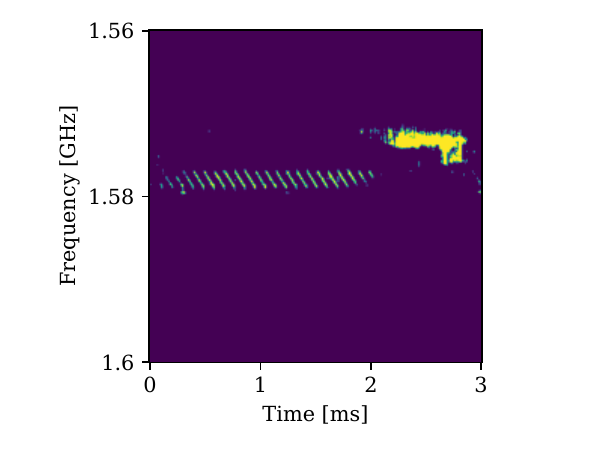}
    \end{minipage}
    \hfill
	\begin{minipage}[t]{0.19\linewidth}
        \centering
    	\includegraphics[trim=26 12 54 10, clip, width=1.0\linewidth]{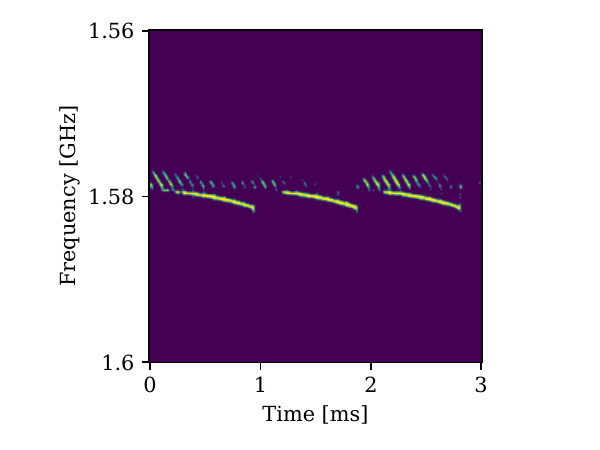}
    \end{minipage}
    \hfill
	\begin{minipage}[t]{0.19\linewidth}
        \centering
    	\includegraphics[trim=26 12 54 10, clip, width=1.0\linewidth]{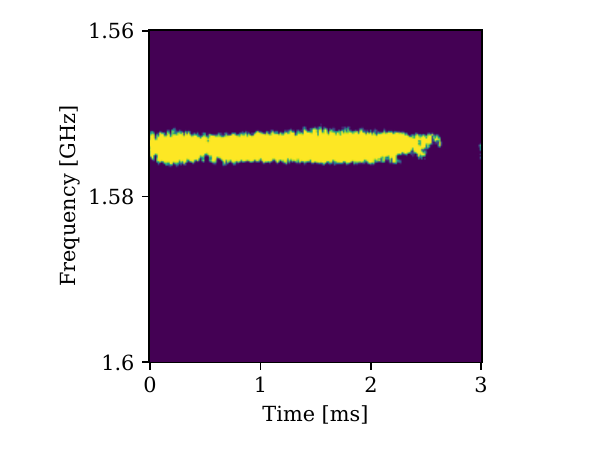}
    \end{minipage}
    \caption{Qualitative source-specific mask predictions for five overlapping interferences; top: ground truth; bottom: prediction.}
    \label{label_app_mask_example_5}
\end{figure}

\begin{figure}[!t]
    \centering
	\begin{minipage}[t]{0.155\linewidth}
        \centering
    	\includegraphics[trim=26 12 54 10, clip, width=1.0\linewidth]{images/mask_prediction/lab_6_0.pdf}
    \end{minipage}
    \hfill
	\begin{minipage}[t]{0.155\linewidth}
        \centering
    	\includegraphics[trim=26 12 54 10, clip, width=1.0\linewidth]{images/mask_prediction/lab_6_1.pdf}
    \end{minipage}
    \hfill
	\begin{minipage}[t]{0.155\linewidth}
        \centering
    	\includegraphics[trim=26 12 54 10, clip, width=1.0\linewidth]{images/mask_prediction/lab_6_2.pdf}
    \end{minipage}
    \hfill
	\begin{minipage}[t]{0.155\linewidth}
        \centering
    	\includegraphics[trim=26 12 54 10, clip, width=1.0\linewidth]{images/mask_prediction/lab_6_3.pdf}
    \end{minipage}
    \hfill
	\begin{minipage}[t]{0.155\linewidth}
        \centering
    	\includegraphics[trim=26 12 54 10, clip, width=1.0\linewidth]{images/mask_prediction/lab_6_4.pdf}
    \end{minipage}
    \hfill
	\begin{minipage}[t]{0.155\linewidth}
        \centering
    	\includegraphics[trim=26 12 54 10, clip, width=1.0\linewidth]{images/mask_prediction/lab_6_5.pdf}
    \end{minipage}
	\begin{minipage}[t]{0.155\linewidth}
        \centering
    	\includegraphics[trim=26 12 54 10, clip, width=1.0\linewidth]{images/mask_prediction/pred_6_3.pdf}
    \end{minipage}
    \hfill
	\begin{minipage}[t]{0.155\linewidth}
        \centering
    	\includegraphics[trim=26 12 54 10, clip, width=1.0\linewidth]{images/mask_prediction/pred_6_2.pdf}
    \end{minipage}
    \hfill
	\begin{minipage}[t]{0.155\linewidth}
        \centering
    	\includegraphics[trim=26 12 54 10, clip, width=1.0\linewidth]{images/mask_prediction/pred_6_5.pdf}
    \end{minipage}
    \hfill
	\begin{minipage}[t]{0.155\linewidth}
        \centering
    	\includegraphics[trim=26 12 54 10, clip, width=1.0\linewidth]{images/mask_prediction/pred_6_1.pdf}
    \end{minipage}
    \hfill
	\begin{minipage}[t]{0.155\linewidth}
        \centering
    	\includegraphics[trim=26 12 54 10, clip, width=1.0\linewidth]{images/mask_prediction/pred_6_4.pdf}
    \end{minipage}
    \hfill
	\begin{minipage}[t]{0.155\linewidth}
        \centering
    	\includegraphics[trim=26 12 54 10, clip, width=1.0\linewidth]{images/mask_prediction/pred_6_0.pdf}
    \end{minipage}
    \caption{Qualitative source-specific mask predictions for six overlapping interferences; top: ground truth; bottom: prediction.}
    \label{label_app_mask_example_6}
\end{figure}

\begin{table}[t]
\centering
\caption{Segmentation performance by number of interference sources corresponding to examples of Figure~\ref{label_iou_evaluation}.}
\label{tab:mask_metrics_by_count}
\begin{tabular}{@{}crrrrr@{}}
\toprule
\textbf{Example} &
\textbf{IoU} &
\textbf{Dice} &
\textbf{Precision} &
\textbf{Recall} &
\textbf{Specificity} \\
\midrule
1 (Fig.~\ref{label_iou_evaluation1}) & 0.890 & 0.942 & 0.969 & 0.916 & 1.000 \\
2 (Fig.~\ref{label_iou_evaluation2}) & 0.884 & 0.938 & 0.954 & 0.923 & 0.985 \\
3 (Fig.~\ref{label_iou_evaluation3}) & 0.852 & 0.920 & 0.906 & 0.935 & 0.985 \\
4 (Fig.~\ref{label_iou_evaluation4}) & 0.795 & 0.885 & 0.869 & 0.903 & 0.994 \\
\bottomrule
\end{tabular}
\end{table}

Figure~\ref{label_iou_evaluation} shows that the predicted masks recover the principal time--frequency structures, ranging from isolated chirps to dense combinations of broadband, narrowband, and nonlinear components. The quantitative results in Table~\ref{tab:mask_metrics_by_count} confirm a gradual degradation with increasing source complexity: the IoU decreases from \(0.890\) to \(0.795\), while the Dice score decreases from \(0.942\) to \(0.885\). Precision declines from \(0.969\) to \(0.869\), whereas recall remains above \(0.90\) for all examples, indicating that increasing complexity primarily introduces additional or enlarged mask regions rather than causing substantial omissions. Specificity remains at least \(0.985\), demonstrating reliable background rejection. Overall, PI-SCAN preserves accurate source-specific segmentation under increasing spectral density, with most discrepancies occurring near mask boundaries and overlapping regions.

\begin{figure*}[!t]
    \centering
	\begin{minipage}[t]{0.19\linewidth}
        \centering
    	\includegraphics[trim=6 6 6 6, clip, width=1.0\linewidth]{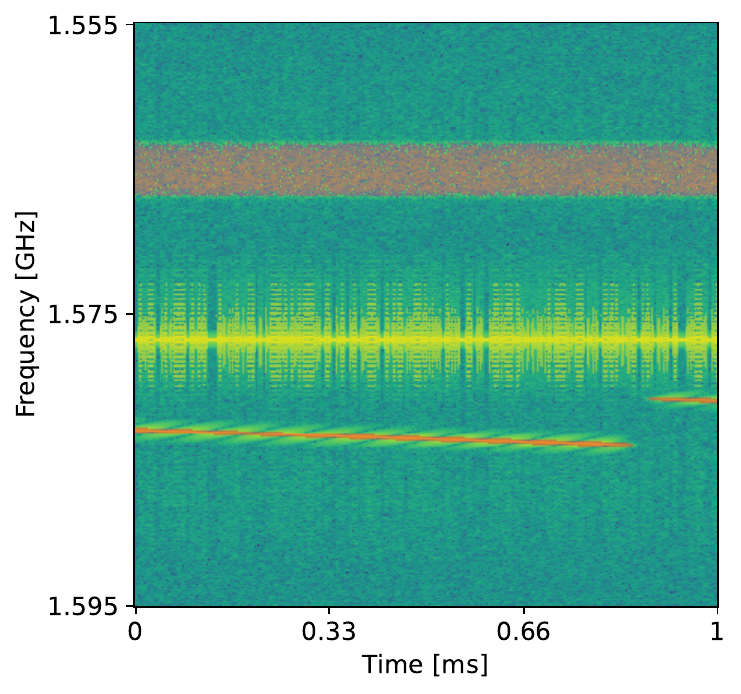}
    \end{minipage}
    \hfill
	\begin{minipage}[t]{0.19\linewidth}
        \centering
    	\includegraphics[trim=6 6 6 6, clip, width=1.0\linewidth]{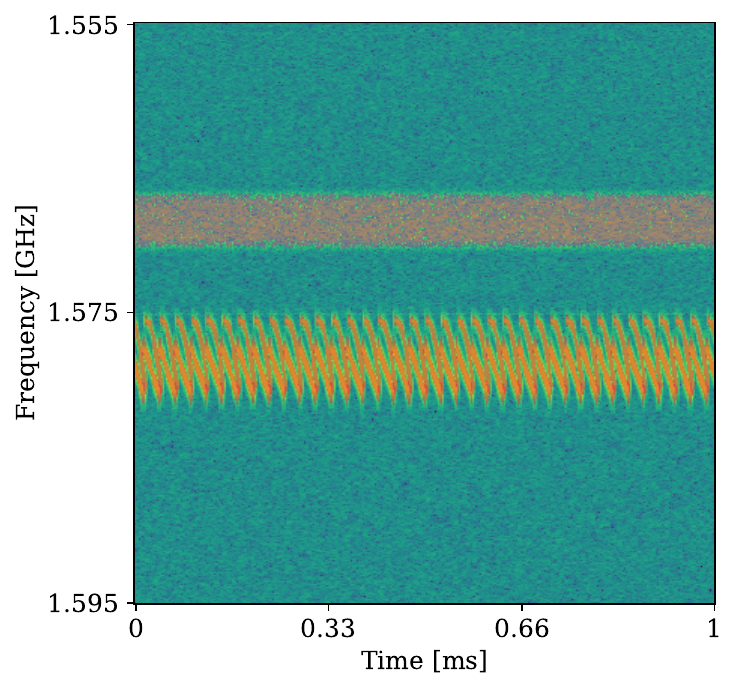}
    \end{minipage}
    \hfill
	\begin{minipage}[t]{0.19\linewidth}
        \centering
    	\includegraphics[trim=6 6 6 6, clip, width=1.0\linewidth]{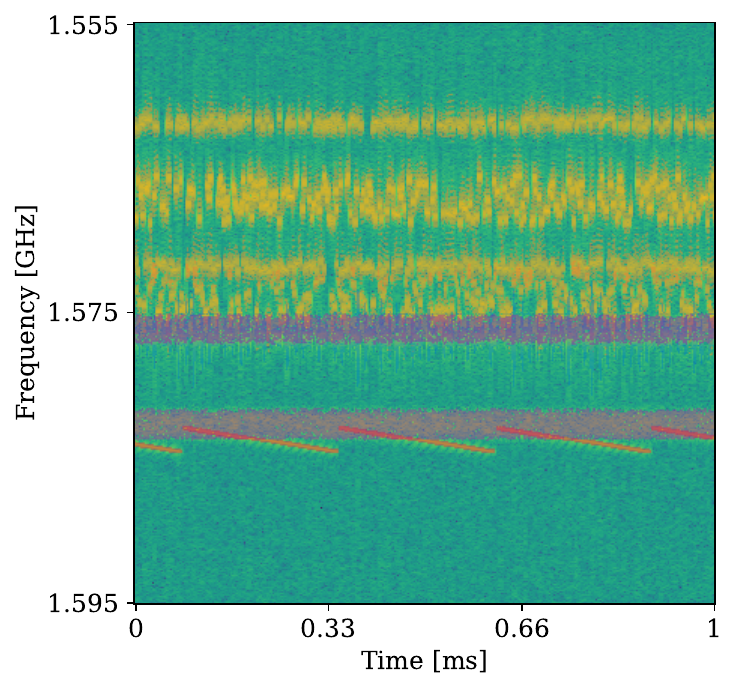}
    \end{minipage}
    \hfill
	\begin{minipage}[t]{0.19\linewidth}
        \centering
    	\includegraphics[trim=6 6 6 6, clip, width=1.0\linewidth]{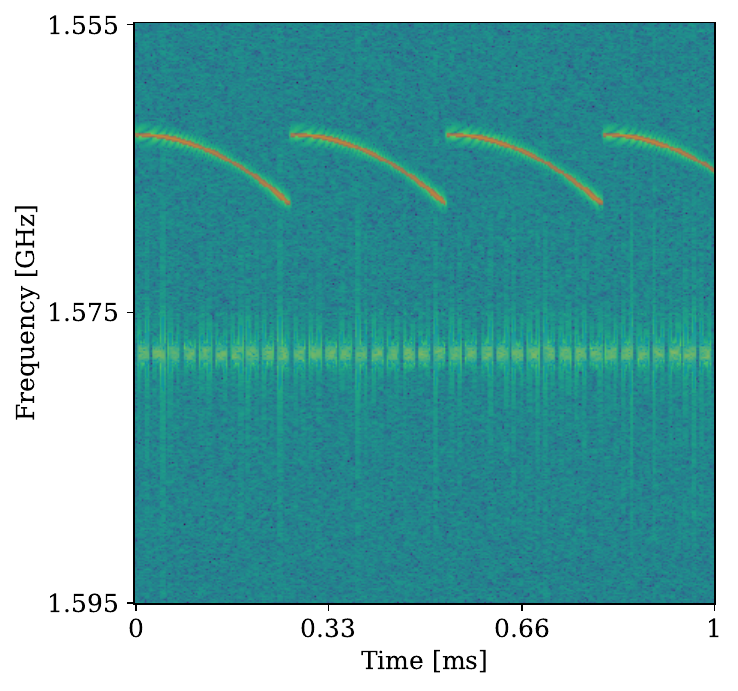}
    \end{minipage}
    \hfill
	\begin{minipage}[t]{0.19\linewidth}
        \centering
    	\includegraphics[trim=6 6 6 6, clip, width=1.0\linewidth]{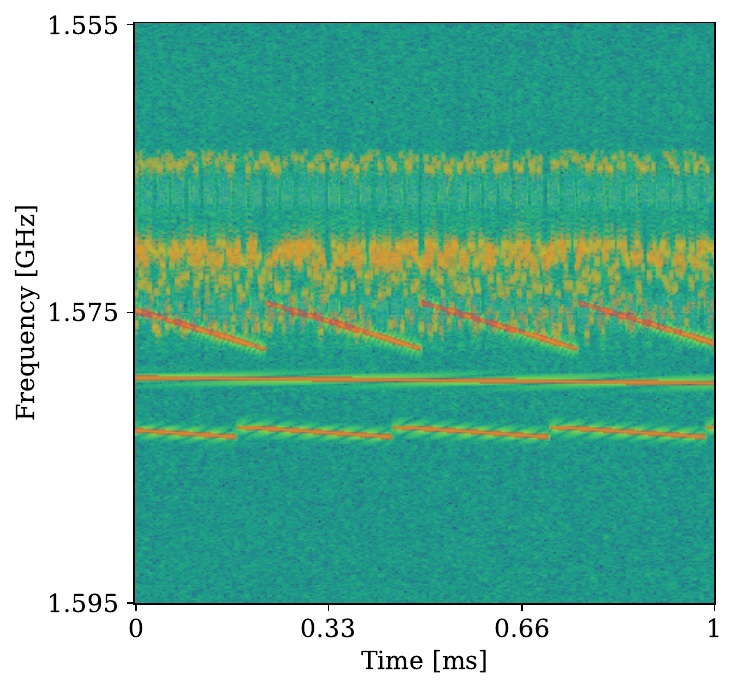}
    \end{minipage}
	\begin{minipage}[t]{0.19\linewidth}
        \centering
    	\includegraphics[trim=6 6 6 6, clip, width=1.0\linewidth]{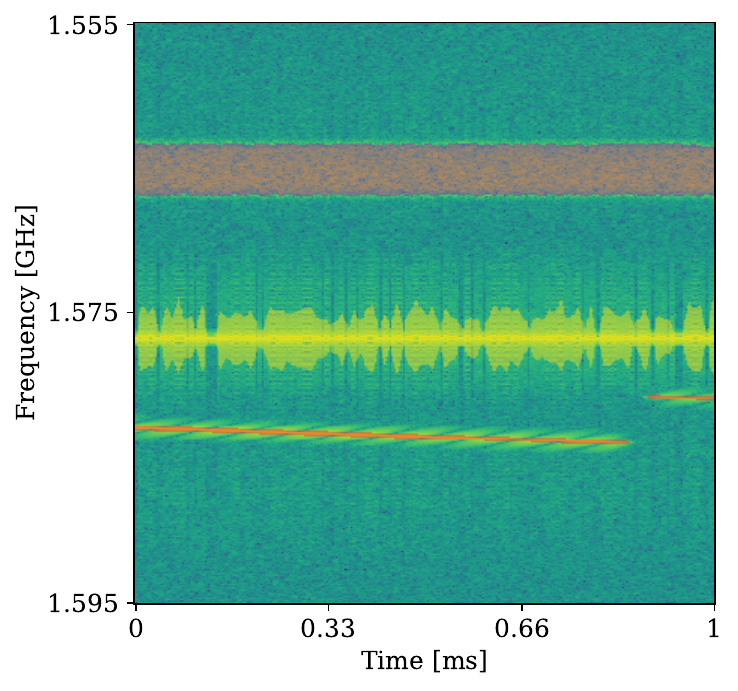}
    \end{minipage}
    \hfill
	\begin{minipage}[t]{0.19\linewidth}
        \centering
    	\includegraphics[trim=6 6 6 6, clip, width=1.0\linewidth]{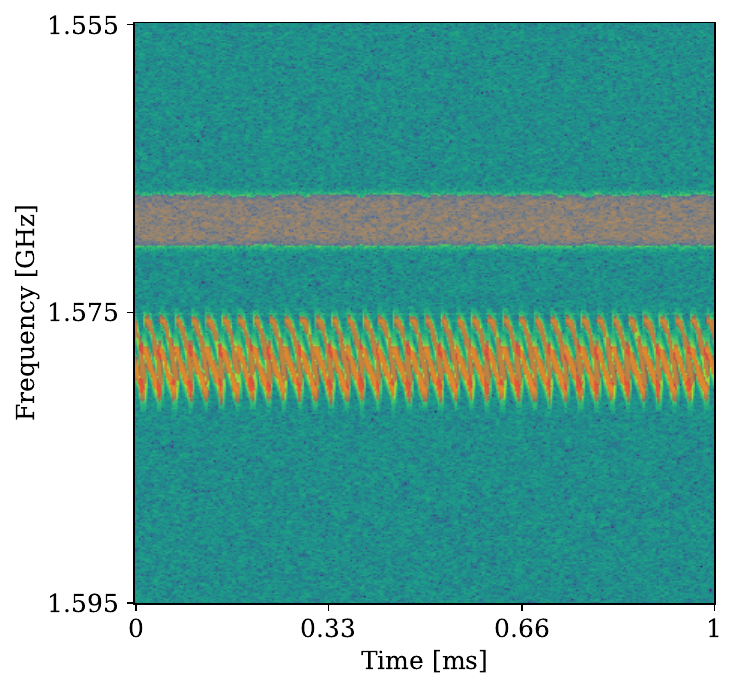}
    \end{minipage}
    \hfill
	\begin{minipage}[t]{0.19\linewidth}
        \centering
    	\includegraphics[trim=6 6 6 6, clip, width=1.0\linewidth]{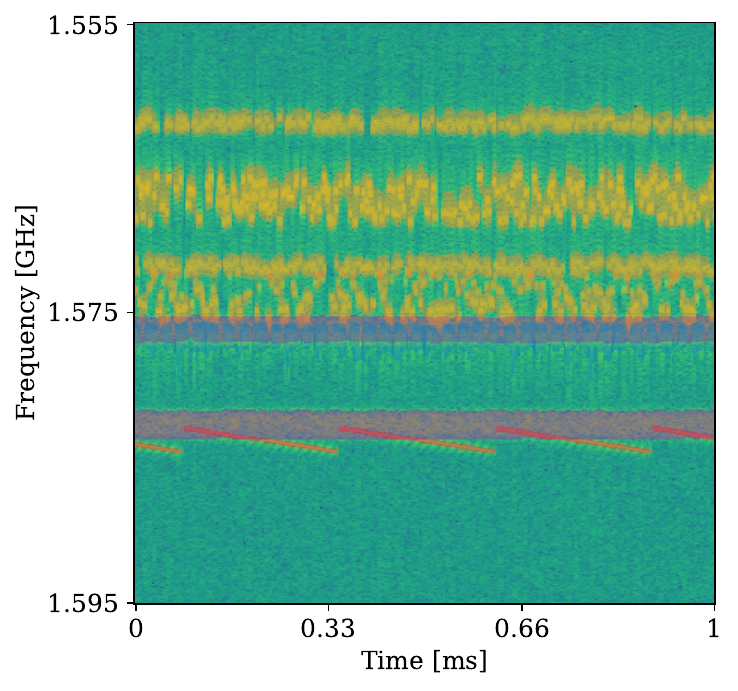}
    \end{minipage}
    \hfill
	\begin{minipage}[t]{0.19\linewidth}
        \centering
    	\includegraphics[trim=6 6 6 6, clip, width=1.0\linewidth]{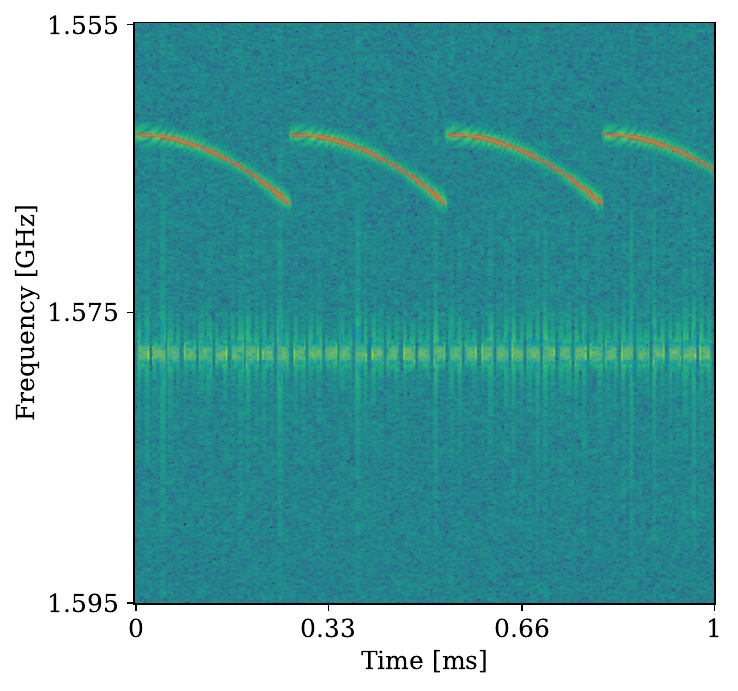}
    \end{minipage}
    \hfill
	\begin{minipage}[t]{0.19\linewidth}
        \centering
    	\includegraphics[trim=6 6 6 6, clip, width=1.0\linewidth]{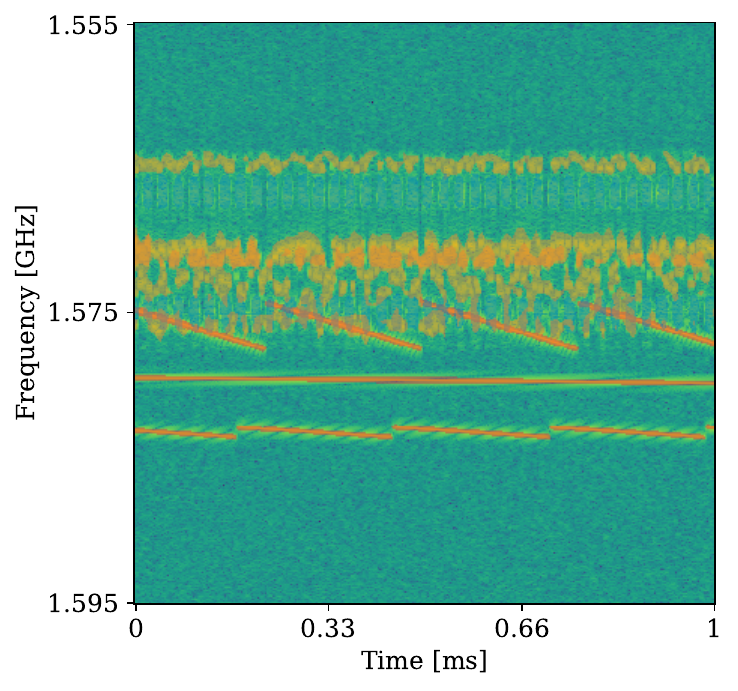}
    \end{minipage}
    \caption{Qualitative RF-DETR instance-segmentation results for five representative multi-interference spectrograms. Each column shows the same sample, with the ground-truth masks in the top row and the corresponding RF-DETR predictions in the bottom row. The colored overlays denote individual interference regions. RF-DETR recovers the principal broadband, narrowband, periodic, and chirped structures, with deviations occurring mainly at mask boundaries and in weak or overlapping regions.}
    \label{label_app_rf_detr_results}
\end{figure*}

To provide a transformer-based segmentation baseline, we trained the
medium-sized RF-DETR model~\cite{robinson_rf_detr} for \(50\) epochs using the
same dataset, comprising scenes with one to six interference sources, and the
same ground-truth masks as PI-SCAN. RF-DETR achieves an F1-score of \(0.912\),
demonstrating strong segmentation performance. The qualitative examples in
Figure~\ref{label_app_rf_detr_results} confirm that the model accurately
recovers the dominant broadband, narrowband, periodic, and chirped structures,
although minor deviations occur around mask boundaries and weak or overlapping
regions. While RF-DETR provides an effective segmentation baseline, it does
not jointly estimate the source-specific physical characteristics predicted by
PI-SCAN.

\subsection{RSSI-Based Evaluation}
\label{label_appendix_rssi_evaluation}

\begin{figure*}[!t]
    \centering
	\begin{minipage}[t]{0.16\linewidth}
        \centering
    	\includegraphics[trim=12 12 10 10, clip, width=1.0\linewidth]{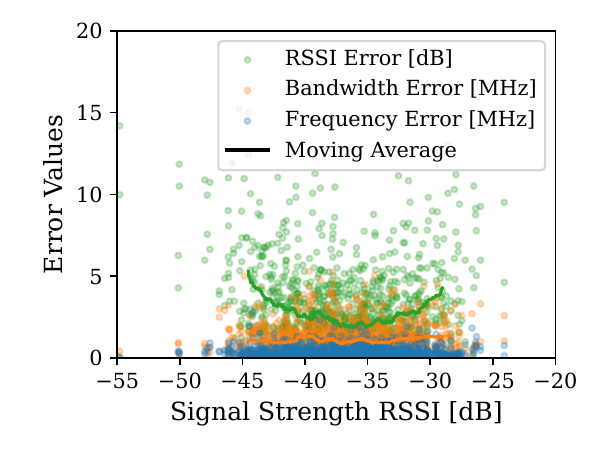}
        \subcaption{One source.}
        \label{label_figure_app_rssi_evaluation1}
    \end{minipage}
    \hfill
	\begin{minipage}[t]{0.16\linewidth}
        \centering
    	\includegraphics[trim=12 12 10 10, clip, width=1.0\linewidth]{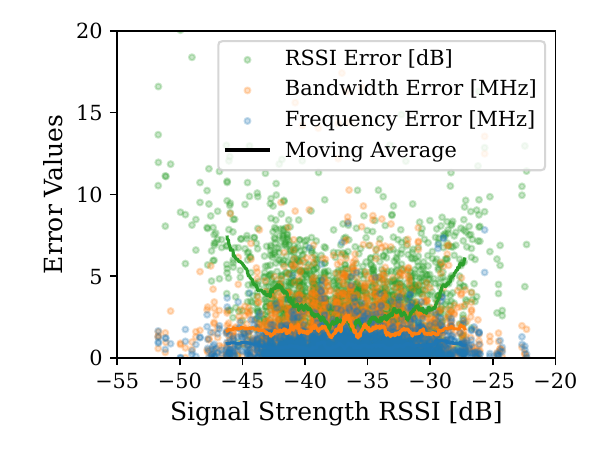}
        \subcaption{Two sources.}
        \label{label_figure_app_rssi_evaluation2}
    \end{minipage}
    \hfill
	\begin{minipage}[t]{0.16\linewidth}
        \centering
    	\includegraphics[trim=12 12 10 10, clip, width=1.0\linewidth]{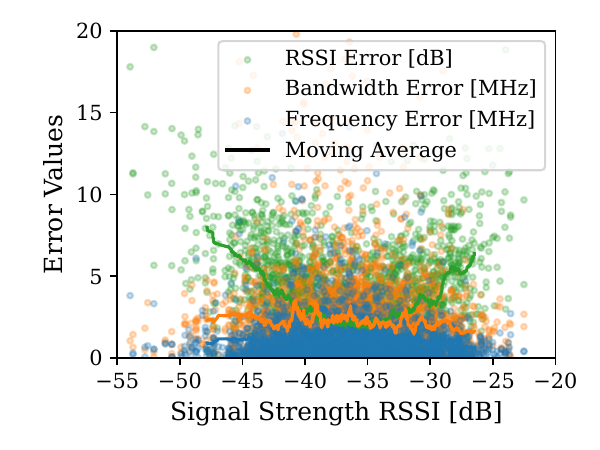}
        \subcaption{Three sources.}
        \label{label_figure_app_rssi_evaluation3}
    \end{minipage}
    \hfill
	\begin{minipage}[t]{0.16\linewidth}
        \centering
    	\includegraphics[trim=12 12 10 10, clip, width=1.0\linewidth]{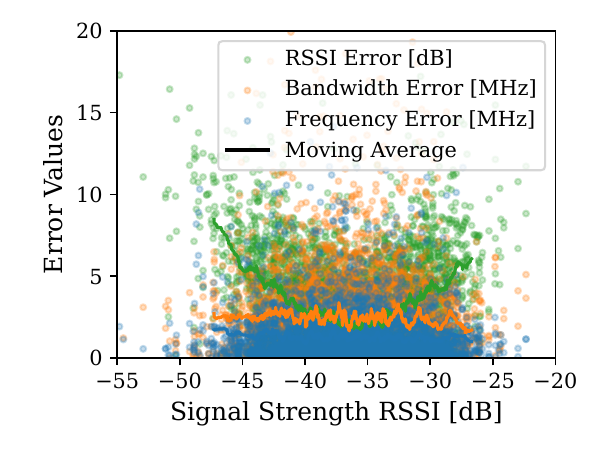}
        \subcaption{Four sources.}
        \label{label_figure_app_rssi_evaluation4}
    \end{minipage}
    \hfill
	\begin{minipage}[t]{0.16\linewidth}
        \centering
    	\includegraphics[trim=12 12 10 10, clip, width=1.0\linewidth]{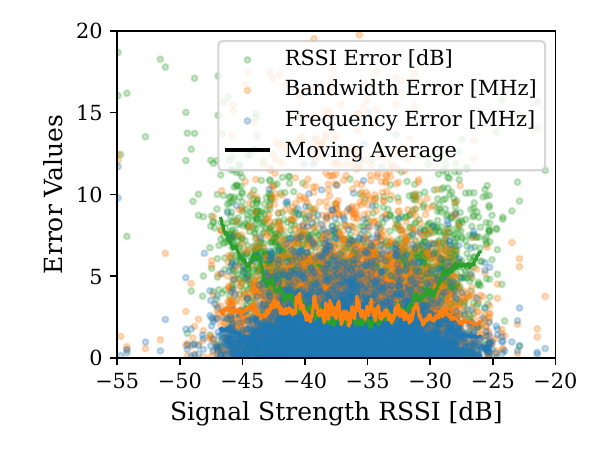}
        \subcaption{Five sources.}
        \label{label_figure_app_rssi_evaluation5}
    \end{minipage}
    \hfill
	\begin{minipage}[t]{0.16\linewidth}
        \centering
    	\includegraphics[trim=12 12 10 10, clip, width=1.0\linewidth]{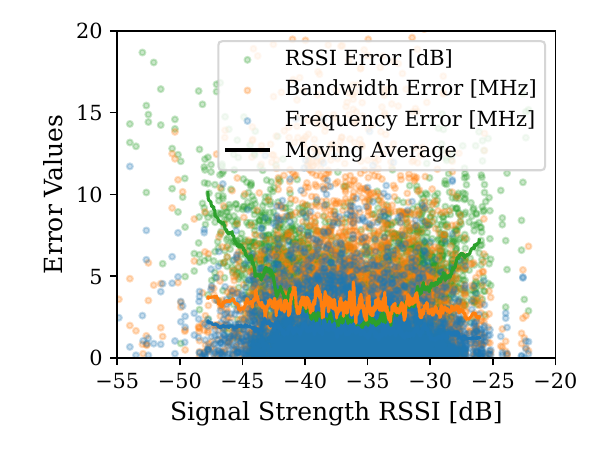}
        \subcaption{Six sources.}
        \label{label_figure_app_rssi_evaluation6}
    \end{minipage}
    \caption{Source-wise characterization errors as a function of the reference RSSI for test scenes containing one to six interference sources. The scatter points show the absolute RSSI error in dB and the absolute bandwidth and center-frequency errors in MHz. The solid curves indicate the corresponding moving averages over the reference RSSI.}
    \label{label_figure_app_rssi_evaluation}
\end{figure*}

Figure~\ref{label_figure_app_rssi_evaluation} shows that center frequency is
estimated with the lowest error, whereas bandwidth estimation exhibits a
larger dispersion that generally increases with the number of simultaneously
active sources. RSSI estimation produces the largest and most widely
distributed errors. Its moving-average error follows a pronounced U-shaped
trend, increasing for weak sources between approximately \(-50\) and
\(-45\,\mathrm{dB}\) and for strong sources between approximately \(-30\) and
\(-25\,\mathrm{dB}\), while remaining lower in the intermediate range. This
behavior indicates regression toward the central RSSI range and can be
attributed to the sample-wise normalization of the input spectrograms, which
largely removes the absolute amplitude scale required for reliable RSSI
estimation. Consequently, the normalized representation retains sufficient
spectral information for frequency and bandwidth estimation but makes absolute
signal-strength recovery intrinsically more difficult.

\subsection{Limitations}
\label{label_appendix_limitations}

The fixed set size \(S=6\) limits the model to at most six individually
resolved interference sources per spectrogram. If more than six sources are
present, additional sources may remain undetected or be merged into the masks
and characteristics of other predictions, while the estimated source count
saturates at six. This limit is a design choice based on the maximum mixture
complexity represented in the dataset rather than a restriction of the
set-based formulation. Larger mixtures can be supported by increasing \(S\),
extending the training data accordingly, and retraining the model, at the
cost of additional prediction-head, mask-decoder, and matching computation.

Furthermore, all models are trained and evaluated using synthetically
generated signals propagated within the considered Sionna ray-tracing
environment. Although this enables controlled variation of source types,
powers, bandwidths, and overlap conditions, the experiments do not establish
generalization to real-world measurements. Differences in receiver hardware,
noise characteristics, propagation conditions, and previously unseen
interference waveforms may cause a simulation-to-real domain gap. Validation
on measured data and adaptation to previously unseen environments therefore
remain important directions for future work.

\vfill

\end{document}